%% file: Revision1-arXiv-CovertCSI.tex
\documentclass[11pt,onecolumn,letterpaper]{IEEEtran}

\IEEEoverridecommandlockouts

\usepackage[cmex10]{amsmath} 
\usepackage{amsfonts,stmaryrd,acronym,amssymb,amsthm,dsfont}
\usepackage{graphicx}
\usepackage{url}
\usepackage{cite}
\usepackage{hyperref}
\usepackage{flushend}
\usepackage{enumerate}
\usepackage{lipsum}
\usepackage{diagbox}
\usepackage[dvipsnames]{xcolor}

\graphicspath{{graphics/}}

\makeatletter
\def\blfootnote{\xdef\@thefnmark{}\@footnotetext}
\makeatother
\allowdisplaybreaks
\allowbreak

\newtheorem{theorem}{Theorem}
\newtheorem{corollary}{Corollary}

\newtheorem{remark}{Remark}

\newtheorem{proposition}{Proposition}

\newtheorem{lemma}{Lemma}[]

\newcommand{\Prob}{\mathbb P}

\newcommand{\mi}{\mathbb I}
\newcommand{\ent}{\mathbb H}
\newcommand{\expec}{\mathbb E}
\newcommand{\kld}{\mathbb D}
\newcommand{\CodeBook}{\mathcal{C}}

\makeatletter%
\if@twocolumn%
\newcommand{\Figwidth}{\columnwidth}%

\def\twocolbreak{\nonumber\\ &}%
\def\twocolalign{&}%
\def\onecolalign{}%
\def\twocolbreakquad{\nonumber\\ &\qquad{}}%
\else
\newcommand{\Figwidth}{4.5in}%

\def\twocolbreak{}%
\def\twocolalign{}%
\def\onecolalign{&}%
\def\twocolbreakquad{}%
\fi%
\makeatother%

\newcommand{\indic}[1]{\ensuremath{\mathds{1}}}
\newcommand{\card}[1]{\ensuremath{\left|{#1}\right|}}   

\newcommand{\eqdef}{\ensuremath{\triangleq}}   
\newcommand{\intseq}[2]{\ensuremath{\llbracket{#1}{\,:\,}{#2}\rrbracket}}

\newcounter{mytempeqcounter}

\newcommand{\D}[3][]{{\mathbb{D}_{#1}}\!\left(#2\,\middle\Vert\,#3\right)} 

\renewcommand{\leq}{\leqslant} 
\renewcommand{\geq}{\geqslant} 

\input{CommandsAndMacros.tex}
\input{Acronyms.tex}

\begin{document}

\title{Keyless Covert Communication via Channel State Information}

\author{Hassan ZivariFard, {\em Student Member, IEEE}, Matthieu R. Bloch, {\em Senior Member, IEEE}, and Aria Nosratinia, {\em Fellow, IEEE}\thanks{H. ZivariFard and A. Nosratinia are with Department of Electrical Engineering, The University of Texas at Dallas, Richardson, TX, USA. M. R. Bloch is with School of Electrical and Computer Engineering, Georgia Institute of Technology, Atlanta, GA, USA. E-mail: hassan@utdallas.edu, matthieu.bloch@ece.gatech.edu, aria@utdallas.edu.}\thanks{The material in this paper was presented in part at IEEE Information Theory Workshop 2019, Visby, Gotland, Sweden.}
\thanks{The work of H. ZivariFard and A. Nosratinia is supported by National Science Foundation (NSF) grant 1956213. The work of M. R. Bloch is supported by National Science Foundation (NSF) grant 1955401.}
}
\maketitle
\date{}

\begin{abstract}
We consider the problem of covert communication over a state-dependent channel when the channel state is available either non-causally, causally, or strictly causally, either at the transmitter alone, or at both transmitter and receiver. Covert communication with respect to an adversary, called ``warden,'' is one in which, despite communication over the channel, the warden's observation remains indistinguishable from an output induced by innocent channel-input symbols. Covert communication involves fooling an adversary in part by a proliferation of codebooks; for reliable decoding at the legitimate receiver the codebook uncertainty is typically removed via a shared secret key that is unavailable to the warden. In contrast to previous work, we do not assume the availability of a shared key at the transmitter and legitimate receiver. Instead, shared randomness is extracted from the channel state in a manner that keeps it secret from the warden, despite the influence of the channel state on the warden's output. When channel state is available at the transmitter and receiver, we derive the covert capacity region. When channel state is only available at the transmitter, we derive inner and outer bounds on the covert capacity. We provide examples for which the covert capacity is positive with knowledge of channel state information but is zero without it.
\end{abstract}

\section{Introduction}
Covert communication refers to scenarios in which reliable communication over a channel must occur while simultaneously ensuring that a separate channel output at a node called the warden has a distribution identical to that induced by an innocent channel symbol~\cite{LPD_on_AWGN,Reliable_Deniable_Comm,LPD_over_DMC,LPD_by_Resolvability,OrderAsymtMehrdad}.
It is known that in a point-to-point \ac{DMC} without state, 
the number of bits that can be reliably and covertly communicated over $n$ channel transmissions scales at most as $O(\sqrt{n})$.\footnote{Except for the special case when the output distribution (at the warden) induced by the innocent symbol is a convex combination of the output distributions generated by the other input symbols \cite{LPD_over_DMC}.} This result has motivated the study of other models in which positive rates are achievable, e.g., when the transmitter and the receiver share on the order of $\sqrt{n}\log(n)$ key bits\cite{LPD_on_AWGN}.
Of particular relevance to the present work, Lee {\em et al.}~\cite{Covert_With_State} considered the problem of covert communication over a state-dependent channel in which the channel state is known either causally or non-causally to the transmitter but unknown to the receiver and the warden. The authors derived the covert capacity when the transmitter and the receiver share a sufficiently long secret key, as well as a lower bound on the minimum secret key length needed to achieve the covert capacity. 
Covert communication over a compound channel in which the object to be masked is the state of the compound channel was studied by Salehkalaibar {\em et al.}~\cite{Covert_Saleh}. Given the presence of a channel state, one may wonder if covert communication with positive rate is possible \emph{without} requiring an external secret key. In particular, several works demonstrated the benefits of exploiting common randomness and channel states to generate secret keys. For instance, the problem of stealth secret key generation from correlated sources was studied by Lin {\em et al.}~\cite{Stealthy_SKG,Stealthy_Keyless_SKG} and covert secret key generation was studied by Tahmasbi and Bloch~\cite{Covert_SKG,Covert_SKG_Quantom}. 

\begin{figure}
\centering
\includegraphics[width=\Figwidth]{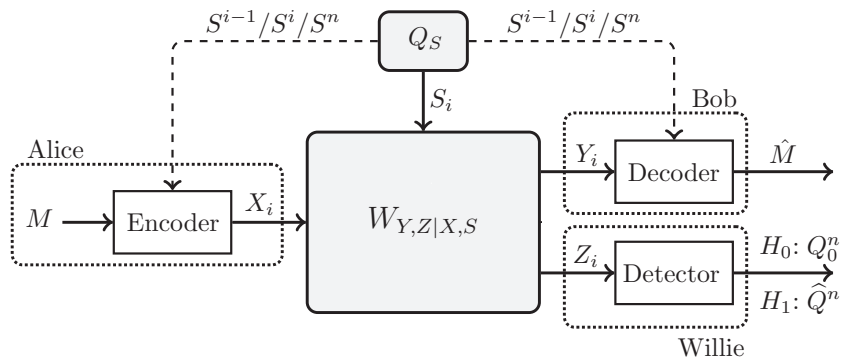}
\caption{Model of covert communication over state-dependent \ac{DMC} when channel state is available at both the transmitter and the receiver}
\label{fig:System_Model}
\end{figure}

The usefulness of exploiting states for secrecy has been extensively investigated in the context of state-dependent wiretap channels. A discrete memoryless  wiretap channel with random states known non-causally at the transmitter was first studied by Chen and Vinck \cite{WTC_With_State}, who established a lower bound on the secrecy capacity based on a combination of wiretap coding with Gel'fand-Pinsker coding. 
Generally speaking, coding schemes with \ac{CSI} outperform those without \ac{CSI} because perfect knowledge of the \ac{CSI} not only enables the transmitter to beamform its signal toward the legitimate receiver but also provides a source of common randomness from which to generate a common secret key and enhance secrecy rates.
Khisti {\em et al.}~\cite{SKA_with_State} studied the problem of secret key generation from non-causal channel state available at the transmitter and established inner and outer bounds on the secret key capacity. Chia and El~Gamal~\cite{Chia_WTC_with_State} studied a wiretap channel in which the state information is available causally at both transmitter and receiver, proposing a scheme in which the transmitter and the receiver extract a weakly secret key from the state and protect the confidential message via a one-time-pad driven with the extracted key (see also \cite{On_Sec_Capa_WTC_CSI} and \cite{WTC_with_SC_State}). Han and Sasaki~\cite{Han_WTC_with_CSI} subsequently extended this result to achieve strong secrecy.
Goldfeld et al.~\cite{Ziv_WTC_with_NonCausaul_CSI} proposed a superposition coding scheme for the problem of transmitting a semantically secure message over a state-dependent channel while the channel state is available non-causally at the transmitter.

This paper studies covert communication over a state-dependent discrete memoryless channel with channel state available either non-causally, causally, or strictly causally, either at both the transmitter and the receiver or at the transmitter alone (see Fig.~\ref{fig:System_Model}). One of the main contributions of this work is to show that the channel state can be used to simultaneously and efficiently accomplish two necessary tasks: using the channel state for a Shannon strategy or Gel'fand-Pinsker coding, while also extracting a shared secret key at the two legitimate terminals to resolve the multiple codebooks that are necessary for covert communication. 
Secret key extraction from channel states replaces the external secret key in other models, thus generalizing and expanding the applicability of covert communication.
A preliminary version of the results of this paper appeared in~\cite{ITWPaper, ISIT2020}.

We characterize the exact covert capacity when \ac{CSI} is available at both the transmitter and the receiver, and derive inner and outer bounds on the covert capacity when channel state information is only available at the transmitter. For some channel models for which the covert capacity is zero without channel state information, we show that the covert capacity is positive with channel state information. The code constructions behind our proofs combine different coding mechanisms, including channel resolvability for covertness, channel randomness extraction for key generation, and Gel'fand-Pinsker coding for state-dependent channels. The key technical challenge consists in properly combining these mechanisms to ensure the overall covertness of the transmission through block-Markov chaining schemes.
In the interest of brevity, the proofs that are more standard, or parallel other proofs in this paper, are omitted from the paper and made available online on arXiv~\cite{KeylessCovertArxiv}, as well as in a supplement provided with the submission.

\section{Preliminaries}
\label{perliminarisec}
Here, $\mathbb{N}$ is the set of natural numbers, which does not include 0, while $\mathbb{R}$ denotes the real numbers. Also, we define $\mathbb{R}_+=\{x\in\mathbb{R}|x\geq 0\}$ and $\mathbb{R}_{++}=\mathbb{R}_+\backslash\{0\}$. Throughout this paper, random variables are denoted by capital letters and their realizations by lower case letters.  The set of $\epsilon-$strongly jointly typical sequences of length $n$, according to $P_{X,Y}$, is denoted by $\mathcal{T}_{\epsilon}^{(n)}({P_{X,Y}})$. For convenience,  typicality will reference the random variables rather than the distribution, e.g., we write  $\mathcal{T}_{\epsilon}^{(n)}(X,Y)$ or $\mathcal{T}_{\epsilon}^{(n)}(X|Y)$.  Superscripts denote the dimension of a vector, e.g., $X^n$. The integer set $\{1,\dots,M\}$ is denoted by $\intseq{1}{M}$, $X_i^j$ indicates the set $\{X_i,X_{i+1},\dots,X_j\}$, and $X_{\sim i}^{n}$ denotes the vector $X^n$ except $X_i$. The cardinality of a set is denoted by $|\cdot|$. The total variation between \ac{PMF} $P$ and \ac{PMF} $Q$ is defined as $||P-Q||_1=\frac{1}{2}\sum_x |P(x)-Q(x)|$ and the Kullback-Leibler (KL) divergence between \ac{PMF}s is defined as $\kld{(P||Q)}=\sum_x p(x)\log\frac{P(x)}{Q(x)}$. The support of a probability distribution $P$ is denoted by supp$(P)$. The $n$-fold product distribution constructed from the same distribution $P$ is denoted $P^{\otimes n}$. Throughout the paper, $\log$ denotes the base 2 logarithm. $\expec_{\backslash i} (\cdot)$ is the expectation with respect to the all random variable in $\mathcal{A}$ except the one with index $i\in\mathcal{A}$, for a set of random variables $\{X_i\}_{i\in \mathcal{A}}$ indexed over a countable set $\mathcal{A}$. Also, $\expec_{X} (\cdot)$ is the expectation w.r.t. the random variable $X$ and $\indic{1}_{\{\cdot\}}$ denotes the indicator function.

Consider a discrete memoryless state-dependent channel as shown in Fig. \ref{fig:System_Model}. 
$\mathcal{X}$ is the channel input, $\mathcal{Y}$ and $\mathcal{Z}$ are the channel outputs at the legitimate receiver and the warden, respectively. Let $x_0\in\mathcal{X}$ be an innocent symbol corresponding to the absence of communication with the receiver. We assume that the channel state is \ac{iid} and drawn according to $Q_S$. Define 
\begin{align}
    Q_0(\cdot) &= \sum_{s\in\mathcal{S}} Q_S(s)W_{Z|X,S}(\cdot|x_0,s),\label{eq:Q0_Definition}\\
    Q_Z(\cdot) &= \sum_{s\in\mathcal{S}}\sum_{x\in\mathcal{X}} Q_S(s)P_{X|S}(x|s)W_{Z|X,S}(\cdot|x,s)\label{eq:Qz_Definition}
\end{align}and let $Q_0^{\otimes n}=\prod\nolimits_{i = 1}^n Q_0$ and also let $Q_Z^{\otimes n}=\prod\nolimits_{i = 1}^n Q_Z$. 
The channel state can be available non-causally, causally, or strictly causally either at both the transmitter and the receiver, or only at the transmitter, but not to the warden. An $(|\mathcal{M}|, n)$ code consists of an encoder-decoder pair, as follows. The encoder maps $(M,S^n) \rightarrow X^n\in \mathcal{X}^n$ when \ac{CSI} is available non-causally at the transmitter, $(M,S^i)\rightarrow X_i\in \mathcal{X}$ when \ac{CSI} is available causally, and $(M,S^{i-1})\rightarrow X_i\in \mathcal{X}$ when the \ac{CSI} is available strictly causally at the transmitter. The decoder maps $(S^n,Y^n)\rightarrow \hat{M}\in\mathcal{M}$ when \ac{CSI} is available at both the transmitter and the receiver, and $Y^n\rightarrow\hat{M}\in\mathcal{M}$ when \ac{CSI} is only available at the transmitter.
The code is assumed known to all parties and the objective is to design a code that is both reliable and covert. The code is defined to be reliable if the probability of error $P_e^{(n)} = \Prob(\hat{M}\ne M)$ goes to zero when $n\to\infty$. The code is covert if the warden cannot determine whether communication is happening (hypothesis $H_1$) or not (hypothesis $H_0$). The probabilities of false alarm (warden deciding $H_1$ when $H_0$ is true) and missed detection (warden deciding $H_0$ when $H_1$ is true), are denoted by $\alpha_n$ and $\beta_n$, respectively. An uninformed, random decision by the warden satisfies $\alpha_n+\beta_n=1$, which is the benchmark for covertness. 
When the channel carries communication, the warden's channel output distribution is denoted $P_{Z^n}$, and the optimal hypothesis test by the warden satisfies $\alpha_n+\beta_n\geq 1-\sqrt{\kld{(P_{Z^n}||Q_0^{\otimes n})}}$ \cite{Hypotheses_Book}. Therefore, to show that the communication is covert, it suffices to show that $\kld{(P_{Z^n}||Q_0^{\otimes n})}\to 0$. Note that supp$(Q_0)=\mathcal{Z}$ otherwise $\kld{(P_{Z^n}||Q_0^{\otimes n})}\to\infty$. Here, we assume $Q_0$ has full support (i.e. does not put zero mass on some point).
Consequently, our goal is to design a sequence of $(2^{nR},n)$ codes such that
\begin{subequations}\label{eq:Code_Definition}
\begin{align}
  &\lim_{n\to\infty}P_e^{(n)}= 0,\label{eq:Error_Probability_General}\\ &\lim_{n\to\infty}\kld{(P_{Z^n}||Q_0^{\otimes n})}= 0.\label{eq:Covertness_General}
\end{align}
\end{subequations}
We define the covert capacity as the supremum of all achievable covert rates and denote it by $C_{\mbox{\scriptsize\rm NC-TR}}$ for non-causal \ac{CSI} available at both the transmitter and the receiver, $C_{\mbox{\scriptsize\rm NC-T}}$ for non-causal \ac{CSI} only available at the transmitter, $C_{\mbox{\scriptsize\rm C-TR}}$ for causal \ac{CSI} available at both the transmitter and the receiver, $C_{\mbox{\scriptsize\rm C-T}}$ for causal \ac{CSI} available only at the transmitter, $C_{\mbox{\scriptsize\rm SC-TR}}$ and $C_{\mbox{\scriptsize\rm SC-T}}$ for strictly causal \ac{CSI} available at both the transmitter and the receiver and strictly causal \ac{CSI} only available at the transmitter, respectively. 
Note that, for the receiver's side availability of the non-causal, causal, and strictly causal \ac{CSI} does not make any difference because decoding is always done after transmission is completed. However, at the transmitter's side availability of the non-causal, causal, and strictly causal \ac{CSI} changes the problem because for these cases the transmitter can use the \ac{CSI} in different manners. Before stating the main results, we state a useful lemma about the relation between the total variation distance and the KL-divergence.
\begin{lemma}[Reverse Pinsker's Inequality {\cite[eq.~(323)]{FDivergence}}]\label{lemma:TV_KLD}
Pinsker's inequality indicates that KL-divergence is larger than total variation distance, i.e. for two arbitrary distributions $P$ and $Q$ on alphabet $\mathcal{A}$ such that $P$ is absolutely continuous with respect to $Q$ we have,
\begin{align}\label{eq:Pinskers_Inequality}
    ||P-Q||_1\leq\sqrt{\frac{1}{2}\kld(P||Q)}.
\end{align}
A reverse inequality is valid when the alphabet $\mathcal{A}$ is finite. Let $P$ and $Q$ be two arbitrary distributions on a finite alphabet set $\mathcal{A}$ such that $P$ is absolutely continuous with respect to $Q$. If $\mu\eqdef\min_{a\in\mathcal{Q}:Q(a)>0}Q(a)$, we have,
\begin{align}\label{eq:Reverse_Pinsker}
    \kld(P||Q)\leq\log\left(\frac{1}{\mu}\right)||P-Q||_1.
\end{align}
\end{lemma}

\section{Channel State Information at the Transmitter and the Receiver}
\subsection{Non-causal Channel State Information}
\begin{theorem}
\label{thm:Capacity_FCSI_NC}
Let
\begin{subequations}\label{eq:finalSD_FCSI_NC}
\begin{align}
\mathcal{S} &\triangleq
  \big\{R\geq0: \exists P_{S,X,Y,Z}\in\calD \text{ such that } R\leq \mi(X;Y|S)\big\}
\label{eq:finalregion_FCSI_NC}
\end{align}
where,
\begin{align}
  \calD \eqdef \left.\begin{cases}P_{S,X,Y,Z}:\\
P_{S,X,Y,Z}=Q_SP_{X|S}W_{Y,Z|S,X}\\
P_Z=Q_0\\
\ent(S|Z) >  \mi(X;Z|S) - \mi(X;Y|S)
\end{cases}\right\}.\label{eq:finalconstraint_FCSI_NC}
\end{align}
\end{subequations}
The covert capacity of the \ac{DMC} $W_{Y,Z|S,X}$ with non-causal \ac{CSI} at both the transmitter and the receiver is
\begin{align}
C_{\mbox{\scriptsize\rm NC-TR}} =\max\{x:x\in\mathcal{S}\}.
\label{eq:Capacity_FCSI_NC}
\end{align}
\end{theorem}

\begin{remark}[Interpretation]
Theorem~\ref{thm:Capacity_FCSI_NC} suggests that the key rate that one can extract from channel state \big(i.e. $H(S|Z)$\big) should be at least as large as the difference between the capacity of the warden and the capacity of the legitimate receiver.
\end{remark}
The achievability is proved by superposition encoding. 
The proof is available in Appendix~\ref{sec:Proof_thm_FCSI_NC}.

\subsection{Causal Channel State Information}

\begin{theorem}
\label{thm:Capacity_FCSI_C}
Let
\begin{subequations}\label{eq:finalSD_FCSI_C}
\begin{align}
\mathcal{S} &\triangleq
  \big\{R\geq0: \exists P_{S,U,X,Y,Z}\in\calD \text{ such that } R\leq \mi(U;Y|S)\big\}
\label{eq:finalregion_FCSI_C}
\end{align}
where,
\begin{align}
  \calD \eqdef \left.\begin{cases}P_{S,U,X,Y,Z}:\\
P_{S,U,X,Y,Z}=Q_SP_U\indic{1}_{\big\{X=X(U,S)\big\}}W_{Y,Z|S,X}\\
P_Z=Q_0\\
\ent(S|Z) >  \mi(U;Z|S) - \mi(U;Y|S)\\
\card{\calU}\leq \card{\calX}+1
\end{cases}\right\}.\label{eq:finalconstraint_FCSI_C}
\end{align}
\end{subequations}
The covert capacity of the \ac{DMC} $W_{Y,Z|S,X}$ with causal \ac{CSI} at both the transmitter and the receiver is 
\begin{align}
C_{\mbox{\scriptsize\rm C-TR}} =\max\{x:x\in\mathcal{S}\}.
\label{eq:Capacity_FCSI_C}
\end{align}
\end{theorem}

\begin{remark}[Interpretation]
Theorem~\ref{thm:Capacity_FCSI_C} suggests that the key rate that we extract from channel state \big(i.e. $H(S|Z)$\big) should be at least as large as the difference between the capacity of the warden and the capacity of the legitimate receiver.
\end{remark}

The proof is based on block Markov encoding and Shannon strategy for transmitting the message according to \ac{CSI}. In order to use the \ac{CSI} simultaneously for both Shannon strategy and key generation we design a block Markov encoding scheme. In this scheme, the transmitter not only generates a key from $S^n$ but also selects its codeword according to $S^n$ to help transmitting the message. The proof is available in Appendix~\ref{sec:Proof_Acivable_Rate_CoNC_State}.

\subsection{Strictly Causal Channel State Information}
We now consider the case in which strictly causal channel state information is available at both the transmitter and the receiver. Even though {\em strictly} causal channel state information is more restricted, it is still useful because the shared randomness can still be used to mislead the warden. 
\begin{theorem}
\label{thm:Capacity_FCSI_SC}
Let
\begin{subequations}\label{eq:finalSD_FCSI_SC}
\begin{align}
\mathcal{S} &\triangleq
  \big\{R\geq0: \exists P_{S,X,Y,Z}\in\calD \text{ such that } R\leq \mi(X;Y|S)\big\}
\label{eq:finalregion_FCSI_SC}
\end{align}
where,
\begin{align}
  \calD \eqdef \left.\begin{cases}P_{S,X,Y,Z}:\\
P_{S,X,Y,Z}=Q_SP_XW_{Y,Z|S,X}\\
P_Z=Q_0\\
\ent(S|Z) >  \mi(X;Z|S) - \mi(X;Y|S)
\end{cases}\right\}.\label{eq:finalconstraint_FCSI_SC}
\end{align}
\end{subequations}
The covert capacity of the \ac{DMC} $W_{Y,Z|S,X}$ with strictly causal \ac{CSI} at both the transmitter and the receiver is
\begin{align}
C_{\mbox{\scriptsize\rm SC-TR}} =\max\{x:x\in\mathcal{S}\}.
\label{eq:Capacity_FCSI_SC}
\end{align}
\end{theorem}
To prove the achievability of Theorem~\ref{thm:Capacity_FCSI_SC}, we merely use the channel state information for key generation and not for data transmission by using a block Markov encoding scheme. The proof is available in Appendix~\ref{sec:Proof_Acivable_Rate_Strictly_Causal_State_Info}. 
\section{Examples}
We provide two examples of covert communication without an external secret key, while surpassing the square-root-$n$ performance with a positive rate.
The two examples explore additive and multiplicative channel states, respectively. The former represents channels in which state can in principle be cancelled, while the latter represents fading-like states.

\subsection{Binary Additive State}
\begin{figure}
\centering
\includegraphics[width=3.8in]{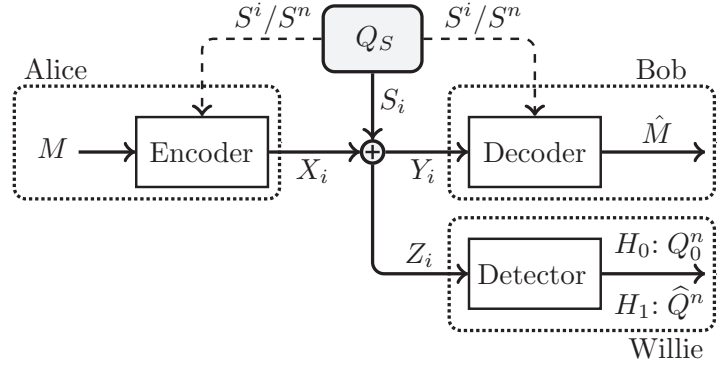}
\caption{Binary symmetric channel with additive \ac{CSI} at the transmitter and the receiver}
\label{fig:BSC_Additive_State_CSITR}
\end{figure}
Consider a channel in which $X, Y, Z,$ and $S$ are all
binary, $Q_s$ obeys a Bernoulli distribution with parameter
$\zeta\in(0, 0.5)$, and the innocent symbol is $x_0=0$. (See Fig.~\ref{fig:BSC_Additive_State_CSITR}). The law of the channel is
\begin{align}
    Y=Z=X\oplus S.\label{eq:BSC_Additive_State_Channel_Law}
\end{align}
\begin{proposition}
\label{thm:Capacity_C_FCSI_Binary_Additive}
The covert capacity of the \ac{DMC} depicted in Fig.~\ref{fig:BSC_Additive_State_CSITR} with causal or non-causal \ac{CSI} available at the transmitter and the receiver is
\begin{align}
C_{\mbox{\scriptsize\rm NC-TR}} = C_{\mbox{\scriptsize\rm C-TR}} = & \ent_b(\zeta)=\zeta\log{\frac{1}{\zeta}}+(1-\zeta)\log{\frac{1}{1-\zeta}},
\label{eq:Capacity_Region_CoNC_State_Binary_Additive}
\end{align}if $\ent(S|Z) > 0$.
\end{proposition}
Intuitively, because the encoder knows the \ac{CSI}, it can perfectly control the warden's observation. By manipulating $X$ the encoder can ensure $Z$ follows the statistics of $S$. In part, this means the set of symbols $X=1$ correspond half to $S=0$ and half to $S=1$, ensuring $Q_S\sim\mbox{Bern}(\zeta)= P_Z$. Further, since the transmitter and receiver share the \ac{CSI}, the legitimate channel is error-free. 
\begin{proof}
First we proof Proposition~\ref{thm:Capacity_C_FCSI_Binary_Additive} when channel state is available non-causally at both the transmitter and the receiver. Substituting $Y=Z=X\oplus S$ in Theorem~\ref{thm:Capacity_FCSI_NC} results in
\begin{align}
C_{\mbox{\scriptsize\rm NC-TR}} = & \mathop {\max }\limits_{Q_SP_{X|S}}\ent(X|S),
\label{eq:Capacity_BAS_CSITR}
\end{align}such that $P_Z =Q_0$. Since $x_0=0$ we have $Q_0= Q_S$.
\begin{table}
  \caption{Joint probability distribution of $X,S$}
  \centering
\setlength\arrayrulewidth{1pt}
  \begin{tabular}{c|cc}
  \diagbox[linewidth=1pt]{$S$\kern-0.7em}{\kern-0.7em $X$} & 0 & 1\\
  \hline
  0 & $\alpha$ & $\beta$\\
  1 & $1-\alpha-\beta-\eta$ & $\eta$
  \end{tabular}
  \label{table:Joint_Dist_X_S}
\end{table}
Let the joint distribution between $X$ and $S$ be according to Table~\ref{table:Joint_Dist_X_S}, we have
\begin{align}
    P_Z(z=0)&=P_{X,S}(x=0,s=0)+P_{X,S}(x=1,s=1)= \alpha+\eta,\label{eq:QZ_equal_to_Zero}\\
    Q_S(s=0)&=P_{X,S}(x=0,s=0)+P_{X,S}(x=1,s=0) = \alpha+\beta.\label{eq:QS_equal_to_Zero}
\end{align}Therefore $P_Z = Q_S$ implies that
\begin{align}
    Q_Z(z=0)=Q_S(s=0)\Rightarrow \alpha+\eta = \alpha+\beta\Rightarrow\eta=\beta\label{eq:equality_QZ_QS}.
\end{align}Therefore,
\begin{align}
    \mathop {\max }\limits_{Q_SP_{X|S}}\ent(X|S)&= \max\limits_{\substack{(\alpha,\beta) \\ \zeta=\alpha+\beta}} \Big[-\alpha\log\frac{\alpha}{\alpha+\beta}-(1-\alpha-2\beta)\log\frac{1-\alpha-2\beta}{1-\alpha-\beta}\nonumber\\
    &\qquad-\beta\log\frac{\beta}{\alpha+\beta}-\beta\log\frac{\beta}{1-\alpha-\beta}\Big].\label{eq:MaxEntropy_First}
\end{align}Considering $Q_S(s=0)=\zeta=\alpha+\beta$ and substituting $\beta=\zeta-\alpha$ in \eqref{eq:MaxEntropy_First} results in
\begin{align}
    \mathop {\max }\limits_{Q_SP_{X|S}}\ent(X|S) &= \mathop {\max }\limits_{\alpha}\Big[-\alpha\log\frac{\alpha}{\zeta}-(1+\alpha-2\zeta)\log\frac{1+\alpha-2\zeta}{1-\zeta}\nonumber\\
    &\qquad-(\zeta-\alpha)\log\frac{\zeta-\alpha}{\zeta}-(\zeta-\alpha)\log\frac{\zeta-\alpha}{1-\zeta}\Big].\label{eq:MaxEntropy_Second}
\end{align} Since entropy is a continuous concave function, the maximizer of $\ent(X|S)$ is found at the root of the first derivative of \eqref{eq:MaxEntropy_Second}. This root is $\alpha=\zeta^2$, resulting in $\max \ent(X|S)=\ent(S)$. Since $Y=Z$, the condition $\ent(S|Z) >  \mi(X;Z|S) - \mi(X;Y|S)$ reduces to $\ent(S|Z) > 0$. If $\ent(S|Z) = 0$ the rate of covert communication between the legitimate terminals will be zero because the warden can recover $S$ from $Z$ and the problem reduces to covert communication for point to point channel.

We now proof Proposition~\ref{thm:Capacity_C_FCSI_Binary_Additive} when channel state is available causally at both the transmitter and the receiver. We assume that $U$ is a Bernoulli random variable with parameter $\eta\in(0, 0.5)$. Here, $U$ and $S$ are independent. To prove the achievability for the causal case, we choose $X=U\oplus S$ therefore $Y=Z=U$ and $\mi(U;Y|S)=\ent(U)$. Since $x_0=0$ we have $Q_0=Q_S$ and the condition $Q_Z=Q_0$ results in $\eta=\zeta$ this follows because
\begin{align}
    Q_S(z=0)&=\Prob(s=0)= \zeta,\label{eq:QZ_equal_to_Zero_C}\\
    Q_Z(z=0)&=\Prob(u=0) = \eta.\label{eq:QS_equal_to_Zero_C}
\end{align}Therefore the covert capacity is equal to $\ent_b(\zeta)$. Again since $Y=Z$, the condition $\ent(S|Z) >  \mi(U;Z|S) - \mi(U;Y|S)$ reduces to $\ent(S|Z) > 0$. If $\ent(S|Z) = 0$ the rate of covert communication between the legitimate terminals will be zero because the warden can recover $S$ from $Z$ and the problem reduces to covert communication for point to point channel. The proof for the upper bound comes from the fact that the covert communication rate for the causal case cannot be higher than the non-causal case.
\end{proof}
\subsection{Binary Multiplicative State}
\begin{figure*}
\centering
\includegraphics[width=4.0in]{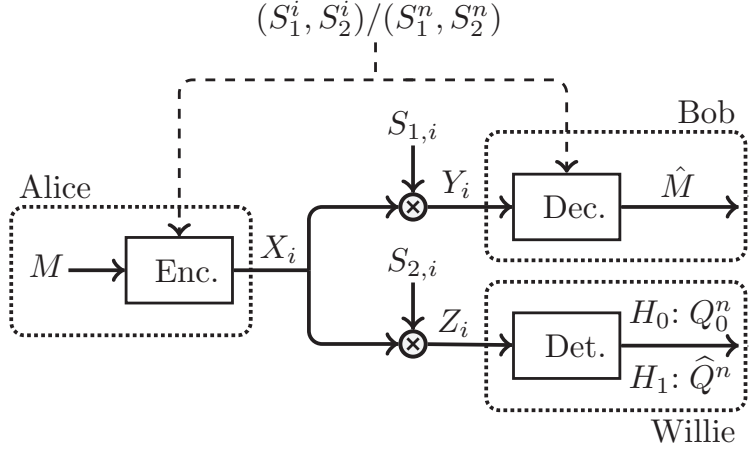}
\caption{Binary symmetric channel with multiplicative \ac{CSI} at the transmitter and the receiver}
\label{fig:BSC_Multiplicative_State_CSITR}
\end{figure*}
Consider a channel in which $X, Y, Z, S_1$, and $S_2$ are all
binary and $S_1$ and $S_2$ have a joint distribution with parameters $P(S_1=i,S_2=j)=p_{i,j}$, for $i,j\in\{0,1\}$ and the innocent symbol is $x_0=0$ (See Fig.~\ref{fig:BSC_Multiplicative_State_CSITR}). The law of the channel is
\begin{align}
    Y=X\otimes S_1,\qquad Z=X\otimes S_2.\label{eq:BSC_Multiplicative_State_Channel_Law_CSITR}
\end{align}
\begin{proposition}
\label{thm:Capacity_C_FCSI_Binary_Multiplicative}
The covert capacity of the \ac{DMC} depicted in Fig.~\ref{fig:BSC_Multiplicative_State_CSITR} with causal or non-causal \ac{CSI} available at the transmitter and the receiver is
\begin{align}
C_{\mbox{\scriptsize\rm NC-TR}} = C_{\mbox{\scriptsize\rm C-TR}} = & p_{1,0}.
\label{eq:Capacity_Region_CoNC_State_Binary_Mul}
\end{align}
\end{proposition}
Intuitively, covert communication occurs when the warden's observation is blocked by a bad channel state while simultaneously the legitimate receiver is enjoying a good channel state. Since the receiver has access to \ac{CSI}, the legitimate channel is effectively noise-free.
\begin{proof}
We proof Proposition~\ref{thm:Capacity_C_FCSI_Binary_Multiplicative} for the non-causal case, the proof for the causal case is similar to the non-causal case and it's omitted for brevity. In Theorem~\ref{thm:Capacity_FCSI_NC}, by substituting $S=(S_1,S_2)$, 
\begin{align}
    C_{\mbox{\scriptsize\rm NC-TR}}  &=\max\limits_{\substack{P_{X|S_1,S_2}, \\ Q_Z=Q_0}}\big[\mi(X;Y|S_1,S_2)\big]\nonumber\\
    &=\max\limits_{\substack{P_{X|S_1,S_2}, \\ Q_Z=Q_0}}\Big[\sum\limits_{i=0}^1\sum\limits_{j=0}^1p_{i,j}\mi(X;Y|S_1=i,S_2=j)\Big]\nonumber\\
    &\mathop=\limits^{(a)}\max\limits_{\substack{P_{X|S_1,S_2}}}\Big[p_{1,0}\mi(X;Y|s_1=1,s_2=0)\Big]\nonumber\\
    &=\mathop {\max }\limits_{P_{X|S_1,S_2}}\Big[p_{1,0}\ent(X|s_1=1,s_2=0)\Big]\nonumber\\
    &=p_{1,0},
\end{align}
where $(a)$ holds because when $S_1=0$, communication is prohibited because $Y$ is annihilated, and when $S_2=1$, we have $Z=X$ therefore $P_X=P_Z=Q_0$, and once again there can be no communication. 
Also, $Z$ always equal to zero implies that $\mi(X;Z|S)\leq\ent(Z)=0$ and therefore the condition $\ent(S|Z) >  \mi(X;Z|S) - \mi(X;Y|S)$ is satisfied.
\end{proof}

\section{Channel State Information Only at the Transmitter}

Fig.~\ref{fig:SystemModel_CSIT} depicts covert communication with channel state available at the transmitter (but not the receiver). We proceed to discuss the effects of non-causal, causal, and strictly causal \ac{CSI} at the transmitter.

\subsection{Non-Causal \ac{CSI}}

\begin{figure*}
\centering
\includegraphics[width=\Figwidth]{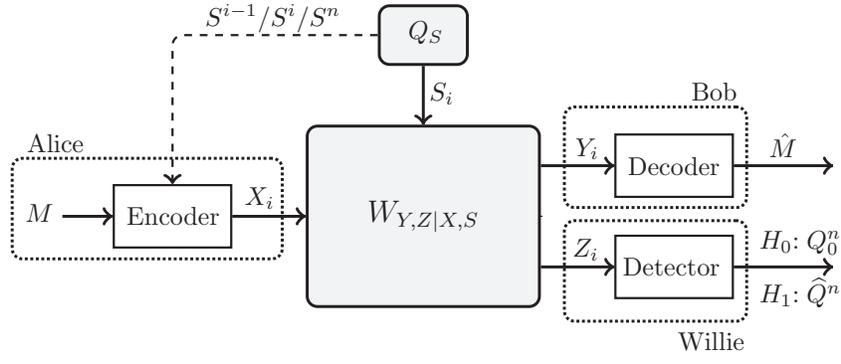}
\caption{Model of covert communication over state-dependent \ac{DMC} when \ac{CSI} is only available at the transmitter}
\label{fig:SystemModel_CSIT}
\end{figure*}

\begin{theorem}
\label{thm:Acivable_Rate_CSIT_WZ_NC}
Let
\begin{subequations}\label{eq:finalSD_CSIT_NC_Inner}
\begin{align}
\mathcal{S} &\triangleq
  \big\{R\geq0: \exists P_{U,V,S,X,Y,Z}\in\calD \text{ such that } R\leq \mi(U;Y)-\max\big\{\mi(U;S),\mi(U,V;S)-\mi(V;Y|U)\big\}\big\}
\label{eq:finalregion_CSIT_NC_Inner}
\end{align}
where,
\begin{align}
  \calD \eqdef \left.\begin{cases}P_{U,V,S,X,Y,Z}:\\
P_{U,V,S,X,Y,Z}=P_UP_VP_{S|U,V}\indic{1}_{\big\{X=X(U,S)\big\}}W_{Y,Z|X,S}\\
Q_S(\cdot)=\sum\nolimits_{u\in\mathcal{U}}\sum\nolimits_{v\in\mathcal{V}}P_U(u)P_V(v)P_{S|U,V}(\cdot|u,v)\\
P_Z=Q_0\\
\mi(V;Y|U) >  \max\{\mi(V;Z),\mi(U,V;Z)-\mi(U;Y)\}\\
\card{\calU}\leq \card{\calX}+5\\
\card{\calV}\leq \card{\calX}+3
\end{cases}\right\}.\label{eq:finalconstraint_CSIT_NC_Inner}
\end{align}
\end{subequations}
The covert capacity of the \ac{DMC} $W_{Y,Z|S,X}$ with non-causal \ac{CSI} at the transmitter is lower-bounded as
\begin{align}
C_{\mbox{\scriptsize\rm NC-T}} \geq\max\{x:x\in\mathcal{S}\}.
\label{eq:Inner_Bound_For_CSIT_NC}
\end{align}
\end{theorem}

The proof is based on block Markov encoding, Gel'fand-Pinsker coding for transmitting the message according to \ac{CSI} \cite{GelfandPinsker}, and Wyner-Ziv coding for secret key generation \cite{WynerZiv76}. To use the \ac{CSI} for both Gel'fand-Pinsker encoding and key generation we design a block Markov encoding scheme. In this scheme, the transmitter not only generates a key from $S^n$ but also selects its codeword according to $S^n$ by using a likelihood encoder \cite{Cuff13,Yassaee13,Watanabe15}. 
Here, instead of generating a secret key directly from the \ac{CSI} we generate a secret key from another random variable which is correlated with \ac{CSI}. This helps to control the rate of the secret key that is generated from the \ac{CSI}, because sometimes a shared key between the legitimate terminals is not needed for covert communication. For example, consider the case where the legitimate receiver's channel is a less noisy version of the warden's channel (see~Corollary~\ref{thm:Capacity_Region_For_Degraded_CSIT_NC}). Also, this arrangement helps the legitimate terminals to generate the secret key from the \ac{CSI} more efficiently (see~Section~\ref{Sec:Example_Reverse_Degraded}). To generate a secret key from a random variable which is correlated with \ac{CSI} we use Wyner-Ziv encoding. Proof details are available in Appendix~\ref{sec:Proof_Acivable_Rate_CSIT_WZ_NC}.

\begin{theorem}
\label{thm:Simple_Acivable_Region_CSIT_NC}
Let
\begin{subequations}\label{eq:finalSD_CSIT_NC_Simple_Inner}
\begin{align}
\mathcal{S} &\triangleq
  \big\{R\geq0: \exists P_{S,U,X,Y,Z}\in\calD \text{ such that } R\leq \mi(U;Y)-\mi(U;S)\big\}
\label{eq:finalregion_CSIT_NC_Simple_Inner}
\end{align}
where,
\begin{align}
  \calD \eqdef \left.\begin{cases}P_{S,U,X,Y,Z}:\\
P_{S,U,X,Y,Z}=Q_SP_{U|S}\indic{1}_{\big\{X=X(U,S)\big\}}W_{Y,Z|X,S}\\
P_Z=Q_0\\
\mi(U;Y) > \mi(U;Z)\\
\card{\calU}\leq \card{\calX}+2
\end{cases}\right\}.\label{eq:finalconstraint_CSIT_NC_Simple_Inner}
\end{align}
\end{subequations}
The covert capacity of the \ac{DMC} $W_{Y,Z|S,X}$ with non-causal \ac{CSI} at the transmitter is lower-bounded as
\begin{align}
C_{\mbox{\scriptsize\rm NC-T}} \geq\max\{x:x\in\mathcal{S}\}.
\label{eq:Simple_Inner_Bound_For_CSIT_NC}
\end{align}
\end{theorem}

The proof uses Gel'fand-Pinsker coding with likelihood encoder, {\em without}  block-Markov encoding or key generation from \ac{CSI}. Theorem~\ref{thm:Simple_Acivable_Region_CSIT_NC} can also be proved by using Theorem~\ref{thm:Acivable_Rate_CSIT_WZ_NC} while generating $S$ independent of $V$ (i.e. $P_{S|U,V}=P_{S|U}$). The independence between $V$ and $S$ implies $\mi(V;S)=0$ and $\mi(V;U,Y)=0$ which results in the region in Theorem~\ref{thm:Simple_Acivable_Region_CSIT_NC}.
Proof details are available in Appendix~\ref{sec:Proof_Simple_Inner_Bound_For_CSIT_NC}.
\begin{remark}[Comparison Between Theorems~\ref{thm:Acivable_Rate_CSIT_WZ_NC} and \ref{thm:Simple_Acivable_Region_CSIT_NC}]
\label{remark:2LowerBounds}
Theorems~\ref{thm:Acivable_Rate_CSIT_WZ_NC} and~\ref{thm:Simple_Acivable_Region_CSIT_NC} introduce two achievable schemes for the same problem (non-causal \ac{CSI} at transmitter). While the expression in \eqref{eq:finalregion_CSIT_NC_Inner} is not larger than \eqref{eq:finalregion_CSIT_NC_Simple_Inner} when calculated for the same probability distribution, the set of admissible probability distributions defined by \eqref{eq:finalconstraint_CSIT_NC_Inner} may contain distributions not in \eqref{eq:finalconstraint_CSIT_NC_Simple_Inner}. 
\end{remark}

\begin{theorem}
\label{thm:Upper_Bound_For_CSIT_NC}
Let
\begin{subequations}\label{eq:finalSD_CSIT_NC}
\begin{align}
\mathcal{S} &\triangleq
  \big\{R\geq0: \exists P_{S,U,V,X,Y,Z}\in\calD \text{ such that } R\leq \min\{\mi(U;Y)-\mi(U;S),\mi(U,V;Y)-\mi(U;S|V)\}\big\}
\label{eq:finalregion_CSIT_NC}
\end{align}
where,
\begin{align}
  \calD \eqdef \left.\begin{cases}P_{S,U,V,X,Y,Z}:\\
P_{S,U,V,X,Y,Z}=Q_SP_{UV|S}\indic{1}_{\big\{X=X(U,S)\big\}}W_{Y,Z|X,S}\\
P_Z=Q_0\\
\min\{\mi(U;Y)-\mi(U;S),\mi(U,V;Y)-\mi(U;S|V)\} \geq \mi(V;Z)-\mi(V;S)\\
\max\{\card{\calU},\card{\calV}\}\leq \card{\calX}+3
\end{cases}\right\}.\label{eq:finalconstraint_CSIT_NC}
\end{align}
\end{subequations}
The covert capacity of the \ac{DMC} $W_{Y,Z|S,X}$ with non-causal \ac{CSI} at the transmitter is upper-bounded as 
\begin{align}
C_{\mbox{\scriptsize\rm NC-T}} \leq\max\{x:x\in\mathcal{S}\}.
\label{eq:Upper_Bound_For_CSIT_NC}
\end{align}
\end{theorem}
Proof details are available in Appendix~\ref{sec:Proof_Upper_Bound_For_CSIT_NC}.

\begin{corollary}
\label{thm:Capacity_Region_For_Degraded_CSIT_NC}
Let
\begin{subequations}\label{eq:finalSD_CSIT_NC_LN}
\begin{align}
\mathcal{S} &\triangleq
  \big\{R\geq0: \exists P_{S,U,X,Y,Z}\in\calD \text{ such that } R\leq \mi(U;Y)-\mi(U;S)\big\}
\label{eq:finalregion_CSIT_NC_LN}
\end{align}
where,
\begin{align}
  \calD \eqdef \left.\begin{cases}P_{S,U,X,Y,Z}:\\
P_{S,U,X,Y,Z}=Q_SP_{U|S}\indic{1}_{\big\{X=X(U,S)\big\}}W_{Y,Z|X,S}\\
P_Z=Q_0\\
\card{\calU}\leq \card{\calX}+2
\end{cases}\right\}.\label{eq:finalconstraint_CSIT_NC_LN}
\end{align}
\end{subequations}
The covert capacity of the less noisy channel $W_{Y,Z|X,S}$ \big(i.e. $\mi(U;Y)\geq \mi(U;Z)$\big) with \ac{CSI} available non-causally only at the transmitter is,
\begin{align}
C_{\mbox{\scriptsize\rm NC-T}} =\max\{x:x\in\mathcal{S}\}.
\label{eq:Upper_Bound_For_CSIT_NC_LN}
\end{align}
\end{corollary}
\begin{proof}
The achievability is proved by using Theorem~\ref{thm:Simple_Acivable_Region_CSIT_NC} and the less noisy property of the channel. 
We can also proof the achievability by using Theorem~\ref{thm:Acivable_Rate_CSIT_WZ_NC} while generating $S$ independent of $V$ (i.e. $P_{S|U,V}=P_{S|U}$) and the less noisy property of the channel. Furthermore, the converse is proved by utilizing Theorem~\ref{thm:Upper_Bound_For_CSIT_NC} and the less noisy property of the channel.
\end{proof}

\subsection{Causal \ac{CSI}}

\begin{theorem}
\label{thm:Acivable_Rate_For_CSIT_C}
Let
\begin{subequations}\label{eq:finalSD_CSIT_C_Inner}
\begin{align}
\mathcal{S} &\triangleq
  \big\{R\geq0: \exists P_{U,V,S,X,Y,Z}\in\calD \text{ such that } R\leq \mi(U;Y)+\min\left\{0, \mi(V;Y|U) -\mi(V;S)\right\}
\label{eq:finalregion_CSIT_C_Inner}
\end{align}
where,
\begin{align}
  \calD \eqdef \left.\begin{cases}P_{U,V,S,X,Y,Z}:\\
P_{U,V,S,X,Y,Z}=P_UP_VP_{S|V}\indic{1}_{\big\{X=X(U,S)\big\}}W_{Y,Z|X,S}\\
Q_S(\cdot)=\sum\nolimits_{v\in\mathcal{V}}P_V(v)P_{S|V}(\cdot|v)\\
P_Z=Q_0\\
\mi(V;Y|U) >  \max\{\mi(V;Z),\mi(U,V;Z)-\mi(U;Y)\}\\
\card{\calU}\leq \card{\calX}+2\\
\card{\calV}\leq \card{\calX}+3
\end{cases}\right\}.\label{eq:finalconstraint_CSIT_C_Inner}
\end{align}
\end{subequations}
The covert capacity of the \ac{DMC} $W_{Y,Z|S,X}$ with causal \ac{CSI} at the transmitter is lower-bounded as
\begin{align}
C_{\mbox{\scriptsize\rm C-T}} \geq\max\{x:x\in\mathcal{S}\}.
\label{eq:Inner_Bound_For_CSIT_C}
\end{align}
\end{theorem}
Theorem~\ref{thm:Acivable_Rate_For_CSIT_C} is proved by a block Markov encoding, Shannon strategy for sending the message according to \ac{CSI}, and Wyner-Ziv coding for secret key generation. We design a block Markov coding scheme in which the transmitter not only generates a key from $S^n$ but also selects its codeword according to $S^n$. The details of the proof are available in Appendix~\ref{sec:Proof_Acivable_Rate_Condition_For_CSIT_C}.
\begin{theorem}
\label{thm:Acivable_Rate_For_CSIT_C_Checking}
Let
\begin{subequations}\label{eq:finalSD_CSIT_C_Simple_Inner}
\begin{align}
\mathcal{S} &\triangleq
  \big\{R\geq0: \exists P_{S,U,X,Y,Z}\in\calD \text{ such that } R\leq \mi(U;Y)\big\}
\label{eq:finalregion_CSIT_C_Simple_Inner}
\end{align}
where,
\begin{align}
  \calD \eqdef \left.\begin{cases}P_{S,U,X,Y,Z}:\\
P_{S,U,X,Y,Z}=Q_SP_U\indic{1}_{\big\{X=X(U,S)\big\}}W_{Y,Z|X,S}\\
P_Z=Q_0\\
\mi(U;Y) > \mi(U;Z)\\
\card{\calU}\leq \card{\calX}+1
\end{cases}\right\}.\label{eq:finalconstraint_CSIT_C_Simple_Inner}
\end{align}
\end{subequations}
The covert capacity of the \ac{DMC} $W_{Y,Z|S,X}$ with causal \ac{CSI} at the transmitter is lower-bounded as
\begin{align}
C_{\mbox{\scriptsize\rm C-T}} \geq\max\{x:x\in\mathcal{S}\}.
\label{eq:Simple_Inner_Bound_For_CSIT_C}
\end{align}
\end{theorem}

The details of the proof are available in Appendix~\ref{sec:Proof_Simple_Inner_Bound_For_CSIT_C}.
\begin{theorem}
\label{thm:Upper_Bound_For_CSIT_C}
Let
\begin{subequations}\label{eq:finalSD_CSIT_C}
\begin{align}
\mathcal{S} &\triangleq
  \big\{R\geq0: \exists P_{S,U,V,X,Y,Z}\in\calD \text{ such that } R\leq \mi(U;Y)\big\}
\label{eq:finalregion_CSIT_C}
\end{align}
where,
\begin{align}
  \calD \eqdef \left.\begin{cases}P_{S,U,V,X,Y,Z}:\\
P_{S,U,V,X,Y,Z}=Q_SP_VP_{U|V}\indic{1}_{\big\{X=X(U,S)\big\}}W_{Y,Z|X,S}\\
P_Z=Q_0\\
\mi(U;Y) \geq \mi(V;Z)\\
\max\{\card{\calU},\card{\calV}\}\leq \card{\calX}
\end{cases}\right\}.\label{eq:finalconstraint_CSIT_C}
\end{align}
\end{subequations}
The covert capacity of the \ac{DMC} $W_{Y,Z|S,X}$ with causal \ac{CSI} at the transmitter is upper-bounded as 
\begin{align}
C_{\mbox{\scriptsize\rm C-T}} \leq\max\{x:x\in\mathcal{S}\}.
\label{eq:Upper_Bound_For_CSIT_C}
\end{align}
\end{theorem}
Proof details are available in Appendix~\ref{sec:Proof_Upper_Bound_For_CSIT_C}.
\begin{corollary}
\label{thm:Capacity_Region_For_Degraded_CSIT_C}
Let
\begin{subequations}\label{eq:finalSD_CSIT_C_LN}
\begin{align}
\mathcal{S} &\triangleq
  \big\{R\geq0: \exists P_{S,U,X,Y,Z}\in\calD \text{ such that } R\leq \mi(U;Y)\big\}
\label{eq:finalregion_CSIT_C_LN}
\end{align}
where,
\begin{align}
  \calD \eqdef \left.\begin{cases}P_{S,U,X,Y,Z}:\\
P_{S,U,X,Y,Z}=Q_SP_U\indic{1}_{\big\{X=X(U,S)\big\}}W_{Y,Z|X,S}\\
P_Z=Q_0\\
\card{\calU}\leq \card{\calX}+1
\end{cases}\right\}.\label{eq:finalconstraint_CSIT_C_LN}
\end{align}
\end{subequations}
The covert capacity of the less noisy channel $W_{Y,Z|X,S}$ \big(i.e. $\mi(U;Y)\geq \mi(U;Z)$\big) with \ac{CSI} available causally only at the transmitter is,
\begin{align}
C_{\mbox{\scriptsize\rm C-T}} =\max\{x:x\in\mathcal{S}\}.
\label{eq:Upper_Bound_For_CSIT_C_LN}
\end{align}
\end{corollary}
\begin{proof}
The achievability is proved by using Theorem~\ref{thm:Acivable_Rate_For_CSIT_C_Checking} and the less noisy property of the channel. 
We can also proof the achievability by using Theorem~\ref{thm:Acivable_Rate_For_CSIT_C} while generating $S$ independent of $V$ (i.e. $P_{S|V}=Q_S$) and the less noisy property of the channel. Furthermore, the converse is proved by utilizing Theorem~\ref{thm:Upper_Bound_For_CSIT_C} and the less noisy property of the channel.
\end{proof}

\subsection{Strictly Causal \ac{CSI}}

The availability of strictly causal channel state information at the transmitter is useful because it provides a source of common randomness between the legitimate users, which can be used to generate a shared secret key between them to mislead the warden. Here, we only use the channel state information for key generation and not for data transmission.
\begin{theorem}
\label{thm:Strictly_Causal_State_CSIT}
Let
\begin{subequations}\label{eq:finalSD_CSIT_SC_Inner}
\begin{align}
\mathcal{S} &\triangleq
  \big\{R\geq0: \exists P_{X,V,S,Y,Z}\in\calD \text{ such that } R\leq \mi(X;Y)+\min\left\{0, \mi(V;Y|X) -\mi(V;S)\right\}
\label{eq:finalregion_CSIT_SC_Inner}
\end{align}
where,
\begin{align}
  \calD \eqdef \left.\begin{cases}P_{X,V,X,Y,Z}:\\
P_{X,V,S,Y,Z}=P_XP_VP_{S|V}W_{Y,Z|X,S}\\
Q_S(\cdot)=\sum\nolimits_{v\in\mathcal{V}}P_V(v)P_{S|V}(\cdot|v)\\
P_Z=Q_0\\
\mi(V;Y|X) >  \max\{\mi(V;Z),\mi(X,V;Z)-\mi(X;Y)\}\\
\card{\calV}\leq \card{\calX}+3
\end{cases}\right\}.\label{eq:finalconstraint_CSIT_SC_Inner}
\end{align}
\end{subequations}
The covert capacity of the \ac{DMC} $W_{Y,Z|S,X}$ with strictly causal \ac{CSI} at the transmitter is lower-bounded as
\begin{align}
C_{\mbox{\scriptsize\rm SC-T}} \geq\max\{x:x\in\mathcal{S}\}.
\label{eq:Inner_Bound_For_CSIT_SC}
\end{align}
\end{theorem}
Proof details are available in Appendix~\ref{sec:Proof_Acivable_Rate_Condition_For_CSIT_SC}.
\begin{theorem}
\label{thm:Strictly_Causal_State_CSIT_App}
Let
\begin{subequations}\label{eq:finalSD_CSIT_SC_Simple_Inner}
\begin{align}
\mathcal{S} &\triangleq
  \big\{R\geq0: \exists P_{S,X,Y,Z}\in\calD \text{ such that } R\leq \mi(X;Y)\big\}
\label{eq:finalregion_CSIT_SC_Simple_Inner}
\end{align}
where,
\begin{align}
  \calD \eqdef \left.\begin{cases}P_{S,X,Y,Z}:\\
P_{S,X,Y,Z}=Q_SP_XW_{Y,Z|X,S}\\
P_Z=Q_0\\
\mi(X;Y) > \mi(X;Z)
\end{cases}\right\}.\label{eq:finalconstraint_CSIT_SC_Simple_Inner}
\end{align}
\end{subequations}
The covert capacity of the \ac{DMC} $W_{Y,Z|S,X}$ with strictly causal \ac{CSI} at the transmitter is lower-bounded as
\begin{align}
C_{\mbox{\scriptsize\rm SC-T}} \geq\max\{x:x\in\mathcal{S}\}.
\label{eq:Simple_Inner_Bound_For_CSIT_SC}
\end{align}
\end{theorem}
Proof details are available in Appendix~\ref{sec:Proof_Simple_Inner_Bound_For_CSIT_SC}.
We now present an upper bound on the covert capacity when the \ac{CSI} is available strictly causally at the transmitter.
\begin{theorem}
\label{thm:Upper_Bound_For_CSIT_SC}
Let
\begin{subequations}\label{eq:finalSD_CSIT_SC}
\begin{align}
\mathcal{S} &\triangleq
  \big\{R\geq0: \exists P_{S,V,X,Y,Z}\in\calD \text{ such that } R\leq \mi(X;Y)\big\}
\label{eq:finalregion_CSIT_SC}
\end{align}
where,
\begin{align}
  \calD \eqdef \left.\begin{cases}P_{S,V,X,Y,Z}:\\
P_{S,V,X,Y,Z}=Q_SP_VP_{X|V}W_{Y,Z|X,S}\\
P_Z=Q_0\\
\mi(X;Y) \geq \mi(V;Z)\\
\card{\calV}\leq \card{\calX}
\end{cases}\right\}.\label{eq:finalconstraint_CSIT_SC}
\end{align}
\end{subequations}
The covert capacity of the \ac{DMC} $W_{Y,Z|S,X}$ with strictly causal \ac{CSI} at the transmitter is upper-bounded as 
\begin{align}
C_{\mbox{\scriptsize\rm SC-T}} \leq\max\{x:x\in\mathcal{S}\}.
\label{eq:Upper_Bound_For_CSIT_SC}
\end{align}
\end{theorem}
Proof details are available in Appendix~\ref{sec:Proof_Upper_Bound_For_CSIT_SC}.
\begin{remark}[Strictly Causal \ac{CSI} Does not Increase Capacity]
The segment of the proof culminating in~\eqref{eq:Upper_Bound_Strictly_Causal_CSIT_1} implies that strictly causal \ac{CSI} does not increase the capacity of a point-to-point \ac{DMC}. This result is known (e.g., see~\cite[Page~13]{Lapidoth2010}) but to the best of our knowledge no formal proof for it exists in the literature.
\end{remark}
\begin{corollary}
\label{thm:Capacity_Region_For_Degraded_CSIT_SC}
Let
\begin{subequations}\label{eq:finalSD_CSIT_SC_LN}
\begin{align}
\mathcal{S} &\triangleq
  \big\{R\geq0: \exists P_{S,X,Y,Z}\in\calD \text{ such that } R\leq \mi(X;Y)\big\}
\label{eq:finalregion_CSIT_SC_LN}
\end{align}
where,
\begin{align}
  \calD \eqdef \left.\begin{cases}P_{S,X,Y,Z}:\\
P_{S,X,Y,Z}=Q_SP_XW_{Y,Z|X,S}\\
P_Z=Q_0
\end{cases}\right\}.\label{eq:finalconstraint_CSIT_SC_LN}
\end{align}
\end{subequations}
The covert capacity when the channel is degraded \big(i.e. $W_{Y,Z|S,X}=W_{Y|S,X}W_{Z|Y}$\big) and \ac{CSI} is available strictly causally at the transmitter is,
\begin{align}
C_{\mbox{\scriptsize\rm SC-T}} =\max\{x:x\in\mathcal{S}\}.
\label{eq:Upper_Bound_For_CSIT_SC_LN}
\end{align}
\end{corollary}
\begin{proof}
The achievability is proved by using Theorem~\ref{thm:Strictly_Causal_State_CSIT_App} and the less noisy property of the channel. 
We can also proof the achievability by using Theorem~\ref{thm:Strictly_Causal_State_CSIT} while generating $S$ independent of $V$ (i.e. $P_{S|V}=Q_S$) and the less noisy property of the channel. Furthermore, the converse is proved by utilizing Theorem~\ref{thm:Upper_Bound_For_CSIT_SC} and the less noisy property of the channel.
\end{proof}
\begin{remark}[Cardinality Bounds]
The cardinality bounds on the auxiliary random variables in Theorems~\ref{thm:Capacity_FCSI_C} to \ref{thm:Strictly_Causal_State_CSIT} follows by a standard application of the Eggleston-Fenchel-Carath\'eodory theorem \cite[Theorem~18]{Eggleston}. Details are omitted for brevity.
\end{remark}
\section{Examples}
We provide two examples of covert communication without an external secret key, while surpassing the square-root-$n$ performance with a positive rate. 
In both of the examples the channel state is additive however in the first example the warden's channel is a degraded version of the legitimate receiver's channel but in the second example the legitimate receiver's channel is a degraded version of the warden's channel. The second example shows that our proposed coding scheme with block-Markov encoding and Wyner-Ziv encoding for secret key generation can outperform the simple approach for deriving the covert rates.
\subsection{Degraded Channel with Binary Additive State}
\begin{figure}
\centering
\includegraphics[width=3.4in]{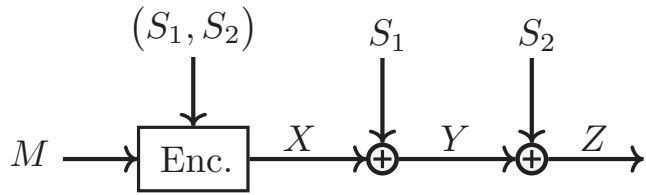}
\caption{Degraded channel with binary additive \ac{CSI} at the transmitter}
\label{fig:BSC_Additive_State_CSIT}
\end{figure}
We provide an example in which the covert capacity is positive. Consider a channel in which $X, Y, Z$ and $S=(S_1,S_2)$ are all binary, and let $S_1$, $S_2$, and $U$ be independent Bernoulli random variables with parameters $\alpha\in\intseq{1}{0.5}$, $\beta\in\intseq{1}{0.5}$, and $\lambda\in\intseq{1}{0.5}$, respectively, and let $x_0=0$ (See Fig.~\ref{fig:BSC_Additive_State_CSIT}). Here, $S_1$ and $S_2$ are the channel states of the legitimate receiver's channel and the warden's channel, respectively, and $U$ is an auxiliary random variable that represents the message. The channel state information is available causally at the Encoder and the law of the channel is as follows
\begin{align}
Y=X\oplus S_1,\label{eq:BSC_Additive_Channel_Law_CSIT_Y}\\
Z=Y\oplus S_2.\label{eq:BSC_Additive_Channel_Law_CSIT_Z}
\end{align}
\begin{proposition}
\label{thm:Achievable_Region_Binary_Additive_Binary_State}
The covert capacity of the \ac{DMC} channel depicted in Fig.~\ref{fig:BSC_Additive_State_CSIT} with causal \ac{CSI} at the transmitter is
\begin{align}
C_{\mbox{\scriptsize\rm C-T}}\mathop = \limits^{(a)}\mathop {\max }\limits_{\mathop {P_U,}\limits_{P_Z = Q_0} }\ent(U)\mathop = \limits^{(b)}\ent_b(\alpha)\,
\label{eq:Achievable_Region_Binary_Additive_Binary_State}
\end{align}where $\ent_b(\cdot)$ is binary entropy.
\end{proposition}
\begin{proof}The achievability proof for $(a)$ follows from the achievability part of Corollary~\ref{thm:Capacity_Region_For_Degraded_CSIT_C} and setting $X=U\oplus S_1$. The converse part of $(a)$ follows from the converse part of Corollary~\ref{thm:Capacity_Region_For_Degraded_CSIT_C} and the fact that $\mi(U;Y)\leq\ent(U)$. To prove $(b)$ in Proposition~\ref{thm:Achievable_Region_Binary_Additive_Binary_State}, we have
\begin{align}
    Q_0(z=0)&=P(s_1\oplus s_2=0)\nonumber\\
    &=P(s_1=0,s_2=0)+P(s_1=1,s_2=1)\nonumber\\
    &=(1-\alpha)(1-\beta)+\alpha\beta.
\end{align}
The distribution induced at the output of the warden when transmitting a codeword is
\begin{align}
    P_Z(z=0)&=P(u\oplus s_2=0)\nonumber\\
    &=P(u=0,s_2=0)+P(u=1,s_2=1)\nonumber\\
    &=(1-\lambda)(1-\beta)+\lambda\beta.
\end{align}Therefore, the covertness constraint $P_Z=Q_0$ requires
\begin{align}
    \lambda=\alpha.
\end{align}
\end{proof}
\subsection{Reverse Degraded Channel with Binary Additive State}
\label{Sec:Example_Reverse_Degraded}
\begin{figure}
\centering
\includegraphics[width=3.4in]{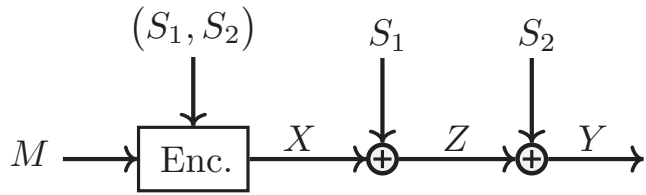}
\caption{Reverse degraded channel with binary additive \ac{CSI} at the transmitter}
\label{fig:Reverse_BSC_Additive_State_CSIT}
\end{figure}
To show the efficiency of the proposed scheme in this paper, we provide an example in which the region in Theorem~\ref{thm:Acivable_Rate_For_CSIT_C} results in a positive rate for covert communication while the region in Theorem~\ref{thm:Acivable_Rate_For_CSIT_C_Checking} results in zero rate. Consider a channel in which $X, Y, Z$ and $S=(S_1,S_2)$ are all binary, and let $S_1$, $S_2$ and $U$ be independent Bernoulli random variables with parameters $\alpha\in\big(0{\,:\,}0.5\rrbracket$, $\beta\in\big(0{\,:\,}0.5\rrbracket$, and $\lambda\in\big(0{\,:\,}0.5\rrbracket$, respectively, and let $x_0=0$ (See Fig.~\ref{fig:Reverse_BSC_Additive_State_CSIT}). Also, let $V$ be a Bernoulli random variable. Here, $S_1$ and $S_2$ are the channel states of the warden's channel and the legitimate receiver's channel, respectively, $U$ is an auxiliary random variable that represents the message, and $V$ is an auxiliary random variable that represents a description of the \ac{CSI}. The channel state information is available causally at the Encoder and the law of the channel is as follows
\begin{align}
Z=X\oplus S_1,\label{eq:Rev_BSC_Additive_Channel_Law_CSIT_Z}\\
Y=Z\oplus S_2.\label{eq:Rev_BSC_Additive_Channel_Law_CSIT_Y}
\end{align}Since for this example $\mi(U;Z)\geq\mi(U;Y)$, the achievable rate region in Theorem~\ref{thm:Acivable_Rate_For_CSIT_C_Checking} results in zero rate but Theorem~\ref{thm:Acivable_Rate_For_CSIT_C} results in the following region.
\begin{proposition}
\label{thm:Rev_Achievable_Region_Binary_Additive_Binary_State}
The covert capacity of the \ac{DMC} channel depicted in Fig.~\ref{fig:Reverse_BSC_Additive_State_CSIT} with causal \ac{CSI} at the transmitter is lower bounded as
\begin{align}
C_{\mbox{\scriptsize\rm C-T}}\geq\ent_b(\eta)-\ent_b(\beta),
\label{eq:Rev_Achievable_Region_Binary_Additive_Binary_State}
\end{align}where $\eta=\alpha\beta+(1-\alpha)(1-\beta)$.
\end{proposition}
\begin{proof}Here we choose $X=U\oplus S_1$ therefore $Z=U$ and $Y=U\oplus S_2$. To proof the region in Proposition~\ref{thm:Rev_Achievable_Region_Binary_Additive_Binary_State} by using Theorem~\ref{thm:Acivable_Rate_For_CSIT_C}, we start with covertness constraint $P_Z=Q_0$,
\begin{align}
    Q_0(z=0)&=\Prob(s_1=0)=\alpha\label{eq:Q0_Rev_Deg},\\
    P_Z(z=0)&=\Prob(u=0)=\lambda\label{eq:QZ_Rev_Deg}.
\end{align}Therefore, the covertness constraint requires $\lambda=\alpha$. We also choose $V=S_2$ therefore, 
\begin{align}
    P_{V|S}&=P_{V|S_1,S_2}=P_{V|S_2}=\indic{1}_{\{V=S_2\}}.\label{eq:V}
\end{align}The chaining between the random variables for this example is depicted in Fig.~\ref{fig:Reverse_BSC_Additive_Chaining}.
\begin{figure}
\centering
\includegraphics[height=1.2in]{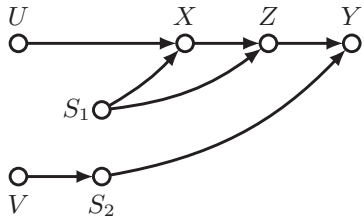}
\caption{Chaining between the random variables for reverse degraded channel with binary additive \ac{CSI}}
\label{fig:Reverse_BSC_Additive_Chaining}
\end{figure}
Now we show that the fourth condition in Theorem~\ref{thm:Acivable_Rate_For_CSIT_C} which is the following conditions is satisfied,
\begin{align}
    \mi(V;Y|U) &> \mi(V;Z),\label{eq:first_Condition}\\ 
    \mi(U,V;Y) &> \mi(U,V;Z).\label{eq:second_Condition}
\end{align}We now have,
\begin{align}
    \mi(V;Y|U)&=\ent(S_2|U)-\ent(S_2|U,Y)\nonumber\\
    &=\ent(S_2)-\ent(S_2|U,U\oplus S_2)\nonumber\\
    &=\ent(S_2)=\ent_b(\beta),\nonumber\\
    \mi(V;Z)&=\ent(S_2)-\ent(S_2|U)\mathop = \limits^{(a)}0.\nonumber
\end{align}where $(a)$ follows since $U$ and $S_2$ are independent. Therefore, the condition \eqref{eq:first_Condition} is satisfied. For the condition in \eqref{eq:second_Condition} we have,
\begin{align}
    \mi(U,V;Y)&=\ent(U,S_2)-\ent(U,S_2|Y)\nonumber\\
    &=\ent(U)+\ent(S_2)-\ent(U,S_2|U\oplus S_2)\nonumber\\
    &=\ent(U)+\ent(S_2)-\ent(S_2|U\oplus S_2),\nonumber\\
    \mi(U,V;Z)&=\mi(U,S_2;U)=\ent(U)\nonumber
\end{align}since $\ent(S_2)-\ent(S_2|U\oplus S_2)>0$ the condition in \eqref{eq:second_Condition} is also satisfied. To calculate the covert rate \eqref{eq:finalregion_CSIT_C_Inner} in Theorem~\ref{thm:Acivable_Rate_For_CSIT_C} we have,
\begin{align}
    \mi(V;Y|U)&=\ent_b(\beta),\nonumber\\
    \mi(V;S)&=\mi(S_2;S_1,S_2)\nonumber\\
    &=\ent(S_2)=\ent_b(\beta),\nonumber\\
    \mi(U;Y)&=\ent(Y)-\ent(Y|U)\nonumber\\
    &=\ent(U\oplus S_2)-\ent(U\oplus S_2|U)\nonumber\\
    &=\ent(U\oplus S_2)-\ent(S_2)\nonumber\\
    &=\ent_b(\eta)-\ent_b(\beta),\nonumber
\end{align}where $\eta=\alpha\beta+(1-\alpha)(1-\beta)$.
\begin{remark}[Covertness vs.\ Security]
  This example also captures the difference between covertness and security. Here the warden has noiseless access to the transmitted sequence and therefore it can decode the transmitted message, but since the transmitted sequence has the same statistics as the \ac{CSI} it cannot prove that communication is happening.
\end{remark}
\begin{remark}[Shared Key]
  In the examples in this section, the codebooks are generated with the same distribution as the channel state $S_1$ therefore the legitimate terminals need to have access to a shared key of rate $\epsilon$ to discriminate the codewords from the \ac{CSI}. Note that, in our achievability scheme for Theorem~\ref{thm:Acivable_Rate_For_CSIT_C} the transmitter and the receiver have access to a shared secret key of rate $\epsilon$ therefore the legitimate terminals can use this shared key in the achievability scheme to discriminate the codewords from the \ac{CSI}.
\end{remark}
\end{proof}

\section{Conclusion}
This paper studies keyless covert communication over state dependant channels, when the channel state information is available either at the transmitter alone, or at both the transmitter and receiver, but not to the adversary (warden). Our results show the feasibility of covertly communicating with a positive rate without an externally shared key between the transmitter and the receiver. This is in stark contrast with the known results showing that in the absence of channel state information, covert communication without a shared key is impossible at positive rates.

\begin{appendices}

\section{Proof of Theorem~\ref{thm:Capacity_FCSI_NC}}
\label{sec:Proof_thm_FCSI_NC}
\subsection{Achievability Proof}
Fix $P_{X|S}(x|s)$ and $\epsilon_1>\epsilon_2>0$ such that, $P_Z = Q_0$.

\emph{{Codebook Generation:}}
For every $S^n\in\mathcal{S}^n$ let $C_n\triangleq\{X^n(S^n,m)\}_{m\in\mathcal{M}}$, where $\mathcal{M}\in\intseq{1}{2^{nR}}$, be a random codebook consisting of independent random sequences each generated according to $P_{X|S}^{\otimes n}(\cdot|s_i)$. We denote a realization of $C_n$ by $\CodeBook_n\triangleq\{x^n(s^n,m)\}_{m\in\mathcal{M}}$.

\emph{{Encoding:}}
Given the channel state information $s^n$, to send message $m$ the transmitter computes $x^n(s^n,m)$ and transmits it over the channel.

For a fixed codebook $\CodeBook_n$, the induced joint distribution is
\begin{align}
    \twocolalign P_{S^n,M,X^n,Z^n}^{(\CodeBook_n)}(s^n,m,\tilde{x}^n,z^n) =Q_S^n(s^n) 2^{-nR}\indic{1}_{\{\tilde{x}^n=x^n(s^n,m)\}} W_{Z|S,X}^{\otimes n}(z^n|s^n,\tilde{x}^n).\label{eq:P_Dist_NC_FCSI}
\end{align}

\emph{Covert Analysis:} We now show  $\expec_{C_n}[\kld(P_{Z^n|C_n} || Q_Z^{\otimes n} )]\underset{n\rightarrow\infty}{\xrightarrow{\hspace{0.2in}}} 0$, and then choose $P_{X|S}$ such that it satisfies $Q_Z=Q_0$. 
Henceforth, we denote by $P^{(\CodeBook_n)}$ the distributions induced by a fixed codebook $\CodeBook_n$, and by $P_{\cdot|C_n}$ the distributions induced by a random codebook $C_n$. 
First, consider the following marginal from \eqref{eq:P_Dist_NC_FCSI},
\begin{align}
    &P_{Z^n|C_n} (z^n)= \sum\limits_{s^n}\sum\limits_{m} Q_S^{\otimes n}(s^n)2^{-nR}W_{Z|S,X}^{\otimes n}\big(z^n|s^n,X^n(s^n,m)\big).\label{eq:Gamma_ZK_NC_FCSI}  
\end{align}
We now have,
\begin{align}
& \expec_{C_n}\big[\kld\big(P_{Z^n|C_n}||Q_Z^{\otimes n}\big)\big] = \expec_{C_n}\Big[\sum\limits_{z^n} P_{Z^n|C_n}(z^n)\log \Big(\frac{P_{Z^n|C_n}(z^n)}{Q_Z^{\otimes n}(z^n)}\Big) \Big] \nonumber\\
&= \expec_{C_n}\Bigg[\sum\limits_{z^n} \sum\limits_{s^n} \sum\limits_{m} \frac{1}{2^{nR}} Q_S^{\otimes n}(s^n)W_{Z|S,X}^{\otimes n}\big(z^n|s^n,X^n(s^n,m)\big)\log \Bigg(\frac{\sum\limits_{(\tilde{s}^n,\tilde{m})}  Q_S^{\otimes n}(\tilde{s}^n)W_{Z|S,X}^{\otimes n}\big(z^n|\tilde{s}^n,X^n(\tilde{s}^n,\tilde{m})\big)}{2^{nR}Q_Z^{\otimes n}(z^n)}\Bigg)     \Bigg] \nonumber\\ 
&\mathop \le \limits^{(a)} \sum\limits_{z^n} \sum\limits_{s^n} \sum\limits_{m} \frac{1}{2^{nR}}\sum\limits_{x^n(s^n,m)} P_{S,X,Z}^{\otimes n}\big(s^n,x^n(s^n,m),z^n\big)\log \expec_{\backslash m} \Bigg[\frac{\sum\limits_{(\tilde{s}^n,\tilde{m})}  Q_S^{\otimes n}(\tilde{s}^n)W_{Z|S,X}^{\otimes n}\big(z^n|\tilde{s}^n,X^n(\tilde{s}^n,\tilde{m})\big)}{2^{nR}Q_Z^{\otimes n}(z^n)}\Bigg]        \nonumber\\
&= \sum\limits_{z^n} \sum\limits_{s^n} \sum\limits_{m} \frac{1}{2^{nR}}\sum\limits_{x^n(s^n,m)} P_{S,X,Z}^{\otimes n}\big(s^n,x^n(s^n,m),z^n\big)\nonumber\\
&\qquad\times\log  \Bigg(\frac{Q_S^{\otimes n}(s^n)W_{Z|S,X}^{\otimes n}\big(z^n|s^n,x^n(s^n,m)\big)}{2^{nR}Q_Z^{\otimes n}(z^n)} + \expec_{\backslash m}\Bigg[\frac{\sum\limits_{(\tilde{s}^n,\tilde{m})\ne(s^n,m)}Q_S^{\otimes n}(s^n)W_{Z|S,X}^{\otimes n}\big(z^n|s^n,X^n(s^n,\tilde{m})\big)}{2^{nR}Q_Z^{\otimes n}(z^n)}\Bigg]\Bigg)\nonumber\\
&\mathop \le \limits^{(b)} \sum\limits_{z^n} \sum\limits_{s^n} \sum\limits_{m} \frac{1}{2^{nR}}\sum\limits_{x^n(s^n,m)} P_{S,X,Z}^{\otimes n}\big(s^n,x^n(s^n,m),z^n\big)\nonumber\\
&\qquad\times\log  \Bigg(\frac{Q_S^{\otimes n}(s^n)W_{Z|S,X}^{\otimes n}\big(z^n|s^n,x^n(s^n,m)\big)}{2^{nR}Q_Z^{\otimes n}(z^n)} + \frac{\sum\limits_{(\tilde{s}^n,\tilde{m})}Q_S^{\otimes n}(s^n)W_{Z|S}^{\otimes n}\big(z^n|s^n\big)}{2^{nR}Q_Z^{\otimes n}(z^n)}\Bigg)\nonumber\\
&= \sum\limits_{z^n} \sum\limits_{s^n} \sum\limits_{m} \frac{1}{2^{nR}}\sum\limits_{x^n(s^n,m)} P_{S,X,Z}^{\otimes n}\big(s^n,x^n(s^n,m),z^n\big)\log  \Bigg(\frac{Q_S^{\otimes n}(s^n)W_{Z|S,X}^{\otimes n}\big(z^n|s^n,x^n(s^n,m)\big)}{2^{nR}Q_Z^{\otimes n}(z^n)} + 1\Bigg)\nonumber\\
&\triangleq \Psi_1 + \Psi_2,\label{eq:KLD_ZK_NC_FCSI}
\end{align}where $(a)$ follows from Jensen's inequality and $(b)$ follows by adding some terms to the nominator of the second term in the argument of the $\log$ function. We define $ \Psi_1$ and $\Psi_2$ as
\begin{align}
 &{\Psi _1} =  \sum\limits_{m} \frac{1}{2^{nR}}\sum\limits_{(s^n,x^n(s^n,m),z^n) \in\mathcal{T}_{\epsilon}^{(n)}} P_{S,X,Z}^{\otimes n}\big(s^n,x^n(s^n,m),z^n\big)\log  \Bigg(\frac{Q_S^{\otimes n}(s^n)W_{Z|S,X}^{\otimes n}\big(z^n|s^n,x^n(s^n,m)\big)}{2^{nR}Q_Z^{\otimes n}(z^n)}  + 1   \Bigg)     \nonumber\\
& \le \log \Bigg(\frac{2^{ - n(1 - \epsilon )\big(\ent(S)+\ent(Z|S,X)\big)}}{2^{nR}2^{ - n(1 + \epsilon )\ent(Z)}}   + 1\Bigg)\label{eq:Psi_1_NC_FCSI}\\
&{\Psi _2} =  \sum\limits_{m} \frac{1}{2^{nR}}\sum\limits_{(s^n,x^n(s^n,m),z^n) \notin\mathcal{T}_{\epsilon}^{(n)}} P_{S,X,Z}^{\otimes n}\big(s^n,x^n(s^n,m),z^n\big)\log  \Bigg(\frac{Q_S^{\otimes n}(s^n)W_{Z|S,X}^{\otimes n}\big(z^n|s^n,x^n(s^n,m)\big)}{2^{nR}Q_Z^{\otimes n}(z^n)} + 1   \Bigg)     \nonumber\\
&\le 2|S||X||Z|e^{ - n\epsilon^2 {\mu _{S,X,Z}}}n\log\Big(\frac{3}{\mu_Z} + 1\Big) \label{eq:Psi_2_NC_FCSI}
\end{align}
\sloppy where $\mu_{S,X,Z}=\min\limits_{(s,x,z)\in(\mathcal{S},\mathcal{X},\mathcal{Z})}P_{S,X,Z}(s,x,z)$ and $\mu_Z=\min\limits_{z\in\mathcal{Z}}P_Z(z)$. When $n\to\infty$ then $\Psi_2\to 0$ when $n$ grows if
\begin{align}
    &R >\mi(S,X;Z)-\ent(S).\label{eq:res_1}
\end{align}
Basic information identities yield:
\begin{align}
    \mi(X,S;Z)=\mi(X;S,Z)+\mi(S;Z)-\mi(X;S).\label{eq:MI_Manipulation}
\end{align}Substituting \eqref{eq:MI_Manipulation} into \eqref{eq:res_1} leads to
\begin{align}
    R>\mi(X;Z|S) -\ent(S|Z).\label{eq:res_NC}
\end{align}

\emph{{Decoding and Error Probability Analysis:}}
By access to the channel state $s^n$, the receiver declares that $\hat{m}=m$ if the exists a unique index $\hat{m}$ such that $(x^n(s^n,\hat{m}),y^n,s^n)\in\mathcal{T}_{\epsilon}^{(n)}(X,Y,S)$. According to the law of large numbers and the packing lemma the probability of error goes to zero as $n\rightarrow\infty$ if \cite{ElGamalKim},
\begin{align}
&R < \mi(X;Y|S).
\label{eq:Decoding_BM_NC_FCSI}
\end{align}


The region in Theorem~\ref{thm:Capacity_FCSI_NC} is derived from \eqref{eq:res_NC} and \eqref{eq:Decoding_BM_NC_FCSI}.

\subsection{Converse Proof}
We now develop an upper bound for the non-causal side information. Consider any sequence of length-$n$ codes for a state-dependent channel with channel state available non-causally at both the transmitter and the receiver such that $P_e^{(n)}\leq\epsilon_n$ and $\kld(P_{Z^n}||Q_0^{\otimes n})\leq\delta$ with $\lim_{n\to\infty}\epsilon_n=0$. Note that the converse is consistent with the model and does \emph{not} require $\delta$ to vanish.
\subsubsection{Epsilon Rate Region}
We first define a region $\mathcal{S}_{\epsilon}$ for $\epsilon>0$ that expands the region defined in~\eqref{eq:finalSD_FCSI_NC} as follows.
\begin{subequations}\label{eq:Epsilon_Rate_Region_FCSI_NC}
\begin{align}
\calS_\epsilon\eqdef \big\{R\geq 0: \exists P_{S,X,Y,Z}\in\calD_\epsilon: R\leq \mi(X;Y|S) + \epsilon \big\}\label{eq:Sepsilon_FCSI_NC}
\end{align}
where 
\begin{align}
  \calD_\epsilon = \left.\begin{cases}P_{S,X,Y,Z}:\\
P_{S,X,Y,Z}=Q_SP_VP_{X|V}W_{Y,Z|X,S}\\
\mathbb{D}\left(P_Z\Vert Q_0\right) \leq \epsilon\\
\ent(S|Z) >  \mi(X;Z|S) - \mi(X;Y|S) - 3\epsilon
\end{cases}\right\},\label{eq:Depsilon_FCSI_NC}
\end{align}
\end{subequations}where $\epsilon\triangleq\max\{\epsilon_n,\nu\geq\frac{\delta}{n}\}$. 
We next show that if a rate $R$ is achievable then $R\in\mathcal{S}_{\epsilon}$ for any $\epsilon>0$. 

For any $\epsilon_n>0$, we start by upper bounding $nR$ using standard techniques.
\begin{align}
nR &= \ent(M)\nonumber\\
&\mathop \le \limits^{(a)} \ent(M|S^n) - \ent(M|Y^n,S^n) + n\epsilon_n \nonumber\\
&= \mi(M;Y^n|S^n) + n\epsilon_n \nonumber\\
&= \sum\limits_{i = 1}^n \mi(M;Y_i|Y^{i - 1},S^n)  + n\epsilon_n \nonumber\\
&= \sum\limits_{i = 1}^n [\ent(Y_i|Y^{i - 1},S^n) - \ent(Y_i|M,Y^{i - 1},S^n)]  + n\epsilon_n \nonumber\\
&\mathop \le \limits^{(b)} \sum\limits_{i = 1}^n [\ent(Y_i|S_i) - \ent(Y_i|M,Y^{i - 1},S^n,X^n)]  + n\epsilon_n \nonumber\\
&= \sum\limits_{i = 1}^n [\ent(Y_i|S_i) - \ent(Y_i|S_i,X_i)]  + n\epsilon_n \nonumber\\
&= \sum\limits_{i = 1}^n \mi(X_i;Y_i|S_i)  + n\epsilon_n \nonumber\\
&\mathop \le \limits^{(c)} n\mi(\tilde{X};\tilde{Y}|\tilde{S}) + n\epsilon_n\\
&\mathop \le \limits^{(d)} n\mi(\tilde{X};\tilde{Y}|\tilde{S}) + n\epsilon
\label{eq:Upper_Bound_on_MRate_FCSI_NC_2}
\end{align}where
\begin{itemize}
    \item[$(a)$] follows from Fano's inequality and independence of $M$ from $S^n$;
    \item[$(b)$] holds because conditioning does not increase entropy;
    \item[$(c)$] follows from the concavity of mutual information, with the resulting random variables $\tilde{X}$, $\tilde{S}$, $\tilde{Y}$, and $\tilde{Z}$ to having the following distributions
\begin{subequations}\label{eq:Joint_Dist_Tilde_FCSI_NC}
\begin{align}
\tilde{P}_{X,S}(x,s) &\triangleq \frac{1}{n}\sum\limits_{i = 1}^n P_{X_i,S_i}(x,s), \label{eq:Joint_Dist_Tilde_FCSI_NC_1}\\
\tilde{P}_{X,S,Y,Z}(x,s,y,z) &\triangleq \tilde{P}_{X,S}(x,s)W_{Y,Z|X,S}(y,z|x,s);\label{eq:Joint_Dist_Tilde_FCSI_NC_2}
\end{align}
\end{subequations}
\item[$(d)$] follows by defining $\epsilon\triangleq\max\{\epsilon_n,\nu\geq\frac{\delta}{n}\}$.
\end{itemize}
We also have,
\begin{align}
nR &\geq \ent(M)\nonumber\\
&= \ent(M|S^n)\nonumber\\
&\ge \mi(M;Z^n|S^n)\nonumber\\
&\mathop = \limits^{(a)} \mi(M,X^n;Z^n|S^n)\nonumber\\
&\ge \mi(X^n;Z^n|S^n)\nonumber\\
&= \mi(X^n,S^n;Z^n) - \mi(S^n;Z^n)\nonumber\\
&= \sum\limits_{x^n} {\sum\limits_{s^n} {\sum\limits_{z^n} {P(x^n,s^n,z^n)\log \frac{W_{Z|X,S}^{\otimes n}(z^n|x^n,s^n)}{P(z^n)} - \ent(S^n) + \ent(S^n|Z^n)} } } \nonumber\\
&\geq \sum\limits_{x^n} {\sum\limits_{s^n} {\sum\limits_{z^n} {P(x^n,s^n,z^n)\log \frac{W_{Z|X,S}^{\otimes n}(z^n|x^n,s^n)}{P(z^n)} + \kld(P_{Z^n}||Q_0^{\otimes n}) - \ent(S^n) - \delta } } } \nonumber\\
&\ge \sum\limits_{i = 1}^n {\sum\limits_{x_i} {\sum\limits_{s_i} {\sum\limits_{z_i} {P(x_i,s_i,z_i)\log \frac{W_{Z|X,S}(z_i|x_i,s_i)}{Q_0(z_i)} - \sum\limits_{i = 1}^n \ent(S_i)  - \delta } } } } \nonumber\\
&= \sum\limits_{i = 1}^n {\kld(P_{X_i,S_i,Z_i}||P_{X_i,S_i}Q_0) - \sum\limits_{i = 1}^n \ent(S_i)  - \delta } \nonumber\\
&\mathop \ge \limits^{(b)} n\kld(\tilde{P}_{X,S,Z}||\tilde{P}_{X,S}Q_0) - n\ent(\tilde S) - \delta \nonumber\\
&= n\kld(\tilde{P}_{X,S,Z}||\tilde{P}_{X,S}\tilde{P}_Z) + \kld(\tilde{P}_Z||Q_0) - n\ent(\tilde S) - \delta \nonumber\\
&\mathop \ge \limits^{(c)} n\mi(\tilde{X},\tilde{S};\tilde{Z}) - n\ent(\tilde{S})- 2\delta\label{eq:Upper_Bound_on_KRate_FCSI_NC}
\end{align}where 
\begin{itemize}
    \item[$(a)$] follows because $X^n$ is a function of $(M,S^n)$;
    \item[$(b)$] follows from Jensen's inequality, the convexity of $\kld(\cdot||\cdot)$, and concavity of $\ent(\cdot)$;
    \item[$(c)$] follows from the definition of random variables $\tilde{X}$, $\tilde{S}$, $\tilde{Y}$, and $\tilde{Z}$ in \eqref{eq:Joint_Dist_Tilde_FCSI_NC}.
\end{itemize}
For any $\nu>0$, choosing $n$ large enough substituting \eqref{eq:MI_Manipulation} into \eqref{eq:Upper_Bound_on_KRate_FCSI_NC} ensures that
\begin{align}
    nR&\ge\mi(\tilde{X};\tilde{Z}|\tilde{S}) -\ent(\tilde{S}|\tilde{Z})-2\nu,\nonumber\\
    &\ge\mi(\tilde{X};\tilde{Z}|\tilde{S}) -\ent(\tilde{S}|\tilde{Z})-2\epsilon,
    \label{eq:Upper_Bound_on_KRate_FCSI_NC_Final}
\end{align}where the last equality follows from the definition of $\epsilon\triangleq\max\{\epsilon_n,\nu\}$. 
To show that $ \kld(P_Z||Q_0)\leq\epsilon$, note that for $n$ large enough
\begin{align}
\kld(P_Z||Q_0)=\kld(P_{\tilde{Z}}||Q_0)=\kld\Bigg(\frac{1}{n}\sum\limits_{i=1}^nP_{Z_i}\Bigg|\Bigg|Q_0\Bigg)\leq\frac{1}{n}\sum\limits_{i=1}^n\kld(P_{Z_i}||Q_0)\leq\frac{1}{n}\kld(P_{Z^n}||Q_0^{\otimes n})\leq\frac{\delta}{n}\leq\nu\leq\epsilon.\label{eq:boundKL_FCSI_NC}
\end{align}
Combining \eqref{eq:Upper_Bound_on_MRate_FCSI_NC_2} and \eqref{eq:Upper_Bound_on_KRate_FCSI_NC_Final} shows that $\forall \epsilon_n,\nu>0$, $R\leq \max\{x:x\in\mathcal{S}_{\epsilon}\}$. Therefore,
\begin{align}
  C_{\mbox{\scriptsize\rm NC-TR}} = \max\left\{x:x\in\bigcap_{\epsilon>0}\mathcal{S}_{\epsilon}\right\}.
\end{align}
\subsubsection{Continuity at Zero}
One can prove the continuity at zero of $\calS_\epsilon$ by substituting $\min\{\mi(U;Y)-\mi(U;S),\mi(U,V;Y)-\mi(U;S|V)\}$ with $\mi(X;Y|S)$ and $\mi(V;Z)-\mi(V;S)$ with $\mi(X;Z|S)-\ent(S|Z)$ in Appendix~\ref{sec:continuity-at-zero} and following the exact same arguments.

\section{Proof of Theorem~\ref{thm:Capacity_FCSI_C}}
\label{sec:Proof_Acivable_Rate_CoNC_State}
\subsection{Achievability Proof}
To prove the achievability of Theorem~\ref{thm:Capacity_FCSI_C} it is
convenient to introduce an associated channel $W_{Y,Z|U,S}$ as follows:
Let $U\in\mathcal{U}$ be an arbitrary auxiliary random variable which is independent of the state $S$, and let $x: U\times S\to X$ be a deterministic mapping subject to
$\indic{1}_{\{x=x(s,u)\}}$. According to the Shannon
strategy \cite{Shannon58state}, we define the $W_{Y,Z|U,S}$ as a channel specified by
\begin{align}
    W_{Y,Z|U,S} =\sum\limits_{x\in\mathcal{X}}\indic{1}_{\{x=x(s,u)\}}W_{Y,Z|X,S}(y,z|x,s),
\end{align}
which results in a channel with input $U$, outputs $Y$, $Z$, and state $S$. Therefore, we only focus on the coding problem for the channel $W_{Y,Z|U,S}$ for the achievability proof.

We use block-Markov coding in which $B$ independent messages are transmitted over $B$ channel blocks, each of length $r$, therefore the overall codeword length is $n=rB$ symbols. The warden's observation $Z^n$ can be described in terms of observations in individual block-Markov blocks  $Z^n=(Z_1^r,\dots,Z_B^r)$. The distribution of the warden's observation, induced by the block-Markov coding, is $P_{Z^n}\triangleq P_{Z_1^r,\dots,Z_B^r}$ and the target output distribution is $Q_Z^{\otimes n}=\prod\nolimits_{j = 1}^B Q_Z^{\otimes r}$. 
\begin{align}
\kld(P_{Z^n}||Q_Z^{\otimes n}) &= \kld(P_{Z_1^r...Z_B^r}||Q_Z^{\otimes rB}) \nonumber\\
&= \sum\limits_{j=1}^B {\kld(P_{Z_j^r|Z_{j+1}^{B,r}}||Q_Z^{\otimes r}|P_{Z_{j+1}^{B,r}})}  \nonumber\\ 
&= \sum\limits_{j = 1}^B {[\kld(P_{Z_j^r}||Q_Z^{\otimes r}) + \kld(P_{Z_j^r|Z_{j+1}^{B,r}}||P_{Z_j^r}|P_{Z_{j + 1}^{B,r}})]}  \nonumber\\ 
&= \sum\limits_{j = 1}^B {[\kld(P_{Z_j^r}||Q_Z^{\otimes r}) + \mi(Z_j^r;Z_{j+1}^{B,r})]},  \label{eq:Total_KLD_CoNC}
\end{align}
where $Z_{j + 1}^{B,r}=\{Z_{j + 1}^r,\dots Z_B^r\}$. Hence, $\kld(P_{Z^n}||Q_Z^{\otimes n})\underset{n\rightarrow\infty}{\xrightarrow{\hspace{0.2in}}} 0$, is equivalent to;
\begin{align}
\kld(P_{Z_j^r}||Q_Z^{\otimes r})\underset{r\rightarrow\infty}{\xrightarrow{\hspace{0.2in}}} 0, \qquad \mi(Z_j^r;Z_{j+1}^{B,r})\underset{r\rightarrow\infty}{\xrightarrow{\hspace{0.2in}}} 0, \qquad\qquad \forall j\in\intseq{1}{B}.
\end{align}
This requires constructing a code that approximates $Q_Z^{\otimes r}$ in each block, while eliminating the dependencies across blocks created by block-Markov coding. The random code generation is as follows:

Fix $P_U(u)$ and $\epsilon_1>\epsilon_2>0$ such that, $P_Z = Q_0$.

\emph{{Codebook Generation for Keys:}}
For each block $j\in\intseq{1}{B}$, create a function $\Phi:S_j^r\to\intseq{1}{2^{rR_K}}$ through random binning by choosing the value of $\Phi(s_j^r)$ independently and uniformly at random for every $s_j^r\in\mathcal{S}^r$. The key $k_j=\Phi(s_j^r)$ obtained in block $j\in\intseq{1}{B}$ from the state sequence $s_j^r$ is used to assist the encoder in the next block.

\emph{Codebook Generation for Messages:}
For each block $j\in\intseq{1}{B}$, let $C_r\triangleq\{U^r(m_j,k_{j-1})\}_{(m_j,k_{j-1})\in\mathcal{M}\times \mathcal{K}}$, where $\mathcal{M}\in\intseq{1}{2^{rR}}$ and $\mathcal{K}\in\intseq{1}{2^{rR_k}}$, be a random codebook consisting of independent random sequences each generated according to $P_U^{\otimes r}$. We denote a realization of $C_r$ by $\CodeBook_r\triangleq\{u^r(m_j,k_{j-1})\}_{(m_j,k_{j-1})\in\mathcal{M}\times\mathcal{K}}$.

\emph{{Encoding:}}
For the first block, we assume that the transmitter and the receiver have access to a shared secret key $k_0$, in this block to transmit $m_1$ the encoder computes $u^r(m_1,k_0)$ and transmits codeword $x^r$, where $x_i=x(u_i(m_1,k_0),s_i)$. At the end of the first block the encoder generates a key from \ac{CSI} $s_1^r$ to be used in Block 2.

For block $j\in\intseq{2}{B}$, to send message $m_j$ according to the generated key $k_{j-1}$ from the previous block, the encoder computes $u^r(m_j,k_{j-1})$ and transmits codeword $x^r$, where $x_i=x(u_i(m_j,k_{j-1}),s_i)$. Also, at the end of each  block $j\in\intseq{2}{B}$ the encoder generates a key from \ac{CSI} $s_j^r$ to be used in next block.

Define
\begin{align}
    \Upsilon _{M_j,K_{j-1},U^r,S_j^r,Z_j^r,K_j}^{(\mathcal{C}_r)}(m_j,k_{j-1},\tilde{u}^r,s_j^r,z_j^r,k_j) \triangleq& 2^{-r(R_k + R)}\indic{1}_{\{ \tilde{u}^r=u^r(m_j,k_{j-1})\}}Q_S^{\otimes r}(s_j^r)\nonumber\\
    &\times W_{Z|U,S}^{\otimes r}(z_j^r|\tilde{u}^r,s_j^r)\indic{1}_{\{ k_j=\Phi (\tilde s_j^r)\}}.\label{eq:P_Dist_CoN}
\end{align}
For a fixed codebook $\mathcal{C}_r$, the induced joint distribution over the codebook (i.e. $P^{(\mathcal{C}_r)}$) satisfies
\begin{align}
\kld\Big(P_{M_j,K_{j-1},U^r,S_j^r,Z_j^r,K_j}^{(\mathcal{C}_r)}||\Upsilon_{M_j,K_{j-1},U^r,S_j^r,Z_j^r,K_j}^{(\mathcal{C}_r)}\Big)\leq\epsilon.
\label{eq:Uniformity_Key}
\end{align} 
This intermediate distribution $\Upsilon^{(\mathcal{C}_r)}$ approximates the true distribution $P^{(\mathcal{C}_r)}$ and will be used in the sequel for bounding purposes. Expression~\eqref{eq:Uniformity_Key} holds because the main difference in $\Upsilon^{(\mathcal{C}_r)}$ is assuming the key $K_{j-1}$ is uniformly distributed, which is made (arbitrarily) nearly uniform in $P^{(\mathcal{C}_r)}$ with appropriate control of rate.

\emph{Covert Analysis:} We now show  $\expec_{C_n}[\kld(P_{Z^n|C_n} || Q_Z^{\otimes n} )]\underset{n\rightarrow\infty}{\xrightarrow{\hspace{0.2in}}} 0$, where $C_n$ is the set of all codebooks from all blocks, and then choose $P_U$ and $x(u,s)$ such that it satisfies $Q_Z=Q_0$.
\begin{figure*}
\centering
\includegraphics[width=5.75in]{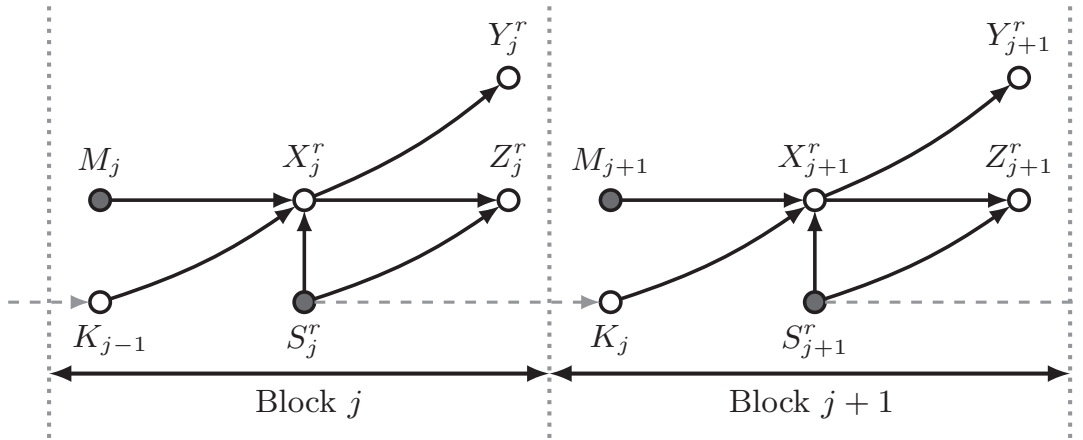}
\caption{Functional dependence graph for the block-Markov encoding scheme}
\label{fig:Chaining_CoN}
\end{figure*}%
From the expansion in \eqref{eq:Total_KLD_CoNC},  for every block~$j\in\intseq{2}{B}$,
\begin{align}
 \mi(Z_j^r;Z_{j+1}^{B,r}) &\leq \mi(Z_j^r;K_j,Z_{j+1}^{B,r}) \nonumber\\
 &\mathop=\limits^{(a)} \mi(Z_j^r;K_j), \label{eq:Indep_Across_Blocks_CoNC}
\end{align}
where $(a)$ holds because $Z_j^r-K_j-Z_{j+1}^{B,r}$ forms a Markov chain, as seen in the functional dependence graph depicted in Fig.~\ref{fig:Chaining_CoN}. Also,
\begin{align}
    \mi(Z_j^r;K_j) &= \kld(P_{Z_j^r,K_j}^{(\mathcal{C}_n)}||P_{Z_j^r}^{(\mathcal{C}_n)}P_{K_j}^{(\mathcal{C}_n)})\nonumber\\
  &\mathop\leq\limits^{(b)} \kld(P_{Z_j^r,K_j}^{(\mathcal{C}_n)}||Q_Z^{\otimes r}Q_{K_j}),\label{eq:KLD_Z3_CoNC}
\end{align}
where $Q_{K_j}$ is the uniform distribution on $\intseq{1}{2^{rR_K}}$
and $(b)$ follows from the positivity of relative entropy and
\begin{align}
  \kld(P_{Z_j^r,K_j}||P_{Z_j^r}P_{K_j}) &= \kld(P_{Z_j^r,K_j}||Q_Z^{\otimes r}Q_{K_j}) \twocolbreak - \kld(P_{K_j}||Q_{K_j}) - \kld(P_{Z_j^r}||Q_Z^{\otimes r}).
  \label{eq:KLD_Property_CoNC}
\end{align} 
Therefore by combining \eqref{eq:Total_KLD_CoNC}, \eqref{eq:KLD_Z3_CoNC}, and \eqref{eq:KLD_Property_CoNC} 
\begin{align}
    \kld(P_{Z^n|C_n}||Q_Z^{\otimes n})\leq 2\sum\limits_{j = 1}^B {\kld(P_{Z_j^r,K_j|C_r}||Q_Z^{\otimes r}Q _{K_j})}.\label{eq:Total_KLD_Bound}
\end{align}
We now proceed to bound the right hand side of \eqref{eq:Total_KLD_Bound}. First, consider the marginal
\begin{align}
    \Upsilon_{Z_j^r,K_j|C_r} (z_j^r,k_j)=\sum\limits_{m_j}\sum\limits_{k_{j-1}}\sum\limits_{s_j^r}   \frac{1}{2^{r(R + R_k)}}\twocolbreak Q_S^{\otimes r}(s_j^r) W_{Z|U,S}^{\otimes r}\big(z_j^r|U^r(m_j,k_{j-1}),s_j^r\big) \indic{1}_{\{ k_j=\Phi (s_j^r)\}}.\label{eq:Joint_Dist_ZK_CoNC}  
\end{align}
From \eqref{eq:Uniformity_Key} and the monotonicity of KL-divergence we have,
\begin{align}
    \kld\big(\Upsilon_{Z_j^r,K_j|C_r}||P_{Z_j^r,K_j|C_r}\big)\leq\epsilon.\label{eq:Marginal_Up_P_FCSI_C}
\end{align}
To bound the \ac{RHS} of \eqref{eq:Total_KLD_Bound} by using Lemma~\ref{lemma:TV_KLD} and the triangle inequality we have
\begin{align}
     \expec_{C_r}||P_{Z_j^r,K_j|C_r}-Q_Z^{\otimes r} Q_{K_j}||_1\leq \expec_{C_r}||P_{Z_j^r,K_j|C_r}-\Upsilon_{Z_j^r,K_j|C_r}||_1 + \expec_{C_r}||\Upsilon_{Z_j^r,K_j|C_r}-Q_Z^{\otimes r} Q_{K_j}||_1.\label{eq:Triangle_with_Real_Dist_1}
\end{align}
From Lemma~\ref{lemma:TV_KLD} and \eqref{eq:Marginal_Up_P_FCSI_C} the first term on the \ac{RHS} of \eqref{eq:Triangle_with_Real_Dist_1} vanishes as $r$ grows; to bound the second term by using Lemma~\ref{lemma:TV_KLD} we have,
\begin{align}
&\expec_{C_r}[\kld(\Upsilon_{Z_j^r,K_j|C_r}||Q_Z^{\otimes r} Q_{K_j})] = \expec_{C_r}\Big[\sum\limits_{z_j^r,k_j} \Upsilon_{Z_j^r,K_j|C_r} (z_j^r,k_j)\log\Big(\frac{\Upsilon_{Z_j^r,K_j|C_r}(z_j^r,k_j)}{Q_Z^{\otimes r}(z_j^r)Q_{K_j}(k_j)}\Big) \Big] \nonumber\\ 
&= \expec_{C_r}\Bigg[\sum\limits_{(z_j^r,k_j)} \sum\limits_{m_j}\sum\limits_{k_{j-1}}  \frac{1}{2^{r(R + R_k)}}\sum\limits_{s_j^r} Q_S^{\otimes r}(s_j^r)W_{Z|U,S}^{\otimes r}\big(z_j^r|U^r(m_j,k_{j-1}),s_j^r\big)\indic{1}_{\{ k_j=\Phi (s_j^r)\}} \nonumber\\ 
&\qquad\times\log \Bigg(\frac{ \sum\limits_{\tilde{m}_j}\sum\limits_{\tilde{k}_{j-1}}\sum\limits_{\tilde{s}_j^r}   Q_S^{\otimes r}(\tilde s_j^r) W_{Z|U,S}^{\otimes r}\big(z_j^r|U^r(\tilde{m}_j,\tilde{k}_{j-1}),\tilde s_j^r\big)\indic{1}_{\{ k_j=\Phi (\tilde s_j^r)\}}}{2^{r(R+ R_k - R_K)}Q_Z^{\otimes r}(z_j^r)}\Bigg)     \Bigg] \nonumber\\ 
&\mathop \le \limits^{(a)} \sum\limits_{(z_j^r,k_j)} \sum\limits_{m_j}\sum\limits_{k_{j-1}}  \frac{1}{2^{r(R+ R_k)}}\sum\limits_{s_j^r} \sum\limits_{u^r(m_j,k_{j-1})} \Upsilon^{(\mathcal{C}_r)}_{U^r,S^r,Z^r}\big(u^r(m_j,k_{j-1}),s_j^r,z_j^r\big)\expec_{\Phi (s_j^r)}\big[\indic{1}_{\{ k_j=\Phi (s_j^r)\}}\big] \nonumber\\
&\qquad\times\log \expec_{\backslash ((m_j,k_{j-1}),\Phi (s_j^r))} \Bigg[ \frac{\sum\limits_{\tilde m_j}\sum\limits_{\tilde k_{j-1}}\sum\limits_{\tilde s_j^r} Q_S^{\otimes r}(\tilde s_j^r) W_{Z|U,S}^{\otimes r}\big(z_j^r|U^r(\tilde{m}_j,\tilde{k}_{j-1}),\tilde s_j^r\big)\indic{1}_{\{ k_j=\Phi (\tilde s_j^r)\}}     }{{{2^{r(R+ R_k - R_K)}}Q_Z^{\otimes r}(z_j^r)}}\Bigg]        \nonumber\\
&\mathop \le \limits^{(b)}\sum\limits_{(z_j^r,k_j)} \sum\limits_{m_j}\sum\limits_{k_{j-1}} \frac{1}{2^{r(R+ R_k)}}\sum\limits_{s_j^r} \sum\limits_{u^r(m_j,k_{j-1})} \Upsilon^{(\mathcal{C}_r)}_{U^r,S^r,Z^r}\big(u^r(m_j,k_{j-1},s_j^r,z_j^r)\big) \frac{1}{2^{rR_K}}  \nonumber\\
&\qquad\times\log\frac{1}{2^{r(R+ R_k - R_K)}Q_Z^{\otimes r}(z_j^r)} \Bigg(Q_S^{\otimes r}(s_j^r) W_{Z|U,S}^{\otimes r}\big(z_j^r|u^r(m_j,k_{j-1}),s_j^r\big) \nonumber\\
&\qquad + \expec_{\backslash (m_j,k_{j-1})} \Bigg[\sum\limits_{(\tilde{m}_j,\tilde{k}_{j-1}) \ne (m_j,k_{j-1})}Q_S^{\otimes r}(s_j^r) W_{Z|U,S}^{\otimes r}\big(z_j^r|U^r(\tilde{m}_j,\tilde{k}_{j-1}),s_j^r\big)\Bigg] \nonumber\\
&\qquad+ \expec_{\backslash \Phi (s_j^r)}\Bigg[ \sum\limits_{\tilde s_j^r}Q_S^{\otimes r}(\tilde s_j^r) W_{Z|U,S}^{\otimes r}\big(z_j^r|u^r(m_j,k_{j-1}),\tilde s_j^r\big)\indic{1}_{\{ k_j=\Phi (\tilde s_j^r)\}}\Bigg]        \nonumber\\
&\qquad+ \expec_{\backslash ((m_j,k_{j-1}),\Phi (s_j^r))}\Bigg[ \sum\limits_{\tilde s_j^r}\sum\limits_{(\tilde{m}_j,\tilde{k}_{j-1}) \ne (m_j,k_{j-1})}Q_S^{\otimes r}(\tilde s_j^r) W_{Z|U,S}^{\otimes r}\big(z_j^r|U^r(\tilde{m}_j,\tilde{k}_{j-1}),\tilde s_j^r\big)\indic{1}_{\{ k_j=\Phi (\tilde s_j^r)\}}\Bigg] \Bigg)        \nonumber\\
&\le\sum\limits_{(z_j^r,k_j)} \sum\limits_{m_j}\sum\limits_{k_{j-1}} \frac{1}{2^{r(R+ R_k+R_K)}}\sum\limits_{s_j^r} \sum\limits_{u^r(m_j,k_{j-1})} \Upsilon^{(\mathcal{C}_r)}_{U^r,S^r,Z^r}\big(u^r(m_j,k_{j-1}),s_j^r,z_j^r\big)  \nonumber\\
&\qquad\times\log \Bigg(\frac{Q_S^{\otimes r}(s_j^r) W_{Z|U,S}^{\otimes r}\big(z_j^r|u^r(m_j,k_{j-1}),s_j^r\big)}{2^{r(R+ R_k - R_K)}Q_Z^{\otimes r}(z_j^r)} + \sum\limits_{(\tilde{k}_{j - 1},\tilde{m}_j) \ne (m_j,k_{j-1})} \frac{Q_{S,Z}^{\otimes r}(s_j^r,z_j^r)}{2^{r(R+ R_k - R_K)}Q_Z^{\otimes r}(z_j^r)}\nonumber\\
&\qquad +\frac{ W_{Z|U}^{\otimes r}\big(z_j^r|u^r(m_j,k_{j-1})\big)}{2^{r(R+ R_k)}Q_Z^{\otimes r}(z_j^r)}
+1 \Bigg)\nonumber\\
&\le\sum\limits_{(z_j^r,k_j)} \sum\limits_{m_j}\sum\limits_{k_{j-1}} \frac{1}{2^{r(R+ R_k+R_K)}}\sum\limits_{s_j^r} \sum\limits_{u^r(m_j,k_{j-1})} \Upsilon^{(\mathcal{C}_r)}_{U^r,S^r,Z^r}\big(u^r(m_j,k_{j-1}),s_j^r,z_j^r\big)  \nonumber\\
&\qquad\times\log \Bigg(\frac{Q_S^{\otimes r}(s_j^r) W_{Z|U,S}^{\otimes r}\big(z_j^r|u^r(m_j,k_{j-1}),s_j^r\big)}{2^{r(R+ R_k - R_K)}Q_Z^{\otimes r}(z_j^r)} +  \frac{2^{rR_K} Q_{S,Z}^{\otimes r}(s_j^r,z_j^r)}{Q_Z^{\otimes r}(z_j^r)} + \frac{ W_{Z|U}^{\otimes r}\big(z_j^r|u^r(m_j,k_{j-1})\big)}{2^{r(R+ R_k)}Q_Z^{\otimes r}(z_j^r)} + 1\Bigg)  \nonumber\\
&\triangleq \Psi_1 + \Psi_2,\label{eq:KLD_ZK_CoNC}
\end{align}
where $(a)$ follows from Jensen's inequality and $(b)$ holds because $\indic{1}_{\{ \cdot\}}\leq 1$. We defined $\Psi_1$ and $\Psi_2$ as
\begin{align}
&{\Psi _1} = \sum\limits_{k_j} \sum\limits_{m_j}\sum\limits_{k_{j-1}}  \frac{1}{2^{r(R+ R_k + R_K)}}\sum\limits_{\big(u^r(m_j,k_{j-1}),s_j^r,z_j^r\big) \in \mathcal{T}_\epsilon ^{(n)}} \Upsilon^{(\mathcal{C}_r)}_{U^r,S^r,Z^r}\big(u^r(m_j,k_{j-1}),s_j^r,z_j^r\big)\nonumber\\
&\qquad\times
 \log \Bigg(\frac{Q_S^{\otimes r}(s_j^r) W_{Z|U,S}^{\otimes r}\big(z_j^r|u^r(m_j,k_{j-1}),s_j^r\big)}{2^{r(R+ R_k - R_K)}Q_Z^{\otimes r}(z_j^r)} +  \frac{2^{rR_K} Q_{S,Z}^{\otimes r}(s_j^r,z_j^r)}{Q_Z^{\otimes r}(z_j^r)} + \frac{ W_{Z|U}^{\otimes r}\big(z_j^r|u^r(m_j,k_{j-1})\big)}{2^{r(R+ R_k)}Q_Z^{\otimes r}(z_j^r)} + 1\Bigg)  \nonumber\\
 & \le \log \Bigg(\frac{2^{rR_K}2^{ - r(1 - \epsilon )(\ent(S)+\ent(Z|U,S))}}{2^{r(R+ R_k)}2^{ - r(1 + \epsilon )\ent(Z)}} + \frac{2^{rR_K}2^{ - r(1 - \epsilon )\ent(S,Z)}}{2^{ - r(1 + \epsilon )\ent(Z)}} + \frac{2^{ - r(1 - \epsilon )\ent(Z|U)}}{2^{r(R+ R_k)}2^{ - r(1 + \epsilon )\ent(Z)}}  + 1\Bigg)\label{eq:Psi_1_CoNC}\\
&{\Psi _2} = \sum\limits_{{k_j}} \sum\limits_{m_j}\sum\limits_{k_{j-1}}  \frac{1}{2^{r(R+ R_k + R_K)}}\sum\limits_{\big(u^r(m_j,k_{j-1}),s_j^r,z_j^r\big) \notin \mathcal{T}_\epsilon ^{(n)}} \Upsilon^{(\mathcal{C}_r)}_{U^r,S^r,Z^r}\big(u^r(m_j,k_{j-1}),s_j^r,z_j^r\big)\nonumber\\
&\qquad\times
 \log \Bigg(\frac{Q_S^{\otimes r}(s_j^r) W_{Z|U,S}^{\otimes r}\big(z_j^r|u^r(m_j,k_{j-1}),s_j^r\big)}{2^{r(R+ R_k - R_K)}Q_Z^{\otimes r}(z_j^r)} +  \frac{2^{rR_K} Q_{S,Z}^{\otimes r}(s_j^r,z_j^r)}{Q_Z^{\otimes r}(z_j^r)} + \frac{ W_{Z|U}^{\otimes r}\big(z_j^r|u^r(m_j,k_{j-1})\big)}{2^{r(R+ R_k)}Q_Z^{\otimes r}(z_j^r)} + 1\Bigg)  \nonumber\\
&\le 2|U||S||Z|e^{ - r\epsilon^2 {\mu _{U,S,Z}}}r\log\Big(\frac{3}{\mu_Z} + 1\Big), \label{eq:Psi_2_CoNC}
\end{align}
\sloppy where $\mu_{U,S,Z}=\min\limits_{(u,s,z)\in(\mathcal{U},\mathcal{S},\mathcal{Z})}P_{U,S,Z}(u,s,z)$ and $\mu_Z=\min\limits_{z\in\mathcal{Z}}P_Z(z)$. When $r\to\infty$ then $\Psi_2\to 0$, by choosing $R_K=\ent(S|Z)-\epsilon$, $\Psi_1$ vanishes when $r$ grows if,
\begin{subequations}
\begin{align}
    R + R_k &>\mi(U;Z|S),\label{eq:Resolv_U_Against_S_and_Z_CoNC}\\
    R + R_k &>\mi(U;Z).\label{eq:Resolv_Redun_CoNC}
\end{align}
\end{subequations}Here since $U$ and $S$ are independent \eqref{eq:Resolv_Redun_CoNC} is redundant because of \eqref{eq:Resolv_U_Against_S_and_Z_CoNC}.

\emph{{Decoding and Error Probability Analysis:}}
At the end of block $j\in\intseq{1}{B}$, using its knowledge of the channel state $s_j^r$ of the current block and the key $k_{j-1}$ generated from the previous block, the receiver finds a unique  $\hat{m}_j$ such that $\big(u^r(\hat{m}_j,k_{j-1}),s_j^r,y_j^r\big)\in\mathcal{T}_{\epsilon}^{(r)}$. To analyze the probability of error we define the following error events for $j\in\intseq{1}{B}$
\begin{subequations}\label{eq:error_Events_FCSI_C}
\begin{align}
    \mathcal{E}&=\big\{\hat{M}\ne M\big\},\\
    \mathcal{E}_j&=\big\{\hat{M}_j\ne M_j\big\},\\
    \mathcal{E}_{1,j}&=\big\{\big(U^r(M_j,K_{j-1}),S_j^r\big)\notin\mathcal{T}_{\epsilon_1}^{(r)}(Q_SP_U)\big\},\\
    \mathcal{E}_{2,j}&=\big\{\big(U^r(M_j,K_{j-1}),S_j^r,Y_j^r\big)\notin\mathcal{T}_{\epsilon_2}^{(r)}(Q_SP_UW_{Y|U,S})\big\},\\
    \mathcal{E}_{3,j}&=\big\{\big(U^r(k_{j-1},\hat{m}_j),S_j^r,Y_j^r\big)\in\mathcal{T}_{\epsilon_2}^{(r)},\quad\mbox{for some}\quad\hat{m}_j\ne M_j\big\},
\end{align}
\end{subequations}where $\epsilon_2>\epsilon_1>\epsilon>0$. The probability of error is upper bounded as follows,
\begin{align}
    \Prob(\mathcal{E})\leq\Prob\Big\{\bigcup\nolimits_{j = 1}^B \mathcal{E}_j\Big\}\leq\sum_{j=1}^B\Prob(\mathcal{E}_j).
\end{align}Now we bound $\Prob(\mathcal{E}_j)$ by using union bound
\begin{align}
    \Prob(\mathcal{E}_j)\leq \Prob(\mathcal{E}_{1,j})+\Prob(\mathcal{E}_{1,j}^c\cap \mathcal{E}_{2,j})+\Prob(\mathcal{E}_{2,j}^c\cap \mathcal{E}_{3,j}).\label{eq:PError}
\end{align}By the law of large numbers the first and second term on \ac{RHS} of \eqref{eq:PError} vanishes when $r$ grows. According to the law of large numbers and the packing lemma, the last term on \ac{RHS} of \eqref{eq:PError} vanishes when $r$ grows if \cite{ElGamalKim},
\begin{align}
&R < \mi(U;S,Y)=\mi(U;Y|S).
\label{eq:Decoding_BM_CoNC}
\end{align}

Furthermore, this scheme requires that 
\begin{align}
    R_k\leq R_K = \ent(S|Z) - \epsilon,\label{eq:Indep_S_and_Z_CoNC}
\end{align}
The region in Theorem~\ref{thm:Capacity_FCSI_C} is obtained by applying Fourier-Motzkin to \eqref{eq:Resolv_U_Against_S_and_Z_CoNC}, \eqref{eq:Decoding_BM_CoNC}, and  \eqref{eq:Indep_S_and_Z_CoNC}.
\subsection{Converse Proof}
We now develop an upper bound when \ac{CSI} is available causally at both of the legitimate terminals. Consider any sequence of length-$n$ codes for a state-dependent channel with channel state available causally at both the transmitter and the receiver such that $P_e^{(n)}\leq\epsilon_n$ and $\kld(P_{Z^n}||Q_0^{\otimes n})\leq\delta$ with $\lim_{n\to\infty}\epsilon_n=0$. Note that the converse is consistent with the model and does \emph{not} require $\delta$ to vanish. 
\subsubsection{Epsilon Rate Region}
We first define a region $\mathcal{S}_{\epsilon}$ for $\epsilon>0$ that expands the region defined in~\eqref{eq:finalSD_FCSI_C} as follows.
\begin{subequations}\label{eq:Epsilon_Rate_Region_FCSI_C}
\begin{align}
\calS_\epsilon\eqdef \big\{R\geq 0: \exists P_{S,U,X,Y,Z}\in\calD_\epsilon: R\leq \mi(U;Y|S) + \epsilon \big\}\label{eq:Sepsilon_FCSI_C}
\end{align}
where 
\begin{align}
  \calD_\epsilon = \left.\begin{cases}P_{S,U,X,Y,Z}:\\
P_{S,U,X,Y,Z}=Q_SP_U\indic{1}_{\big\{X=X(U,S)\big\}}W_{Y,Z|X,S}\\
\mathbb{D}\left(P_Z\Vert Q_0\right) \leq \epsilon\\
\ent(S|Z) >  \mi(U;Z|S) - \mi(U;Y|S) - 3\epsilon\\
\card{\calU}\leq \card{\calX}+1
\end{cases}\right\},\label{eq:Depsilon_FCSI_C}
\end{align}
\end{subequations}where $\epsilon\triangleq\max\{\epsilon_n,\nu\geq\frac{\delta}{n}\}$. 
We next show that if a rate $R$ is achievable then $R\in\mathcal{S}_{\epsilon}$ for any $\epsilon>0$. 

For any $\epsilon_n>0$, we start by upper bounding $nR$ using standard techniques.
\begin{align}
nR &= \ent(M)\nonumber\\
&\mathop \le \limits^{(a)} \ent(M|S^n) - \ent(M|Y^n,S^n) + n\epsilon_n \nonumber\\
&= \mi(M;Y^n|S^n) + n\epsilon_n \nonumber\\
&= \sum\limits_{i = 1}^n \mi(M;Y_i|Y^{i - 1},S^n)  + n\epsilon_n \nonumber\\
&= \sum\limits_{i = 1}^n [\ent(Y_i|Y^{i - 1},S^n) - \ent(Y_i|M,Y^{i - 1},S^n)]  + n\epsilon_n \nonumber\\
&\mathop \le \limits^{(b)} \sum\limits_{i = 1}^n [\ent(Y_i|S_i) - \ent(Y_i|U_i,S_i)]  + n\epsilon_n \nonumber\\
&= \sum\limits_{i = 1}^n \mi(U_i;Y_i|S_i)  + n\epsilon_n \nonumber\\
&\mathop \le \limits^{(c)} n\mi(\tilde{U};\tilde{Y}|\tilde{S}) + n\epsilon_n\nonumber\\
&\mathop \le \limits^{(d)} n\mi(\tilde{U};\tilde{Y}|\tilde{S}) + n\epsilon
\label{eq:Upper_Bound_on_MRate_FCSI_C_2}
\end{align}where 
\begin{itemize}
    \item[$(a)$] follows from Fano's inequality and independence of $M$ from $S^n$;
    \item[$(b)$] holds because conditioning does not increase entropy and $U_i=(M,Y^{i - 1},S_{\sim i}^n)$;
    \item[$(c)$] follows from the concavity of mutual information, with the resulting random variables $\tilde{U}$, $\tilde{S}$, and $\tilde{Y}$ having the following distributions
\begin{subequations}\label{eq:Joint_Dist_Tilde_FCSI_C}
\begin{align}
\tilde{P}_{U,S,X}(u,s,x) &\triangleq \frac{1}{n}\sum\limits_{i = 1}^n P_{U_i,S_i,X_i}(u,s,x), \label{eq:Joint_Dist_Tilde_FCSI_C_1}\\
\tilde{P}_{U,S,X,Y,Z}(u,s,x,y,z) &\triangleq \tilde{P}_{U,S,X}(u,s,x)W_{Y,Z|X,S}(y,z|x,s),\label{eq:Joint_Dist_Tilde_FCSI_C_2}
\end{align}
\end{subequations}
\item[$(d)$] follows by defining $\epsilon\triangleq\max\{\epsilon_n,\nu\geq\frac{\delta}{n}\}$.
\end{itemize}

We now have,
\begin{align}
    nR&\mathop \ge \limits^{(a)} n\mi(\tilde{X},\tilde{S};\tilde{Z}) - n\ent(\tilde{S})-2\epsilon\nonumber\\
    &\mathop \ge \limits^{(b)} n\mi(\tilde{U},\tilde{S};\tilde{Z}) - n\ent(\tilde{S})-2\epsilon,\label{eq:Upper_Bound_on_Key_FCSI_C}
\end{align}where $(a)$ follows by the steps in \eqref{eq:Upper_Bound_on_KRate_FCSI_NC_Final} and $(b)$ follows from the Markov chain $U-(X,S)-Z$ and from the definition of random variables $\tilde{U}$, $\tilde{X}$, $\tilde{S}$, $\tilde{Y}$, and $\tilde{Z}$ in \eqref{eq:Joint_Dist_Tilde_FCSI_C}. 
Rewriting the bound in \eqref{eq:Upper_Bound_on_Key_FCSI_C} by using the basic property in \eqref{eq:MI_Manipulation} leads to 
\begin{align}
    nR\ge\mi(\tilde{U};\tilde{Z}|\tilde{S}) -\ent(\tilde{S}|\tilde{Z})-2\epsilon.\label{eq:Upper_Bound_on_KRate_CoNC_Final}
\end{align}
To show that $ \kld(P_Z||Q_0)\leq\epsilon$, note that for $n$ large enough
\begin{align}
\kld(P_Z||Q_0)=\kld(P_{\tilde{Z}}||Q_0)=\kld\Bigg(\frac{1}{n}\sum\limits_{i=1}^nP_{Z_i}\Bigg|\Bigg|Q_0\Bigg)\leq\frac{1}{n}\sum\limits_{i=1}^n\kld(P_{Z_i}||Q_0)\leq\frac{1}{n}\kld(P_{Z^n}||Q_0^{\otimes n})\leq\frac{\delta}{n}\leq\nu\leq\epsilon.\label{eq:boundKL_FCSI_C}
\end{align}
Combining \eqref{eq:Upper_Bound_on_MRate_FCSI_C_2} and \eqref{eq:Upper_Bound_on_KRate_CoNC_Final} shows that $\forall \epsilon_n,\nu>0$, $R\leq \max\{x:x\in\mathcal{S}_{\epsilon}\}$. Therefore,
\begin{align}
  C_{\mbox{\scriptsize\rm C-TR}} = \max\left\{x:x\in\bigcap_{\epsilon>0}\mathcal{S}_{\epsilon}\right\}.
\end{align}
\subsubsection{Continuity at Zero}
Continuity at zero for $\calS_\epsilon$ is established by substituting $\min\{\mi(U;Y)-\mi(U;S),\mi(U,V;Y)-\mi(U;S|V)\}$ with $\mi(U;Y|S)$ and $\mi(V;Z)-\mi(V;S)$ with $\mi(U;Z|S)-\ent(S|Z)$ in Appendix~\ref{sec:continuity-at-zero} and following the same arguments.


\section{Proof of Theorem~\ref{thm:Capacity_FCSI_SC}}
\label{sec:Proof_Acivable_Rate_Strictly_Causal_State_Info}
\subsection{Achievability Proof}
We adopt a block-Markov encoding scheme in which $B$ independent messages are transmitted over $B$ channel blocks each of length $r$, such that $n=rB$.
The warden's observation is $Z^n=(Z_1^r,\dots,Z_B^r)$, the target output distribution is $Q_Z^{\otimes n}$, and Equation \eqref{eq:Total_KLD_CoNC}, describing the distance between the two distributions, continues to  hold. 
The random code generation is as follows:

Fix $P_X$ and $\epsilon_1>\epsilon_2>0$ such that, $P_Z = Q_0$.

\emph{{Codebook Generation for Keys:}}
For each block $j\in\intseq{1}{B}$, create a function $\Phi:S_j^r\to\intseq{1}{2^{rR_K}}$ through random binning by choosing the value of $\Phi(s_j^r)$ independently and uniformly at random for every $s_j^r\in\mathcal{S}^r$. The key $k_j=\Phi(s_j^r)$ obtained in block $j\in\intseq{1}{B}$ from the state sequence $s_j^r$ is used to assist the encoder in the next block.

\emph{Codebook Generation for Messages:}
For each block $j\in\intseq{1}{B}$, let $C_r\triangleq\{X^r(m_j,k_{j-1})\}_{(m_j,k_{j-1})\in\mathcal{M}\times \mathcal{K}}$, where $\mathcal{M}\in\intseq{1}{2^{rR}}$ and $\mathcal{K}\in\intseq{1}{2^{rR_k}}$, be a random codebook consisting of independent random sequences each generated according to $P_X^{\otimes r}$. 
We denote a realization of $C_r$ by $\CodeBook_r\triangleq\{x^r(m_j,k_{j-1})\}_{(m_j,k_{j-1})\in\mathcal{M}\times\mathcal{K}}$.

\emph{{Encoding:}}
For the first block, we assume that the transmitter and the receiver have access to a shared secret key $k_0$, in this block to transmit $m_1$ the encoder computes $x^r(m_1,k_0)$ and transmits it over the channel. At the end of the first block the encoder generates a key from \ac{CSI} $s_1^r$ to be used in Block 2.

For block $j\in\intseq{2}{B}$, to send message $m_j$ according to the generated key $k_{j-1}$ from the previous block, the encoder computes $x^r(m_j,k_{j-1})$ and transmits it over the channel. Also, at the end of each  block $j\in\intseq{2}{B}$ the encoder generates a key from \ac{CSI} $s_j^r$ to be used in next block.

Define
\begin{align}
    \twocolalign \Upsilon_{M_j,K_{j-1},X^r,S_j^r,Z_j^r,K_j}^{(\mathcal{C}_r)}(m_j,k_{j-1},\tilde{x}^r,s_j^r,z_j^r,k_j) \triangleq\onecolalign 2^{-r(R + R_k)}\indic{1}_{\{ \tilde{x}^r=x^r(m_j,k_{j-1})\}}Q_S^{\otimes r}(s_j^r)\nonumber\\
    &\times W_{Z|X,S}^{\otimes r}\big(z_j^r|\tilde{x}^r,s_j^r\big)\indic{1}_{\{ k_j=\Phi (\tilde s_j^r)\}}.\label{eq:P_Dist_SCCSI}
\end{align}
For a fixed codebook $\mathcal{C}_r$, the induced joint distribution over the codebook (i.e. $P^{(\mathcal{C}_r)}$) satisfies
\begin{align}
\kld\Big(P_{M_j,K_{j-1},X^r,S_j^r,Z_j^r,K_j}^{(\mathcal{C}_r)}||\Upsilon_{M_j,K_{j-1},X^r,S_j^r,Z_j^r,K_j}^{(\mathcal{C}_r)}\Big)\leq\epsilon.\label{eq:Uniformity_Key_SC}
\end{align}
This intermediate distribution $\Upsilon^{(\mathcal{C}_r)}$ approximates the true distribution $P^{(\mathcal{C}_r)}$ and will be used in the sequel for bounding purposes. Expression~\eqref{eq:Uniformity_Key_SC} holds because the main difference in $\Upsilon^{(\mathcal{C}_r)}$ is assuming the key $K_{j-1}$ is uniformly distributed, which is made (arbitrarily) nearly uniform in $P^{(\mathcal{C}_r)}$ with appropriate control of rate.

\begin{figure*}
\centering
\includegraphics[width=5.0in]{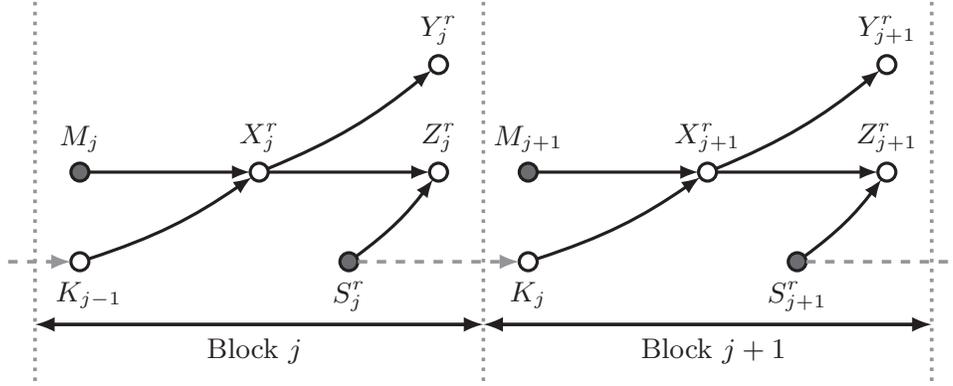}
\caption{Functional dependence graph for the block-Markov encoding scheme}
\label{fig:SC_Chaining}
\end{figure*}
\emph{{Covert Analysis:}}  We now show that this coding scheme  guarantees that $\expec_{C_n}[\kld(P_{Z^n|C_n} || Q_Z^{\otimes n} )]\underset{n\rightarrow\infty}{\xrightarrow{\hspace{0.2in}}} 0$, where $C_n$ is the set of all the codebooks for all blocks, and then choose $P_X$ such that it satisfies $Q_Z=Q_0$.

Similar to \eqref{eq:Total_KLD_Bound}, by using the functional dependence graph depicted in Fig.~\ref{fig:SC_Chaining},
\begin{align}
    \kld(P_{Z^n|C_n}||Q_Z^{\otimes n})\leq 2\sum\limits_{j = 1}^B {\kld(P_{Z_j^r,K_j|C_r}||Q_Z^{\otimes r}Q _{K_j})}.\label{eq:Total_KLD_Bound_SC}
\end{align}

We now proceed to bound the right hand side of \eqref{eq:Total_KLD_Bound_SC}. First, consider the marginal
\begin{align}
    \twocolalign \Upsilon_{Z_j^r,K_j|C_r} (z_j^r,k_j)=  \sum\limits_{m_j}\sum\limits_{k_{j-1}}\sum\limits_{s_j^r} \frac{1}{2^{r(R + R_k)}} Q_S^{\otimes r}(s_j^r) W_{Z|X,S}^{\otimes r}\big(z_j^r|X^r(m_j,k_{j-1}),s_j^r\big) \indic{1}_{\{ k_j=\Phi (s_j^r)\}}.\label{eq:Joint_Dist_ZK_SCCSI}  
\end{align}
To bound the \ac{RHS} of \eqref{eq:Total_KLD_Bound_SC} by using Lemma~\ref{lemma:TV_KLD} and the triangle inequality we have
\begin{align}
     \expec_{C_r}||P_{Z_j^r,K_j|C_r}-Q_Z^{\otimes r} Q_{K_j}||_1\leq \expec_{C_r}||P_{Z_j^r,K_j|C_r}-\Upsilon_{Z_j^r,K_j|C_r}||_1 + \expec_{C_r}||\Upsilon_{Z_j^r,K_j|C_r}-Q_Z^{\otimes r} Q_{K_j}||_1.\label{eq:Triangle_with_Real_Dist_SC}
\end{align}
From Lemma~\ref{lemma:TV_KLD} and \eqref{eq:Uniformity_Key_SC} the first term on the \ac{RHS} of \eqref{eq:Triangle_with_Real_Dist_SC} vanishes as $r$ grows; to bound the second term by using Lemma~\ref{lemma:TV_KLD} we have,
\begin{align}
& \expec_{C_r}[\kld(\Upsilon_{Z_j^r,K_j|C_r}||Q_Z^{\otimes r} Q_{K_j})] = \expec_{C_r}\bigg[\sum\limits_{(z_j^r,k_j)} \Upsilon_{Z_j^r,K_j|C_r}(z_j^r,k_j)\log\bigg(\frac{\Upsilon_{Z_j^r,K_j|C_r}(z_j^r,k_j)}{Q_Z^{\otimes r}(z_j^r)Q_{K_j}(k_j)}\bigg) \bigg] \nonumber\\
&= \expec_{C_r}\Bigg[\sum\limits_{(z_j^r,k_j)}  \sum\limits_{m_j}\sum\limits_{k_{j-1}} \frac{1}{2^{r(R + R_k)}}\sum\limits_{s_j^r} Q_S^{\otimes r}(s_j^r) W_{Z|X,S}^{\otimes r}\big(z_j^r|X^r(m_j,k_{j-1}),s_j^r\big)\indic{1}_{\{ k_j=\Phi (s_j^r)\}} \nonumber\\ 
&\qquad\times\log \Bigg(\frac{\sum\limits_{\tilde m_j}\sum\limits_{\tilde k_{j-1}}\sum\limits_{\tilde s_j^r}  Q_S^{\otimes r}(\tilde s_j^r) W_{Z|X,S}^{\otimes r}\big(z_j^r|X^r(\tilde{m}_j,\tilde{k}_{j-1}),\tilde s_j^r\big)\indic{1}_{\{ k_j=\Phi (\tilde s_j^r)\}}     }{2^{r(R + R_k - R_K)}Q_Z^{\otimes r}(z_j^r)}\Bigg)     \Bigg] \nonumber\\ 
&\mathop \le \limits^{(a)} \sum\limits_{(z_j^r,k_j)} \sum\limits_{m_j}\sum\limits_{k_{j-1}}  \frac{1}{2^{r(R + R_k)}}\sum\limits_{s_j^r} \sum\limits_{x^r(m_j,k_{j-1})} \Upsilon^{(\mathcal{C}_r)}_{X^r,S^r,Z^r}\big(x^r(m_j,k_{j-1}),s_j^r,z_j^r\big)\times\expec_{\Phi (s_j^r)}\big[\indic{1}_{\{ k_j=\Phi (s_j^r)\}}\big] \nonumber\\
&\qquad\times\log \expec_{\backslash ((m_j,k_{j-1}),\Phi (s_j^r))} \Bigg[ \frac{\sum\limits_{\tilde m_j}\sum\limits_{\tilde k_{j-1}} \sum\limits_{\tilde s_j^r} Q_S^{\otimes r}(\tilde s_j^r) W_{Z|X,S}^{\otimes r}\big(z_j^r|X^r(\tilde{m}_j,\tilde{k}_{j-1}),\tilde s_j^r\big)\indic{1}_{\{k_j= \Phi(\tilde s_j^r)\}}     }{{{2^{r(R + R_k - R_K)}}Q_Z^{\otimes r}(z_j^r)}}\Bigg]        \nonumber\\
&\mathop \le \limits^{(b)}\sum\limits_{(z_j^r,k_j)} \sum\limits_{m_j}\sum\limits_{k_{j-1}}  \frac{1}{2^{r(R + R_k)}}\sum\limits_{s_j^r} \sum\limits_{x^r(m_j,k_{j-1})} \Upsilon^{(\mathcal{C}_r)}_{X^r,S^r,Z^r}\big(x^r(m_j,k_{j-1}),s_j^r,z_j^r\big)\times \frac{1}{2^{rR_K}}  \nonumber\\
&\qquad\times\log\frac{1}{{2^{r(R + R_k - R_K)}Q_Z^{\otimes r}(z_j^r)} }\Bigg(Q_S^{\otimes r}(s_j^r) W_{Z|X,S}^{\otimes r}\big(z_j^r|x^r(m_j,k_{j-1}),s_j^r\big)\nonumber\\
&\qquad + \expec_{\backslash (m_j,k_{j-1})} \Bigg[\sum\limits_{(\tilde{m}_j,\tilde{k}_{j-1}) \ne (m_j,k_{j-1})}Q_S^{\otimes r}(s_j^r) W_{Z|X,S}^{\otimes r}\big(z_j^r|X^r(\tilde{m}_j,\tilde{k}_{j-1}),s_j^r\big)\Bigg] \nonumber\\
&\qquad + \expec_{\backslash \Phi (s_j^r)}\Bigg[ \sum\limits_{\tilde s_j^r \ne s_j^r}Q_S^{\otimes r}(\tilde s_j^r)\times W_{Z|X,S}^{\otimes r}\big(z_j^r|x^r(m_j,k_{j-1}),\tilde s_j^r\big)\indic{1}_{\{ k_j=\Phi (\tilde s_j^r)\}}\Bigg] \nonumber\\
&\qquad+ \expec_{\backslash ((m_j,k_{j-1}),\Phi (s_j^r))}\Bigg[ \sum\limits_{\tilde s_j^r \ne s_j^r}\sum\limits_{(\tilde{m}_j,\tilde{k}_{j-1}) \ne (m_j,k_{j-1})}Q_S^{\otimes r}(\tilde s_j^r) W_{Z|X,S}^{\otimes r}\big(z_j^r|X^r(\tilde{m}_j,\tilde{k}_{j-1}),\tilde s_j^r\big)\indic{1}_{\{ k_j=\Phi (\tilde s_j^r)\}}\Bigg] \Bigg)        \nonumber\\
&\mathop \le \limits^{(c)} \sum\limits_{(z_j^r,k_j)} \sum\limits_{m_j}\sum\limits_{k_{j-1}}  \frac{1}{2^{r(R + R_k+R_K)}}\sum\limits_{s_j^r} \sum\limits_{x^r(m_j,k_{j-1})} \Upsilon^{(\mathcal{C}_r)}_{X^r,S^r,Z^r}\big(x^r(m_j,k_{j-1}),s_j^r,z_j^r\big)  \nonumber\\
&\qquad\times\log \Bigg(\frac{Q_S^{\otimes r}(s_j^r) W_{Z|X,S}^{\otimes r}\big(z_j^r|x^r(m_j,k_{j-1}),s_j^r\big)}{2^{r(R + R_k - R_K)}Q_Z^{\otimes r}(z_j^r)} + \sum\limits_{(\tilde{m}_j,\tilde{k}_{j-1}) \ne (m_j,k_{j-1})} {\frac{{Q_{S,Z}^{\otimes r}(s_j^r,z_j^r)}}{{{2^{r(R + R_k - R_K)}}Q_Z^{\otimes r}(z_j^r)}}} \nonumber\\
&\qquad + \sum\limits_{\tilde s_j^r \ne s_j^r} \frac{Q_S^{\otimes r}(\tilde s_j^r) W_{Z|X,S}^{\otimes r}\big(z_j^r|x^r(m_j,k_{j-1}),\tilde s_j^r\big)}{2^{r(R + R_k)}Q_Z^{\otimes r}(z_j^r)} +1 \Bigg)        \nonumber\\
&\le \sum\limits_{(z_j^r,k_j)} \sum\limits_{m_j}\sum\limits_{k_{j-1}}  \frac{1}{2^{r(R + R_k+R_K)}}\sum\limits_{s_j^r} \sum\limits_{x^r(m_j,k_{j-1})} \Upsilon^{(\mathcal{C}_r)}_{X^r,S^r,Z^r}\big(x^r(m_j,k_{j-1}),s_j^r,z_j^r\big)  \nonumber\\
&\times\log \Bigg(\frac{Q_S^{\otimes r}(s_j^r) W_{Z|X,S}^{\otimes r}\big(z_j^r|x^r(m_j,k_{j-1}),s_j^r\big)}{2^{r(R + R_k - R_K)}Q_Z^{\otimes r}(z_j^r)} +  \frac{2^{rR_K} Q_{S,Z}^{\otimes r}(s_j^r,z_j^r)}{Q_Z^{\otimes r}(z_j^r)} + \frac{W_{Z|X}^{\otimes r}\big(z_j^r|x^r(m_j,k_{j-1})\big)}{2^{r(R + R_k)}Q_Z^{\otimes r}(z_j^r)}  + 1\Bigg)      \nonumber\\
&\triangleq \Psi_1 + \Psi_2,\label{eq:KLD_ZK_SCCSI}
\end{align}
where $(a)$ follows from Jensen's inequality, $(b)$ and $(c)$ are because $\indic{1}_{\{\cdot\}}\leq 1$,  and the last term in the \ac{RHS} of $(b)$ is smaller than $1$. We defined $\Psi_1$ and $\Psi_2$ as
\begin{align}
 &{\Psi _1} = \sum\limits_{{k_j}} \sum\limits_{k_{j-1}} \sum\limits_{m_j} \frac{1}{2^{r(R + R_k + R_K)}}\sum\limits_{\big(x^r(m_j,k_{j-1}),s_j^r,z_j^r\big) \in \mathcal{T}_\epsilon ^{(n)}} \Upsilon^{(\mathcal{C}_r)}_{X^r,S^r,Z^r}\big(x^r(m_j,k_{j-1}),s_j^r,z_j^r\big)\nonumber\\
&\qquad\times
 \log \Bigg(\frac{Q_S^{\otimes r}(s_j^r) W_{Z|X,S}^{\otimes r}\big(z_j^r|x^r(m_j,k_{j-1}),s_j^r\big)}{2^{r(R + R_k - R_K)}Q_Z^{\otimes r}(z_j^r)} +  \frac{2^{rR_K} Q_{S,Z}^{\otimes r}(s_j^r,z_j^r)}{Q_Z^{\otimes r}(z_j^r)} + \frac{W_{Z|X}^{\otimes r}\big(z_j^r|x^r(m_j,k_{j-1})\big)}{2^{r(R + R_k)}Q_Z^{\otimes r}(z_j^r)}  + 1\Bigg)      \nonumber\\
 & \le \log \Bigg(\frac{2^{rR_K}2^{ - r(1 - \epsilon )\big(\ent(S)+\ent(Z|X,S)\big)}}{2^{r(R + R_k)}2^{ - r(1 + \epsilon )\ent(Z)}} + \frac{2^{rR_K}2^{ - r(1 - \epsilon )\ent(S,Z)}}{2^{ - r(1 + \epsilon )\ent(Z)}} + \frac{2^{ - r(1 - \epsilon )\ent(Z|X)}}{2^{r(R + R_k)}2^{ - r(1 + \epsilon )\ent(Z)}} + 1\Bigg)\label{eq:Psi_1_SCCSI}\\
&{\Psi _2} =\sum\limits_{{k_j}} \sum\limits_{k_{j-1}} \sum\limits_{m_j} \frac{1}{2^{r(R + R_k + R_K)}}\sum\limits_{\big(x^r(m_j,k_{j-1}),s_j^r,z_j^r\big) \notin \mathcal{T}_\epsilon ^{(n)}} \Upsilon^{(\mathcal{C}_r)}_{X^r,S^r,Z^r}\big(x^r(m_j,k_{j-1}),s_j^r,z_j^r\big)\nonumber\\
&\qquad\times
 \log \Bigg(\frac{Q_S^{\otimes r}(s_j^r) W_{Z|X,S}^{\otimes r}\big(z_j^r|x^r(m_j,k_{j-1}),s_j^r\big)}{2^{r(R + R_k - R_K)}Q_Z^{\otimes r}(z_j^r)} +  \frac{2^{rR_K} Q_{S,Z}^{\otimes r}(s_j^r,z_j^r)}{Q_Z^{\otimes r}(z_j^r)} + \frac{W_{Z|X}^{\otimes r}\big(z_j^r|x^r(m_j,k_{j-1})\big)}{2^{r(R + R_k)}Q_Z^{\otimes r}(z_j^r)}  + 1\Bigg)      \nonumber\\
&\le 2|X||S||Z|e^{ - r\epsilon^2 {\mu _{X,S,Z}}}r\log\Big(\frac{3}{\mu_Z} + 1\Big), \label{eq:Psi_2_SCCSI}
\end{align}
\sloppy where $\mu_{X,S,Z}=\min\limits_{(x,s,z)\in(\mathcal{X},\mathcal{S},\mathcal{Z})}P_{X,S,Z}(x,s,z)$ and $\mu_Z=\min\limits_{z\in\mathcal{Z}}P_Z(z)$. When $r\to\infty$ then $\Psi_2\to 0$, by choosing $R_K=\ent(S|Z)-\epsilon$, $\Psi_1$ vanishes when $r$ grows if,
\begin{subequations}
\begin{align}
    R + R_k &>\mi(X;Z|S),\label{eq:Resolv_X_Against_S_and_Z__SCCSI}\\
    R + R_k &>\mi(X;Z)\label{eq:Resolv_Redun__SCCSI}.
\end{align}
\end{subequations}Here \eqref{eq:Resolv_Redun__SCCSI} is redundant because of \eqref{eq:Resolv_X_Against_S_and_Z__SCCSI} and the fact that $X$ and $S$ are independent.

\emph{{Decoding and Error Probability Analysis:}}
At the end of block $j\in\intseq{1}{B}$, using its knowledge of the channel state $s_j^r$ of the current block and the key $k_{j-1}$ generated from the previous block, the receiver finds a unique  $\hat{m}_j$ such that $\big(u^r(\hat{m}_j,k_{j-1}),s_j^r,y_j^r\big)\in\mathcal{T}_{\epsilon}^{(r)}$. To analyze the probability of error we define the following error events for $j\in\intseq{1}{B}$
\begin{subequations}\label{eq:error_Events_FCSI_SC}
\begin{align}
    \mathcal{E}&=\big\{\hat{M}\ne M\big\},\\
    \mathcal{E}_j&=\big\{\hat{M}_j\ne M_j\big\},\\
    \mathcal{E}_{1,j}&=\big\{\big(X^r(M_j,K_{j-1}),S_j^r\big)\notin\mathcal{T}_{\epsilon_1}^{(r)}(Q_SP_X)\big\},\\
    \mathcal{E}_{2,j}&=\big\{\big(X^r(M_j,K_{j-1}),S_j^r,Y_j^r\big)\notin\mathcal{T}_{\epsilon_2}^{(r)}(Q_SP_XW_{Y|X,S})\big\},\\
    \mathcal{E}_{3,j}&=\big\{\big(X^r(k_{j-1},\hat{m}_j),S_j^r,Y_j^r\big)\in\mathcal{T}_{\epsilon_2}^{(r)},\quad\mbox{for some}\quad\hat{m}_j\ne M_j\big\},
\end{align}
\end{subequations}where $\epsilon_2>\epsilon_1>\epsilon>0$. The probability of error is upper bounded as follows,
\begin{align}
    \Prob(\mathcal{E})\leq\Prob\Big\{\bigcup\nolimits_{j = 1}^B \mathcal{E}_j\Big\}\leq\sum_{j=1}^B\Prob(\mathcal{E}_j).
\end{align}Now we bound $\Prob(\mathcal{E}_j)$ by using union bound
\begin{align}
    \Prob(\mathcal{E}_j)\leq \Prob(\mathcal{E}_{1,j})+\Prob(\mathcal{E}_{1,j}^c\cap \mathcal{E}_{2,j})+\Prob(\mathcal{E}_{2,j}^c\cap \mathcal{E}_{3,j}).\label{eq:PError_SC}
\end{align}By the law of large numbers the first and second term on \ac{RHS} of \eqref{eq:PError_SC} vanishes when $r$ grows. According to the law of large numbers and the packing lemma, the last term on \ac{RHS} of \eqref{eq:PError_SC} vanishes when $r$ grows if \cite{ElGamalKim},
\begin{align}
&R < \mi(X;S,Y)=\mi(X;Y|S).
\label{eq:Decoding_BM_SCCSI}
\end{align}

Furthermore, this scheme requires that,
\begin{align}
    R_k\leq R_K = \ent(S|Z) - \epsilon.\label{eq:Indep_S_and_Z_SCCSI}
\end{align}
The region in Theorem~\ref{thm:Capacity_FCSI_SC} is obtained by applying Fourier-Motzkin to \eqref{eq:Resolv_X_Against_S_and_Z__SCCSI}, \eqref{eq:Decoding_BM_SCCSI}, and  \eqref{eq:Indep_S_and_Z_SCCSI}.
\subsection{Converse Proof}
To establish the upper bound, consider any sequence of length-$n$ codes for a state-dependent channel with channel state available strictly causally at both the transmitter and the receiver such that $P_e^{(n)}\leq\epsilon_n$ and $\kld(P_{Z^n}||Q_0^{\otimes n})\leq\delta$ with $\lim_{n\to\infty}\epsilon_n=0$. Note that the converse is consistent with the model and does \emph{not} require $\delta$ to vanish.
\subsubsection{Epsilon Rate Region}
We first define a region $\mathcal{S}_{\epsilon}$ for $\epsilon>0$ that expands the region defined in~\eqref{eq:finalSD_FCSI_SC} as follows.
\begin{subequations}\label{eq:Epsilon_Rate_Region_FCSI_SC}
\begin{align}
\calS_\epsilon\eqdef \big\{R\geq 0: \exists P_{S,X,Y,Z}\in\calD_\epsilon: R\leq \mi(X;Y|S) + \epsilon \big\}\label{eq:Sepsilon_FCSI_SC}
\end{align}
where 
\begin{align}
  \calD_\epsilon = \left.\begin{cases}P_{S,X,Y,Z}:\\
P_{S,X,Y,Z}=Q_SP_XW_{Y,Z|X,S}\\
\mathbb{D}\left(P_Z\Vert Q_0\right) \leq \epsilon\\
\ent(S|Z) >  \mi(X;Z|S) - \mi(X;Y|S) - 3\epsilon
\end{cases}\right\},\label{eq:Depsilon_FCSI_SC}
\end{align}
\end{subequations}where $\epsilon\triangleq\max\{\epsilon_n,\nu\geq\frac{\delta}{n}\}$. 
We next show that if a rate $R$ is achievable then $R\in\mathcal{S}_{\epsilon}$ for any $\epsilon>0$. 

For any $\epsilon_n>0$, we start by upper bounding $nR$ using standard techniques.
\begin{align}
    nR &= \ent(M) \nonumber\\
    &\mathop\leq\limits^{(a)} \ent(M|S^n) - \ent(M|Y^n,S^n) + n\epsilon_n\nonumber\\
    &= \mi(M;Y^n|S^n) + n\epsilon_n\nonumber\\
    &= \sum_{i=1}^n{\mi(M;Y_i|Y^{i-1},S^n)} + n\epsilon_n\nonumber\\
    &\leq \sum_{i=1}^n{[\ent(Y_i|Y^{i-1},S^n)-\ent(Y_i|Y^{i-1},S^n,X_i,M)]} + n\epsilon_n\nonumber\\
    &\mathop\leq\limits^{(b)}\sum_{i=1}^n{\mi(X_i;Y_i|S_i)} + n\epsilon_n\nonumber\\
    &\mathop\leq\limits^{(c)} n\mi(\tilde{X};\tilde{Y}|\tilde{S}) + n\epsilon_n\nonumber\\
    &\mathop\leq\limits^{(d)} n\mi(\tilde{X};\tilde{Y}|\tilde{S}) + n\epsilon,\label{eq:Upper_Bound_on_MRate_SC}
\end{align}where 
\begin{itemize}
    \item[$(a)$] follows from Fano's inequality and independence of $M$ from $S^n$;
    \item[$(b)$] holds because conditioning does not increase entropy and $(M,Y^{i-1},S_{\sim i}^{n},X_{\sim i}^{n})-(X_i,S_i)-Y_i$ forms a Markov chain;
    \item[$(c)$] follows from concavity of mutual information, with respect to the input distribution, with the random variables $\tilde{X}$, $\tilde{S}$, $\tilde{Y}$, and $\tilde{Z}$ having the following distributions
\begin{subequations}\label{eq:Joint_Dist_Tilde_FCSI_SC}
\begin{align}
\tilde{P}_{X,S}(x,s) &\triangleq \frac{1}{n}\sum\limits_{i = 1}^n P_{X_i,S_i}(x,s), \label{eq:Joint_Dist_Tilde_FCSI_SC_1}\\
\tilde{P}_{X,S,Y,Z}(x,s,y,z) &\triangleq \tilde{P}_{X,S}(x,s)W_{Y,Z|X,S}(y,z|x,s),\label{eq:Joint_Dist_Tilde_FCSI_SC_2}
\end{align}
\end{subequations}
\item[$(d)$] follows by defining $\epsilon\triangleq\max\{\epsilon_n,\nu\geq\frac{\delta}{n}\}$.
\end{itemize}

By following the same steps as in \eqref{eq:Upper_Bound_on_KRate_FCSI_NC_Final}
we also have,
\begin{align}
nR 
&\ge n\mi(\tilde{X},\tilde{S};\tilde{Z}) - n\ent(\tilde{S})-2\epsilon,
\label{eq:Upper_Bound_on_KRate_SC}
\end{align}
where the random variables $\tilde{X}$, $\tilde{S}$, $\tilde{Y}$, and $\tilde{Z}$ have been defined in \eqref{eq:Joint_Dist_Tilde_FCSI_SC}.
Substituting \eqref{eq:MI_Manipulation} into \eqref{eq:Upper_Bound_on_KRate_SC} leads to 
\begin{align}
    R\ge\mi(\tilde{X};\tilde{Z}|\tilde{S}) -\ent(\tilde{S}|\tilde{Z})-2\epsilon.\label{eq:Upper_Bound_on_KRate_SC_Final}
\end{align}
To show that $ \kld(P_Z||Q_0)\leq\epsilon$, note that for $n$ large enough
\begin{align}
\kld(P_Z||Q_0)=\kld(P_{\tilde{Z}}||Q_0)=\kld\Bigg(\frac{1}{n}\sum\limits_{i=1}^nP_{Z_i}\Bigg|\Bigg|Q_0\Bigg)\leq\frac{1}{n}\sum\limits_{i=1}^n\kld(P_{Z_i}||Q_0)\leq\frac{1}{n}\kld(P_{Z^n}||Q_0^{\otimes n})\leq\frac{\delta}{n}\leq\nu\leq\epsilon.\label{eq:boundKL_FCSI_SC}
\end{align}
Combining \eqref{eq:Upper_Bound_on_MRate_SC} and \eqref{eq:Upper_Bound_on_KRate_SC_Final} shows that $\forall \epsilon_n,\nu>0$, $R\leq \max\{x:x\in\mathcal{S}_{\epsilon}\}$. Therefore,
\begin{align}
  C_{\mbox{\scriptsize\rm SC-TR}} = \max\left\{x:x\in\bigcap_{\epsilon>0}\mathcal{S}_{\epsilon}\right\}.
\end{align}
\subsubsection{Continuity at Zero}
One can prove the continuity at zero of $\calS_\epsilon$ by substituting $\min\{\mi(U;Y)-\mi(U;S),\mi(U,V;Y)-\mi(U;S|V)\}$ with $\mi(X;Y|S)$ and $\mi(V;Z)-\mi(V;S)$ with $\mi(X;Z|S)-\ent(S|Z)$ in Appendix~\ref{sec:continuity-at-zero} and following the exact same arguments.


\section{Proof of Theorem~\ref{thm:Acivable_Rate_CSIT_WZ_NC}}
\label{sec:Proof_Acivable_Rate_CSIT_WZ_NC}

We adopt a block-Markov encoding scheme in which $B$ independent messages are transmitted over $B$ channel blocks each of length $r$, such that $n=rB$.
The warden's observation is $Z^n=(Z_1^r,\dots,Z_B^r)$, the distribution induced at the output of the warden is $P_{Z^n}$, the target output distribution is $Q_Z^{\otimes n}$, and Equation \eqref{eq:Total_KLD_CoNC}, describing the distance between the two distributions, continues to  hold. 
The random code generation is as follows:

Fix $P_{U|S}(u|s)$, $P_{V|S}(v|s)$, $x(u,s)$, and $\epsilon_1>\epsilon_2>0$ such that, $P_Z = Q_0$. 

\emph{Codebook Generation for Keys:}
For each block $j\in\intseq{1}{B}$, let $C_1^{(r)}\triangleq\big\{V^r(a_j)\big\}_{a_j\in\mathcal{A}}$, where $\mathcal{A}=\intseq{1}{2^{r\tilde{R}}}$, be a random codebook consisting of independent random sequences each generated according to $P_V^{\otimes r}$, where $P_V=\sum_{s\in\mathcal{S}}Q
_S(s)P_{V|S}(v|s)$. We denote a realization of $C_1^{(r)}$ by $\mathcal{C}_1^{(r)}\triangleq\big\{v^r(a_j)\big\}_{a_j\in\mathcal{A}}$. Partition the set of indices $a_j\in\intseq{1}{2^{r\tilde{R}}}$ into bins $\mathcal{B}(t)$, $t\in\intseq{1}{2^{rR_T}}$ by using function $\varphi:V^r(a_j)\to\intseq{1}{2^{rR_T}}$ through random binning by choosing the value of $\varphi(v^r(a_j))$ independently and uniformly at random for every $v^r(a_j)\in\mathcal{V}^r$.
For each block $j\in\intseq{1}{B}$, create a function $\Phi:V^r(a_j)\to\intseq{1}{2^{rR_K}}$ through random binning by choosing the value of $\Phi(v^r(a_j))$ independently and uniformly at random for every $v^r(a_j)\in\mathcal{V}^r$. The key $k_j=\Phi(v^r(a_j))$ obtained in block $j\in\intseq{1}{B}$ from the description of the channel state sequence $v^r(a_j)$ is used to assist the encoder in block $j+2$.

\emph{Codebook Generation for Messages:}
For each block $j\in\intseq{1}{B}$, let $C_2^{(r)}\triangleq\big\{U^r(m_j,t_{j-1},k_{j-2},\ell_j)\big\}_{(m_j,t_{j-1},k_{j-2},\ell_j)\in\mathcal{M}\times\mathcal{T}\times\mathcal{K}\times\mathcal{L}}$, where $\mathcal{M}=\intseq{1}{2^{rR}}$, $\mathcal{T}=\intseq{1}{2^{rR_t}}$, $\mathcal{K}=\intseq{1}{2^{rR_k}}$, and $\mathcal{L}=\intseq{1}{2^{rR'}}$, be a random codebook consisting of independent random sequences each generated according to $P_U^{\otimes r}$. We denote a realization of $C_2^{(r)}$ by $\mathcal{C}_2^{(r)}\triangleq\big\{u^r(m_j,t_{j-1},k_{j-2},\ell_j)\big\}_{(m_j,t_{j-1},k_{j-2},\ell_j)\in\mathcal{M}\times\mathcal{T}\times\mathcal{K}\times\mathcal{L}}$. Let, $C_r=\big\{C_1^{(r)},C_2^{(r)}\big\}$ and $\mathcal{C}_r=\big\{\mathcal{C}_1^{(r)},\mathcal{C}_2^{(r)}\big\}$. The indices $(m_j,t_{j-1},k_{j-2},\ell_j)$ can be viewed as a three layer binning. 
We define an ideal \ac{PMF} for codebook $\mathcal{C}_r$, as an approximate distribution to facilitate the analysis
\begin{align}
    &\Gamma_{M_j,T_{j-1},K_{j-2},L_j,A_j,U^r,V^r,S_j^r,Z_j^r,K_{j-1},T_j,K_j}^{(\mathcal{C}_r)}(m_j,t_{j-1},k_{j-2},\ell_j,a_j,\tilde{u}^r,\tilde{v}_j^r,s_j^r,z_j^r,k_{j-1},t_j,k_j) \nonumber\\
    &\qquad=2^{-r(R + R_t + R_k  + R'+\tilde{R})}\indic{1}_{\{ \tilde{u}^r=u^r(m_j,t_{j-1},k_{j-2},\ell_j) \}}\indic{1}_{\{ \tilde{v}^r=v^r(a_j)\}}
     P_{S|U,V}^{\otimes r}(s_j^r|\tilde{u}^r,\tilde{v}^r)\nonumber\\
    &\qquad\quad\times W_{Z|U,S}^{\otimes r}(z_j^r|\tilde{u}^r,s_j^r)2^{-rR_k} \indic{1}_{\{ t_j=\sigma (\tilde{v}^r) \}} \indic{1}_{\{ k_j=\Phi (\tilde{v}^r) \}},\label{eq:Ideal_PMF}
\end{align}where $W_{Z|U,S}$ is the marginal distribution $W_{Z|U,S}=\sum_{x\in\mathcal{X}}\indic{1}_{\{x=x(u,s)\}}W_{Z|X,S}$ and $P_{S|U,V}$ is defined as follows
\begin{align}
    P_{S|U,V}(s|u,v)&\triangleq\frac{P_{S,U,V}(s,u,v)}{P_{U,V}(u,v)}=\frac{Q_S(s)P_{U|S}(u|s)P_{V|S}(v|s)}{\sum_{s\in\mathcal{S}}Q_S(s)P_{U|S}(u|s)P_{V|S}(v|s)}.\label{eq:S_Dist_NC}
\end{align}

\emph{{Encoding:}}
\begin{figure*}
\centering
\includegraphics[width=6.0in]{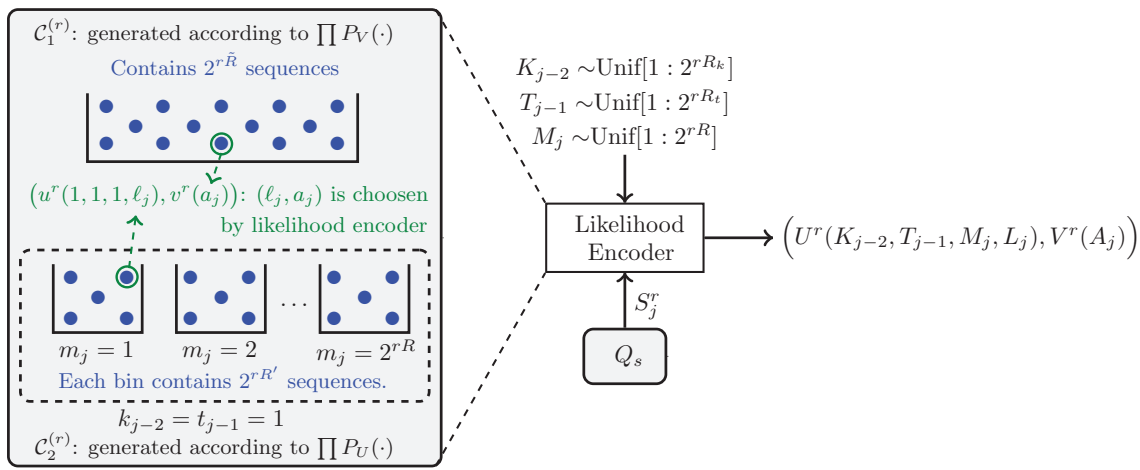}
\caption{Proposed coding scheme for the dual use of \ac{CSI}}
\label{fig:Encoding}
\end{figure*}

We assume that the transmitter and the receiver have access to shared secret keys $k_{-1}$ and $k_0$ for the first two blocks but after the first two blocks they use the key that they generate from the \ac{CSI}.  

In the first block, to send the message $m_1$ according to $k_{-1}$, the encoder chooses the index $t_0$ uniformly at random and then chooses the indices $\ell_1$ and $a_1$ according to the following distribution with $j=1$,
\begin{align}
 f(\ell_j,a_j|s_j^r,m_j,t_{j-1},k_{j-2})=\frac{P_{S|U,V}^{\otimes r}\left(s_j^r|u^r(m_j,t_{j-1},k_{j-2},\ell_j),v^r(a_j)\right)}{\sum\limits_{\ell'_j\in\intseq{1}{2^{rR'}}}\sum\limits_{a'_j\in\intseq{1}{2^{r\tilde{R}}}}P_{S|U,V}^{\otimes r}\left(s_j^r|u^r(m_j,t_{j-1},k_{j-2},\ell'_j),v^r(a'_j)\right)},
\label{eq:Likelihood_Encoder_BM}
\end{align}
where $P_{S|U,V}$ is defined in~\eqref{eq:S_Dist_NC}. Based on these indices, the encoder computes $u^r(m_1,t_0,k_{-1},\ell_1)$ and $v^r(a_1)$ and transmits codeword $x^r$, where $x_i=x(u_i(m_1,t_0,k_{-1},\ell_1),s_i)$. Note that, the index $t_0$ does not convey any useful information. Simultaneously, it uses the description of its \ac{CSI} $v^r(a_1)$ to generate a reconciliation index $t_1$ and a key $k_1$ to be used in the second and the third blocks, respectively. 

In the second block, to send the message $m_2$ and reconciliation index $t_1$ according to $k_0$, the encoder chooses the indices $\ell_2$ and $a_2$ according to the likelihood encoder described in \eqref{eq:Likelihood_Encoder_BM} with $j=2$. Based on these indices the encoder computes $u^r(m_2,t_1,k_0,\ell_2)$ and $v^r(a_2)$ and transmits codeword $x^r$, where $x_i=x(u_i(m_2,t_1,k_0,\ell_1),s_i)$. Simultaneously, it uses the description of its \ac{CSI} $v^r(a_2)$ to generate a reconciliation index $t_2$ and a key $k_2$ to be used in the third and the fourth block, respectively.

In block $j\in\intseq{3}{B}$, to send the message $m_j$ and reconciliation index $t_{j-1}$, generated in the previous block, according to the generated key $k_{j-2}$ from the the block $j-2$ and the \ac{CSI} of the current block, the encoder selects indices $\ell_j$ and $a_j$ from bin $(m_j,t_{j-1},k_{j-2})$ according to the likelihood encoder described in \eqref{eq:Likelihood_Encoder_BM}. 
The encoder then transmits codeword $x^r$, where each coordinate of the transmitted signal is a function of the state, as well as the corresponding sample of the transmitter's codeword $u_i$, i.e., $x_i=x(u_i(m_j,t_{j-1},k_{j-2},\ell_j),s_i)$. Simultaneously, the encoder uses the description of its \ac{CSI} $v^r(a_j)$ to generate a reconciliation index $t_j$ and a key $k_j$ to be used in the block $j+1$ and the block $j+2$, respectively. The encoding scheme in block $j\in\intseq{3}{B}$ is depicted in Fig.~\ref{fig:Encoding}. 

Define
\begin{align}
    \twocolalign &\Upsilon_{M_j,T_{j-1},K_{j-2},S_j^r,L_j,A_j,U^r,V^r,Z_j^r,K_{j-1},T_j,K_j}^{(\mathcal{C}_r)}(m_j,t_{j-1},k_{j-2},s_j^r,\ell_j,a_j,\tilde{u}^r,\tilde{v}^r,z_j^r,k_{j-1},t_j,k_j)\nonumber\\
    &\qquad\triangleq 2^{-r(R+R_t+R_k)} Q_S^{\otimes r}(s_j^r)  f(\ell_j,a_j|s_j^r,m_j,t_{j-1},k_{j-2})\indic{1}_{\{ \tilde{u}^r=u^r(m_j,t_{j-1},k_{j-2},\ell_j) \}}\indic{1}_{\{ \tilde{v}^r=v^r(a_j)\}}\nonumber\\
    &\qquad\quad\times W_{Z|U,S}^{\otimes r}(z_j^r|\tilde{u}^r,s_j^r)2^{-rR_k}\indic{1}_{\{ t_j=\sigma(\tilde{v}^r)\}}\indic{1}_{\{ k_j=\Phi (\tilde{v}^r) \}}.\label{eq:P_Dist}
\end{align}
For a given codebook $\mathcal{C}_r$, the induced joint distribution over the codebook (i.e. $P^{(\mathcal{C}_r)}$) satisfies
\begin{align}
\kld\Big(P_{M_j,T_{j-1},K_{j-2},S_j^r,L_j,A_j,U^r,V^r,Z_j^r,K_{j-1},T_j,K_j}^{(\mathcal{C}_r)}||\Upsilon_{M_j,T_{j-1},K_{j-2},S_j^r,L_j,A_j,U^r,V^r,Z_j^r,K_{j-1},T_j,K_j}^{(\mathcal{C}_r)}\Big)\leq\epsilon.\label{eq:Uniformity_Key_NC_CSIT}
\end{align}
This intermediate distribution $\Upsilon^{({\cal C}_r)}$ approximates the true distribution $P^{({\cal C}_r)}$ and will be used in the sequel for bounding purposes. Expression~\eqref{eq:Uniformity_Key_NC_CSIT} holds because the main difference in $\Upsilon^{({\cal C}_r)}$ is assuming the keys $K_{j-2}$, $K_{j-1}$ and the reconciliation index $T_{j-1}$ are uniformly distributed, which is made (arbitrarily) nearly uniform in $P^{({\cal C}_r)}$ with appropriate control of rate.

\emph{{Covert Analysis:}} 
We now show $\expec_{C_n}[\kld(P_{Z^n|C_n}|| Q_Z^{\otimes n} )]\underset{n\rightarrow\infty}{\xrightarrow{\hspace{0.2in}}} 0$, where $C_n$ is the set of all the codebooks for all blocks, and then choose $P_U$, $P_V$, and $x(u,s)$ such that it satisfies $Q_Z=Q_0$. From the expansion in \eqref{eq:Total_KLD_CoNC}, for every $j\in\intseq{2}{B}$,
\begin{align}
 \mi(Z_j^r;Z_{j+1}^{B,r}) &\leq \mi(Z_j^r;K_{j-1},T_j,K_j,Z_{j+1}^{B,r}) \nonumber\\
 &\mathop=\limits^{(a)} \mi(Z_j^r;K_{j-1},T_j,K_j), \label{eq:Indep_Across_Blocks_C_CSIT}
\end{align}
where $(a)$ holds because $Z_j^r-(K_{j-1},T_j,K_j)-Z_{j+1}^{B,r}$ forms a Markov chain, as seen in the functional dependence graph depicted in Fig.~\ref{fig:Chaining_Noncausal_CSIT}. Also,
\begin{align}
    \mi(Z_j^r;K_{j-1},T_j,K_j) &= \kld(P_{Z_j^r,K_{j-1},T_j,K_j}||P_{Z_j^r}P_{K_{j-1},T_j,K_j})\nonumber\\
  &\mathop\leq\limits^{(b)} \kld(P_{Z_j^r,K_{j-1},T_j,K_j}||Q_Z^{\otimes r}Q_{K_{j-1}}Q_{T_j}Q_{K_j}),\label{eq:KLD_Z3_C_CSIT}
\end{align}
where $Q_{K_{j-1}}Q_{K_j}Q_{T_j}$ is the uniform distribution on $\intseq{1}{2^{rR_k}}\times\intseq{1}{2^{rR_K}}\times\intseq{1}{2^{rR_T}}$
and $(b)$ follows from
\begin{align}
  \kld(P_{Z_j^r,K_{j-1},T_j,K_j}||P_{Z_j^r}P_{K_{j-1},T_j,K_j}) &= \kld(P_{Z_j^r,K_{j-1},T_j,K_j}||Q_Z^{\otimes r}Q_{K_{j-1}}Q_{T_j}Q_{K_j}) \nonumber\\
  &\quad- \kld(P_{Z_j^r}||Q_Z^{\otimes r}) - \kld(P_{K_{j-1},T_j,K_j}||Q_{K_{j-1}}Q_{T_j}Q_{K_j}) .
  \label{eq:KLD_Property_C_CSIT}
\end{align} 
Therefore, by combining \eqref{eq:Total_KLD_CoNC}, \eqref{eq:KLD_Z3_C_CSIT}, and \eqref{eq:KLD_Property_C_CSIT}
\begin{align}
    \kld(P_{Z^n}||Q_Z^{\otimes n})\leq 2\sum\limits_{j = 1}^B {\kld(P_{Z_j^r,K_{j-1},T_j,K_j}||Q_Z^{\otimes r}Q_{K_{j-1}}Q_{T_j}Q _{K_j})}.\label{eq:Total_KLD_Bound_NC_CSIT}
\end{align}
\begin{figure*}
\centering
\includegraphics[width=6.0in]{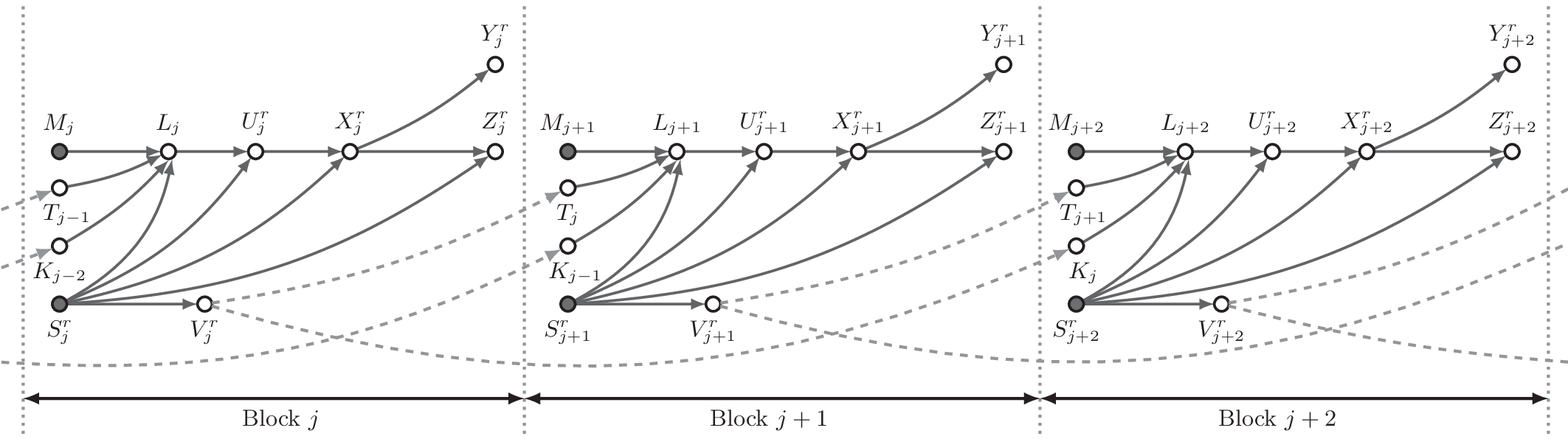}
\caption{Functional dependence graph for the block-Markov encoding scheme}
\label{fig:Chaining_Noncausal_CSIT}
\end{figure*}
To bound the \ac{RHS} of \eqref{eq:Total_KLD_Bound_NC_CSIT} by using Lemma~\ref{lemma:TV_KLD} and the triangle inequality we have,
\begin{align}
    &\expec_{C_r}||P_{Z_j^r,K_{j-1},T_j,K_j|C_r}-Q_Z^{\otimes r}Q_{K_{j-1}}Q_{T_j}Q _{K_j}||_1 \leq\expec_{C_r}||P_{Z_j^r,K_{j-1},T_j,K_j|C_r}-\Gamma_{Z_j^r,K_{j-1},T_j,K_j|C_r}||_1  \nonumber\\ 
    &\qquad+ \expec_{C_r}||\Gamma_{Z_j^r,K_{j-1},T_j,K_j|C_r}-Q_Z^{\otimes r}Q_{K_{j-1}}Q_{T_j}Q _{K_j}||_1\nonumber\\ &\leq\expec_{C_r}||P_{Z_j^r,K_{j-1},T_j,K_j|C_r}-\Upsilon_{Z_j^r,K_{j-1},T_j,K_j|C_r}||_1 + \expec_{C_r}||\Upsilon_{Z_j^r,K_{j-1},T_j,K_j|C_r}-\Gamma_{Z_j^r,K_{j-1},T_j,K_j|C_r}||_1  \nonumber\\
    &\qquad+ \expec_{C_r}||\Gamma_{Z_j^r,K_{j-1},T_j,K_j|C_r}-Q_Z^{\otimes r}Q_{K_{j-1}}Q_{T_j}Q _{K_j}||_1.\label{eq:Triangle_Inequality_2}
\end{align}
From \eqref{eq:Uniformity_Key_NC_CSIT} and the monotonicity of KL-divergence the first term on the \ac{RHS} of \eqref{eq:Triangle_Inequality_2} goes to zero when $r$ grows. To bound the second term on the \ac{RHS} of \eqref{eq:Triangle_Inequality_2} for a fixed codebook $\mathcal{C}_r$ we have,
\begin{subequations}\label{eq:Upsilon_Gamma}
\begin{align}
  \Gamma_{M_j,T_{j-1},K_{j-2}}^{(\mathcal{C}_r)} &= 2^{-r(R+R_t+R_k)} = \Upsilon_{M_j,T_{j-1},K_{j-2}}^{(\mathcal{C}_r)},\label{eq:Uniformity_M_2}\\
  \Gamma_{L_j,A_j|M_j,T_{j-1},K_{j-2},S_j^r}^{(\mathcal{C}_r)} &= f^{LE}(\ell_j,a_j|s_j^r,m_j,t_{j-1},k_{j-2}) = \Upsilon_{L_j,A_j|M_j,T_{j-1},K_{j-2},S_j^r}^{(\mathcal{C}_r)},\label{eq:Likelihood_Encoder_Equality_2}\\
  \Gamma_{U^r|M_j,T_{j-1},K_{j-2},S_j^r,L_j,A_j}^{(\mathcal{C}_r)} &= \indic{1}_{\{\tilde{u}^r= u^r(m_j,t_{j-1},k_{j-2},\ell_j) \}} \twocolbreak= \Upsilon_{U^r|M_j,T_{j-1},K_{j-2},S_j^r,L_j,A_j}^{(\mathcal{C}_r)},\label{eq:eq4}\\
  \Gamma_{V^r|M_j,T_{j-1},K_{j-2},S_j^r,L_j,A_j,U^r}^{(\mathcal{C}_r)} &=  \indic{1}_{\{ \tilde{v}^r = v^r(a_j) \}} = \Upsilon_{V^r|M_j,T_{j-1},K_{j-2},S_j^r,L_j,A_j,U^r}^{(\mathcal{C}_r)},\label{eq:eq5}\\
\Gamma_{Z_j^r|M_j,T_{j-1},K_{j-2},S_j^r,L_j,A_j,U^r,V^r}^{(\mathcal{C}_r)} &= W_{Z|U,S}^{\otimes r}  = \Upsilon_{Z_j^r|M_j,T_{j-1},K_{j-2},S_j^r,L_j,A_j,U^r,V^r}^{(\mathcal{C}_r)},\label{eq:Channel_Equality_2}\\
\Gamma_{K_{j-1}|M_j,T_{j-1},K_{j-2},S_j^r,L_j,A_j,U^r,V^r,Z_j^r}^{(\mathcal{C}_r)} &= 2^{-rR_k}  = \Upsilon_{K_{j-1}|M_j,T_{j-1},K_{j-2},S_j^r,L_j,A_j,U^r,V^r,Z_j^r}^{(\mathcal{C}_r)},\label{eq:Second_Key}\\
\Gamma_{T_j|M_j,T_{j-1},K_{j-2},S_j^r,L_j,A_j,U^r,V^r,Z_j^r,K_{j-1}}^{(\mathcal{C}_r)} &=\indic{1}_{\{ t_j=\sigma (v^r) \}}=\Upsilon_{T_j|M_j,T_{j-1},K_{j-2},S_j^r,L_j,A_j,U^r,V^r,Z_j^r,K_{j-1}}^{(\mathcal{C}_r)},\label{eq:K_j_Distribution}\\
\Gamma_{K_j|M_j,T_{j-1},K_{j-2},S_j^r,L_j,A_j,U^r,V^r,Z_j^r,K_{j-1},T_j}^{(\mathcal{C}_r)} &=\indic{1}_{\{ k_j=\Phi (v^r)\}}=\Upsilon_{K_j|M_j,T_{j-1},K_{j-2},S_j^r,L_j,A_j,U^r,V^r,Z_j^r,K_{j-1},T_j}^{(\mathcal{C}_r)},\label{eq:T_j_Distribution}
\end{align}
\end{subequations}
where \eqref{eq:Likelihood_Encoder_Equality_2} follows from \eqref{eq:Likelihood_Encoder_BM}. Hence,
\begin{align}
 &\expec_{C_r}||\Upsilon_{Z_j^r,K_{j-1},T_j,K_j|C_r}-\Gamma_{Z_j^r,K_{j-1},T_j,K_j|C_r}||_1
 \twocolbreak \nonumber\\ &\leq\expec_{C_r}||\Upsilon_{M_j,T_{j-1},K_{j-2},S_j^r,L_j,A_j,U^r,V^r,Z_j^r,K_{j-1},T_j,K_j|C_r}- \Gamma_{M_j,T_{j-1},K_{j-2},S_j^r,L_j,A_j,U^r,V^r,Z_j^r,K_{j-1},T_j,K_j|C_r}||_1 \nonumber\\
&\mathop = \limits^{(a)}\expec_{C_r}||\Upsilon_{M_j,T_{j-1},K_{j-2},S_j^r|C_r}- \Gamma_{M_j,T_{j-1},K_{j-2},S_j^r|C_r}||_1 \nonumber\\
&\mathop = \limits^{(b)} \expec_{C_r}||Q_S^{\otimes r} - \Gamma_{S_j^r|M_j=1,T_{j-1}=1,K_{j-2}=1,C_r}||_1,\label{eq:General_Distribution_Expansion_2}
\end{align}where $(a)$ follows from \eqref{eq:Likelihood_Encoder_Equality_2}-\eqref{eq:T_j_Distribution} and $(b)$ follows from the symmetry of the codebook construction with respect to $M_j$, $T_{j-1}$, and $K_{j-2}$ and \eqref{eq:Uniformity_M_2}. Based on \cite[Theorem~2]{YassaeeMAWC} the \ac{RHS} of \eqref{eq:General_Distribution_Expansion_2} vanishes if
\begin{subequations}\label{eq:MAC_Covering_Lemma}
\begin{align}
    R' &> \mi(U;S),\label{eq:MAC_Covering_Lemma_1}\\
    \tilde{R} &> \mi(V;S),\label{eq:MAC_Covering_Lemma_2}\\
    R' + \tilde{R} &> \mi(U,V;S).\label{eq:MAC_Covering_Lemma_3}
\end{align}
\end{subequations}
We now proceed to bound the third term on the \ac{RHS} of \eqref{eq:Triangle_Inequality_2}. First, consider the marginal
\begin{align}
    &\Gamma_{Z_j^r,K_{j-1},T_j,K_j|C_r} (z_j^r,k_{j-1},t_j,k_j)= \sum\limits_{m_j} \sum\limits_{t_{j-1}}\sum\limits_{k_{j-2}} \sum\limits_{\ell_j}\sum\limits_{a_j}\sum\limits_{s_j^r} \frac{1}{2^{r(R + R_t + 2R_k + R'+\tilde{R})}}\nonumber\\
    &\qquad\times P_{S|U,V}^{\otimes r}\big(s_j^r|U^r(m_j,t_{j-1},k_{j-2},\ell_j),V^r(a_j)\big) W_{Z|U,S}^{\otimes r}\big(z_j^r|U^r(m_j,t_{j-1},k_{j-2},\ell_j),s_j^r\big)\nonumber\\
    &\qquad\times \indic{1}_{\{ t_j=\sigma (V^r(a_j))\}}\indic{1}_{\{ k_j=\Phi (V^r(a_j)) \}}\\
    &= \sum\limits_{m_j} \sum\limits_{t_{j-1}}\sum\limits_{k_{j-2}} \sum\limits_{\ell_j}\sum\limits_{a_j}\frac{1}{2^{r(R + R_t + 2R_k + R'+\tilde{R})}} W_{Z|U,V}^{\otimes r}\big(z_j^r|U^r(m_j,t_{j-1},k_{j-2},\ell_j),V^r(a_j)\big)\nonumber\\
    &\qquad\times \indic{1}_{\{ t_j=\sigma (V^r(a_j))\}}\indic{1}_{\{ k_j=\Phi (V^r(a_j)) \}},\label{eq:Gamma_ZK} 
\end{align}where $W_{Z|U,V}(z|u,v)=\sum_{s\in\mathcal{S}} P_{S|U,V}(s|u,v)W_{Z|U,S}(z|u,s)$. 
\sloppy To bound the third term on the \ac{RHS} of \eqref{eq:Triangle_Inequality_2} by using Pinsker's inequality, it is sufficient to bound $\expec_{C_r}[\kld({\Gamma_{Z_j^r,K_{j-1},T_j,K_j|C_r}||Q_Z^{\otimes r}Q_{K_{j-1}}Q_{T_j}Q_{K_j}})]$ as follows,
\begin{align}
& \expec_{C_r}[\kld({\Gamma_{Z_j^r,K_{j-1},T_j,K_j|C_r}||Q_Z^{\otimes r}Q_{K_{j-1}}Q_{T_j}Q_{K_j}})] \nonumber\\ 
&= \expec_{C_r}\Big[\sum\limits_{(z_j^r,k_{j-1},t_j,k_j)} \Gamma_{Z_j^r,K_{j-1},T_j,K_j|C_r} (z_j^r,k_{j-1},t_j,k_j)\log \Big(\frac{\Gamma_{Z_j^r,K_{j-1},T_j,K_j|C_r}(z_j^r,k_{j-1},t_j,k_j)}{Q_Z^{\otimes r}(z_j^r)Q_{K_{j-1}}(k_{j-1})Q_{T_j}(t_j)Q_{K_j}(k_j)}\Big) \Big] \nonumber\\ 
&= \expec_{C_r}\Bigg[\sum\limits_{(z_j^r,k_{j-1},t_j,k_j)}   \sum\limits_{m_j}\sum\limits_{t_{j - 1}}\sum\limits_{k_{j-2}} \sum\limits_{\ell_j}\sum\limits_{a_j} \frac{1}{2^{r( R + R_t +2R_k + R'+\tilde{R})}}  W_{Z|U,V}^{\otimes r}\big(z_j^r|U^r(m_j,t_{j-1},k_{j-2},\ell_j),V^r(a_j)\big)\nonumber\\
&\times\indic{1}_{\{ t_j=\sigma (V^r(a_j))\}}\indic{1}_{\{ k_j=\Phi (V^r(a_j)) \}}\nonumber\\
&\times\log \Bigg(\frac{\sum\limits_{\tilde{m}_j}\sum\limits_{\tilde{t}_{j - 1}}\sum\limits_{\tilde{k}_{j-2}} \sum\limits_{\tilde{\ell}_j}\sum\limits_{\tilde{a}_j}  W_{Z|U,V}^{\otimes r}\big(z_j^r|U^r(\tilde{m}_j,\tilde{t}_{j-1},\tilde{k}_{j-2},\tilde{\ell}_j),V^r(\tilde{a}_j)\big)  \indic{1}_{\{ t_j=\sigma (V^r(\tilde{a}_j))\}}\indic{1}_{\{k_j=\Phi (V^r(\tilde{a}_j))\}}}{2^{r(R + R_t + R_k + R' +\tilde{R} - R_T- R_K)} Q_Z^{\otimes r}(z_j^r)}  \Bigg)\Bigg]\nonumber\\ 
&\mathop\le\limits^{(a)}\sum\limits_{(z_j^r,k_{j-1},t_j,k_j)} \sum\limits_{m_j}\sum\limits_{t_{j-1}}\sum\limits_{k_{j-2}} \sum\limits_{\ell_j} \sum\limits_{a_j} \frac{1}{2^{r(R + R_t + 2R_k + R'+\tilde{R})}} \sum\limits_{(u^r,v^r)} \Gamma^{(\mathcal{C}_r)}_{U^r,V^r,Z^r}\big(u^r(m_j,t_{j-1},k_{j-2},\ell_j),v^r(a_j),z_j^r\big)\nonumber\\  
&\times\expec_{\sigma(v^r(a_j))}\big[\indic{1}_{\{ t_j=\sigma (v^r(a_j)) \}}\big]\times\expec_{\Phi(v^r(a_j))}\big[\indic{1}_{\{ k_j=\Phi (v^r(a_j))\}}\big] \nonumber\\
&\times\log \expec_{\mathop {\backslash (m_j,t_{j-1},k_{j-2},\ell_j,a_j),}\limits_{\backslash (\sigma (v^r(a_j)),\Phi (v^r(a_j)))} }\Bigg[\frac{1}{2^{r(R + R_t + R_k + R' +\tilde{R} - R_T- R_K)} Q_Z^{\otimes r}(z_j^r)}\nonumber\\
&\times\sum\limits_{\tilde{m}_j}\sum\limits_{\tilde{t}_{j - 1}}\sum\limits_{\tilde{k}_{j-2}} \sum\limits_{\tilde{\ell}_j}\sum\limits_{\tilde{a}_j}  W_{Z|U,V}^{\otimes r}\big(z_j^r|U^r(,\tilde{m}_j,\tilde{t}_{j-1},\tilde{k}_{j-2},\tilde{\ell}_j),V^r(\tilde{a}_j)\big)  \indic{1}_{\{ t_j=\sigma (V^r(\tilde{a}_j))\}}\indic{1}_{\{k_j=\Phi (V^r(\tilde{a}_j))\}}\Bigg]    \nonumber  \\ 
&\mathop\le\limits^{(b)}\sum\limits_{(z_j^r,k_{j-1},t_j,k_j)} \sum\limits_{m_j}\sum\limits_{t_{j-1}}\sum\limits_{k_{j-2}} \sum\limits_{\ell_j} \sum\limits_{a_j} \frac{1}{2^{r(R + R_t + 2R_k + R'+\tilde{R}+R_T+R_K)}} \sum\limits_{(u^r,v^r)} \Gamma^{(\mathcal{C}_r)}_{U^r,V^r,Z^r}\big(u^r(m_j,t_{j-1},k_{j-2},\ell_j),v^r(a_j),z_j^r\big)\nonumber\\ 
&\times\log \frac{1}{2^{r(R + R_t +R_k + R' +\tilde{R} -R_T- R_K)} Q_Z^{\otimes r}(z_j^r)}\Bigg(W_{Z|U,V}^{\otimes r}\big(z_j^r|u^r(m_j,t_{j-1},k_{j-2},\ell_j),v^r(a_j)\big) \nonumber\\
&+  \expec_{\backslash (m_j,t_{j-1},k_{j-2},\ell_j)}\Bigg[\sum\limits_{(\tilde{m}_j,\tilde{t}_{j-1},\tilde{k}_{j-2},\tilde{\ell}_j) \ne (m_j,t_{j-1},k_{j-2},\ell_j)} W_{Z|U,V}^{\otimes r}\big(z_j^r|U^r(\tilde{m}_j,\tilde{t}_{j-1},\tilde{k}_{j-2},\tilde{\ell}_j),v^r(a_j)\big)\Bigg] \nonumber\\
&+  \expec_{\backslash (a_j,\sigma(v^r(a_j)),\Phi(v^r(a_j)))}\Bigg[\sum\limits_{\tilde{a}_{j} \ne a_j}W_{Z|U,V}^{\otimes r}\big(z_j^r|u^r(m_j,t_{j-1},k_{j-2},\ell_j),V^r(\tilde{a}_j)\big)\indic{1}_{\{ t_j=\sigma (V^r(\tilde{a}_j))\}}\indic{1}_{\{ k_j=\Phi (V^r(\tilde{a}_j))\}}\Bigg]         \nonumber\\
&+  \expec_{\mathop {\backslash (m_j,t_{j-1},k_{j-2},\ell_j,a_j),}\limits_{\backslash (\sigma (v^r(a_j)),\Phi (v^r(a_j)))}}\Bigg[ \sum\limits_{(\tilde{m}_j,\tilde{t}_{j-1},\tilde{k}_{j-2},\tilde{\ell}_j) \ne (m_j,t_{j-1},k_{j-2},\ell_j)}\sum\limits_{\tilde{a}_{j}\ne a_j} W_{Z|U,V}^{\otimes r}\big(z_j^r|U^r(\tilde{m}_j,\tilde{t}_{j-1},\tilde{k}_{j-2},\tilde{\ell}_j),V^r(\tilde{a}_j)\big)\nonumber\\
&\times\indic{1}_{\{ t_j=\sigma (V^r(\tilde{a}_j))\}}\indic{1}_{\{ k_j=\Phi (V^r(\tilde{a}_j))\}}\Bigg] \Bigg)\nonumber\\
&\mathop\le\limits^{(c)}\sum\limits_{(z_j^r,k_{j-1},t_j,k_j)} \sum\limits_{m_j}\sum\limits_{t_{j-1}}\sum\limits_{k_{j-2}} \sum\limits_{\ell_j} \sum\limits_{a_j} \frac{1}{2^{r(R + R_t + 2R_k + R'+\tilde{R}+R_T+R_K)}} \sum\limits_{(u^r,v^r)} \Gamma^{(\mathcal{C}_r)}_{U^r,V^r,Z^r}\big(u^r(m_j,t_{j-1},k_{j-2},\ell_j),v^r(a_j),z_j^r\big)\nonumber\\
&\times\log \Bigg(\frac{W_{Z|U,V}^{\otimes r}\big(z_j^r|u^r(m_j,t_{j-1},k_{j-2},\ell_j),v^r(a_j)\big)}{2^{r(R + R_t +R_k + R' +\tilde{R} -R_T- R_K)} Q_Z^{\otimes r}(z_j^r)} \nonumber\\
& + \sum\limits_{(\tilde{m}_j,\tilde{t}_{j-1},\tilde{k}_{j-2},\tilde{\ell}_j) \ne (m_j,t_{j-1},k_{j-2},\ell_j)} {\frac{ W_{Z|V}^{\otimes r}\big(z_j^r|v^r(a_j)\big)}{{{2^{r( R + R_t +R_k + R' +\tilde{R}-R_T- R_K)}} Q_Z^{\otimes r}(z_j^r)}}} + \sum\limits_{\tilde{a}_j \ne a_j}\frac{W_{Z|U}^{\otimes r}\big(z_j^r|u^r(m_j,t_{j-1},k_{j-2},\ell_j)\big)}{2^{r( R + R_t +R_k + R'+\tilde{R})} Q_Z^{\otimes r}(z_j^r)}  +1 \Bigg)\nonumber\\
&\le\sum\limits_{(z_j^r,k_{j-1},t_j,k_j)} \sum\limits_{m_j}\sum\limits_{t_{j-1}}\sum\limits_{k_{j-2}} \sum\limits_{\ell_j} \sum\limits_{a_j} \frac{1}{2^{r(R + R_t + 2R_k + R'+\tilde{R}+R_T+R_K)}} \sum\limits_{(u^r,v^r)} \Gamma^{(\mathcal{C}_r)}_{U^r,V^r,Z^r}\big(u^r(m_j,t_{j-1},k_{j-2},\ell_j),v^r(a_j),z_j^r\big)\nonumber\\
&\times\log \Bigg(\frac{W_{Z|U,V}^{\otimes r}\big(z_j^r|u^r(m_j,t_{j-1},k_{j-2},\ell_j),v^r(a_j)\big)}{2^{r(R + R_t +R_k + R' +\tilde{R} -R_T- R_K)} Q_Z^{\otimes r}(z_j^r)} +  \frac{W_{Z|V}^{\otimes r}\big(z_j^r|v^r(a_j)\big)}{{{2^{r( \tilde{R}-R_T- R_K)}} Q_Z^{\otimes r}(z_j^r)}} + \frac{W_{Z|U}^{\otimes r}\big(z_j^r|u^r(m_j,t_{j-1},k_{j-2},\ell_j)\big) }{2^{r( R + R_t +R_k + R')} Q_Z^{\otimes r}(z_j^r)} +1 \Bigg)        \nonumber\\
&\triangleq \Psi_1 + \Psi_2,\label{eq:KLD_ZK}
\end{align}
where $(a)$ follows from Jensen's inequality, $(b)$ and $(c)$ hold because $\indic{1}_{\{\cdot\}}\leq 1$. We defined $\Psi_1$ and $\Psi_2$ as
\begin{align}
&{\Psi _1} = \sum\limits_{(k_{j-1},t_j,k_j)}  \sum\limits_{m_j} \sum\limits_{t_{j-1}}\sum\limits_{k_{j-2}} \sum\limits_{\ell_j}\sum\limits_{a_j} \frac{1}{{{2^{r(R + R_t +2R_k + R'+\tilde{R}+R_T+R_K)}}}}\nonumber\\
&\times\sum\limits_{(u^r(m_j,t_{j-1},k_{j-2},\ell_j),v^r(a_j),z_j^r) \in \mathcal{T}_\epsilon ^{(n)}} \Gamma^{(\mathcal{C}_r)}_{U^r,V^r,Z^r}\big(u^r(m_j,t_{j-1},k_{j-2},\ell_j),v^r(a_j),z_j^r\big)\nonumber\\
&\times\log \Bigg(\frac{W_{Z|U,V}^{\otimes r}\big(z_j^r|u^r(m_j,t_{j-1},k_{j-2},\ell_j),v^r(a_j)\big)}{2^{r(R + R_t +R_k + R' +\tilde{R} -R_T- R_K)} Q_Z^{\otimes r}(z_j^r)} +  \frac{W_{Z|V}^{\otimes r}\big(z_j^r|v^r(a_j)\big)}{{{2^{r( \tilde{R}-R_T- R_K)}} Q_Z^{\otimes r}(z_j^r)}} + \frac{W_{Z|U}^{\otimes r}\big(z_j^r|u^r(m_j,t_{j-1},k_{j-2},\ell_j)\big) }{2^{r( R + R_t +R_k + R')} Q_Z^{\otimes r}(z_j^r)} +1 \Bigg)        \nonumber\\
& \le \log \Bigg(\frac{2^{r(R_T+R_K)}\times 2^{ - r(1 - \epsilon )\ent(Z|U,V)}}{2^{r( R + R_t +R_k + R'+\tilde{R})}\times 2^{ - r(1 + \epsilon )\ent(Z)}} + \frac{2^{r(R_T+R_K)}\times 2^{ - r(1 - \epsilon )\ent(Z|V)}}{2^{r\tilde{R}}\times 2^{ - r(1 + \epsilon )\ent(Z)}} + \frac{2^{ - r(1 - \epsilon )\ent(Z|U)}}{2^{r( R + R_t +R_k + R')}\times 2^{ - r(1 + \epsilon )\ent(Z)}} + 1\Bigg)\label{eq:Psi_1}\\
&{\Psi _2} = \sum\limits_{(k_{j-1},t_j,k_j)}  \sum\limits_{m_j} \sum\limits_{t_{j-1}}\sum\limits_{k_{j-2}} \sum\limits_{\ell_j}\sum\limits_{a_j} \frac{1}{{{2^{r(R + R_t +2R_k + R'+\tilde{R}+R_T+R_K)}}}}\nonumber\\
&\times\sum\limits_{(u^r(m_j,t_{j-1},k_{j-2},\ell_j),v^r(a_j),z_j^r) \notin \mathcal{T}_\epsilon ^{(n)}} \Gamma^{(\mathcal{C}_r)}_{U^r,V^r,Z^r}\big(u^r(m_j,t_{j-1},k_{j-2},\ell_j),v^r(a_j),z_j^r\big)\nonumber\\
&\times\log \Bigg(\frac{W_{Z|U,V}^{\otimes r}\big(z_j^r|u^r(m_j,t_{j-1},k_{j-2},\ell_j),v^r(a_j)\big)}{2^{r(R + R_t +R_k + R' +\tilde{R} -R_T- R_K)} Q_Z^{\otimes r}(z_j^r)} +  \frac{W_{Z|V}^{\otimes r}\big(z_j^r|v^r(a_j)\big)}{{{2^{r( \tilde{R}-R_T- R_K)}} Q_Z^{\otimes r}(z_j^r)}} + \frac{W_{Z|U}^{\otimes r}\big(z_j^r|u^r(m_j,t_{j-1},k_{j-2},\ell_j)\big) }{2^{r( R + R_t +R_k + R')} Q_Z^{\otimes r}(z_j^r)} +1 \Bigg)        \nonumber\\
&\le 2|V||U||Z|e^{ - r\epsilon^2 {\mu _{S,V,U,Z}}}r\log \Big(\frac{3}{\mu_Z} + 1\Big). \label{eq:Psi_2}
\end{align}
In \eqref{eq:Psi_2} $\mu_{V,U,Z}=\min\limits_{(v,u,z)\in(\mathcal{V},\mathcal{U},\mathcal{Z})}P_{V,U,Z}(v,u,z)$ and $\mu_Z=\min\limits_{z\in\mathcal{Z}}P_Z(z)$. When $r\to\infty$ then $\Psi_2\to 0$ and $\Psi_1$ goes to zero when $r$ grows if
\begin{subequations}\label{eq:KLD_Bound}
\begin{align}
    R +R_t+R_k+ R'+\tilde{R}-R_T-R_K&>\mi(U,V;Z),\label{eq:KLD_Bound_1}\\
    \tilde{R}-R_T-R_K&>\mi(V;Z),\label{eq:KLD_Bound_2}\\
    R +R_t+R_k+ R'&>\mi(U;Z).\label{eq:KLD_Bound_3}
\end{align}
\end{subequations}

\emph{{Decoding and Error Probability Analysis:}}
At the end of block $j\in\intseq{1}{B}$, using its knowledge of the key $k_{j-2}$ generated from the block $j-2$, the receiver finds a unique triple $(\hat{m}_j,\hat{t}_{j-1},\hat{\ell}_j)$ such that $\big(u^r(\hat{m}_j,\hat{t}_{j-1},k_{j-2},\hat{\ell}_j),y_j^r\big)\in\mathcal{T}_{\epsilon}^{(r)}$. 
To bound the probability of error at the encoder and the decoder we need the following lemma,
\begin{lemma}[Typical With High Probability]
\label{lemma:Typicality}
If $(R',\tilde{R})\in\mathds{R}_+^2$ satisfies \eqref{eq:MAC_Covering_Lemma_1} to \eqref{eq:MAC_Covering_Lemma_3}, then for any $(m_j,t_{j-1},k_{j-2})\in(\mathcal{M}_j,\mathcal{T}_{j-1},\mathcal{K}_{j-2})$ and $\epsilon>0$, we have
\begin{align}
    \expec_{C_r}\Prob_P\Big(\big(U^r(m_j,t_{j-1},k_{j-2},L_j),V^r(A_j),S_j^r\big)\notin\mathcal{T}_{\epsilon}^{(n)}|C_r\Big)\underset{r\rightarrow\infty}{\xrightarrow{\hspace{0.2in}}} 0,
\end{align}
\end{lemma}
where $P$ is the induced distribution over the codebook defined in \eqref{eq:Uniformity_Key_NC_CSIT}. 
The proof of Lemma~\ref{lemma:Typicality} is given in Appendix~\ref{Appen_Proof_Typicality_Lemma}.

To analyze the probability of error we define the following error events for $j\in\intseq{1}{B}$,
\begin{subequations}\label{eq:error_Events_CSIT_NC}
\begin{align}
    \mathcal{E} &\triangleq \big\{M\ne\hat{M}\big\},\label{eq:error_Events_CSIT_NC_0}\\
    \mathcal{E}_j &\triangleq \big\{M_j\ne\hat{M}_j\big\},\label{eq:error_Events_CSIT_NC_01}\\
    \mathcal{E}_{1,j} &\triangleq\big\{\big(U^r(M_j,T_{j-1},K_{j-2},L_j),S_j^r\big)\notin \mathcal{T}_{\epsilon_1}^{(n)}(U,S)\big\},\label{eq:error_Events_CSIT_NC_1}\\
    \mathcal{E}_{2,j} &\triangleq \big\{\big(U^r(M_j,T_{j-1},K_{j-2},L_j),Y_j^r\big)\notin \mathcal{T}_{\epsilon_2}^{(n)}(U,Y)\big\},\label{eq:error_Events_CSIT_NC_2}\\
    \mathcal{E}_{3,j} &\triangleq \big\{\big(U^r(M_j,T_{j-1},K_{j-2},L_j),Y_j^r\big)\in \mathcal{T}_{\epsilon_2}^{(n)}(U,Y)\twocolbreakquad\mbox{\,for some $m_j\ne M_j$ and $\ell_j\in [1:2^{rR'}]$}\big\}.\label{eq:error_Events_CSIT_NC_3}
\end{align}
\end{subequations}
where $\epsilon_2>\epsilon_1>\epsilon>0$. The probability of error is upper bounded as follows,
\begin{align}
    \Prob(\mathcal{E})=\Prob\Big\{\bigcup\nolimits_{j = 1}^B \mathcal{E}_j\Big\}\leq\sum_{j=1}^B\Prob(\mathcal{E}_j).
\end{align}Now we bound $\Prob(\mathcal{E}_j)$ by using union bound,
\begin{align}
    \Prob(\mathcal{E}_j)\leq \Prob(\mathcal{E}_{1,j})+\Prob(\mathcal{E}_{1,j}^c\cap \mathcal{E}_{2,j})+\Prob(\mathcal{E}_{2,j}^c\cap \mathcal{E}_{3,j}).\label{eq:PError_CSIT_NC}
\end{align}According to Lemma~\ref{lemma:Typicality} the first term on the \ac{RHS} of \eqref{eq:PError_CSIT_NC} vanishes when $r$ grows, and by the law of large numbers the second term on the \ac{RHS} of \eqref{eq:PError_CSIT_NC} vanishes when $r$ grows. Also, according to the law of large numbers and the packing lemma, the last term on the \ac{RHS} of \eqref{eq:PError_CSIT_NC} vanishes when $r$ grows if \cite{ElGamalKim},
\begin{align}
&R +R_t + R' \leq \mi(U;Y).
\label{eq:Decoding_BM_NC_CSIT}
\end{align}
We now analyze the probability of error at the encoder and the decoder for key generation. Let $(A_{j-1},T_{j-1})$ denote the chosen indices at the encoder and $\hat{A}_{j-1}$ and $\hat{T}_{j-1}$ be the estimates of the indices $A_{j-1}$ and $T_{j-1}$ at the decoder. At the end of block $j$, by decoding $U_j^r$, the decoder has access to $\hat{T}_{j-1}$. To find $A_{j-1}$ we define the error event,
\begin{align}
    \mathcal{E}' = \Big\{ \Big(V_{j-1}^r(\hat{A}_{j-1}),S_{j-1}^r,U_{j-1}^r,Y_{j-1}^r\Big) \notin \mathcal{T}_\epsilon ^{(n)}\Big\},
\end{align}and consider the events,
\begin{subequations}\label{eq:Error_Events}
\begin{align}
    \mathcal{E}'_1 &= \Big\{ \Big(V_{j-1}^r(a_{j-1}),S_{j-1}^r\Big) \notin \mathcal{T}_{\epsilon'}^{(n)} \,\,\mbox{for all}\,\, a_{j-1}\in\intseq{1}{2^{r\tilde{R}}}\Big\},\\
    \mathcal{E}'_2 &= \Big\{ \Big(V_{j-1}^r(A_{j-1}),S_{j-1}^r,U_{j-1}^r,Y_{j-1}^r\Big) \notin \mathcal{T}_\epsilon ^{(n)}\Big\},\\
    \mathcal{E}'_3 &= \Big\{ \Big(V_{j-1}^r(\tilde{a}_{j-1}),U_{j-1}^r,Y_{j-1}^r\Big) \in \mathcal{T}_{\epsilon}^{(n)} \,\,\mbox{for some}\,\, \tilde{a}_{j-1}\in\mathcal{B}(\hat{T}_{j-1}),\tilde{a}_{j-1}\ne A_{j-1}\Big\},
\end{align}where $\epsilon>\epsilon'>0$. 
\end{subequations}By the union bound we have,
\begin{align}
    P(\mathcal{E}')\leq P(\mathcal{E}'_1) + P(\mathcal{E}_1^{\prime c}\cap\mathcal{E}'_2) + P(\mathcal{E}'_3)\label{eq:Error_Event_NonCausal_CSIT}.
\end{align}
According to Lemma~\ref{lemma:Typicality} the first term on the \ac{RHS} of \eqref{eq:Error_Event_NonCausal_CSIT} vanishes when $r$ grows if we have \eqref{eq:MAC_Covering_Lemma}. Following the steps in \cite[Sec.~11.3.1]{ElGamalKim}, the last two terms on the right hand side of \eqref{eq:Error_Event_NonCausal_CSIT} go to zero when $r$ grows if,
\begin{subequations}\label{eq:Error_Analysis_Encoder_NC_CSIT}
\begin{align}
    \tilde{R}&>\mi(V;S),\label{eq:Error_Analysis_Encoder_NC_CSIT_1}\\
    \tilde{R}-R_t&<\mi(V;U,Y).\label{eq:Error_Analysis_Encoder_NC_CSIT_2}
\end{align}
\end{subequations}


The region in Theorem~\ref{thm:Acivable_Rate_CSIT_WZ_NC} is derived by remarking that the scheme requires $R_K+R_T\geq R_k+R_t$ and applying Fourier-Motzkin to \eqref{eq:MAC_Covering_Lemma} and  \eqref{eq:KLD_Bound}, \eqref{eq:Decoding_BM_NC_CSIT}, and \eqref{eq:Error_Analysis_Encoder_NC_CSIT}.

\section{Proof of Lemma~\ref{lemma:Typicality}}
\label{Appen_Proof_Typicality_Lemma}
For a fix $\epsilon>0$ consider the \ac{PMF} $\Gamma$ defined in \eqref{eq:Ideal_PMF}. With respect to the random experiment described by $\Gamma$ we have
\begin{align}\label{eq:Gamma_Typicality}
    \expec_{C_r}\Prob_{\Gamma}\Big(\big(U^r(m_j,t_{j-1},k_{j-2},L_j),V^r(A_j),S_j^r\big)\notin\mathcal{T}_{\epsilon}^{(n)}|C_r\Big)\underset{r\rightarrow\infty}{\xrightarrow{\hspace{0.2in}}} 0,
\end{align}since $U^r(m_j,t_{j-1},k_{j-2},L_j)\sim P_U^r$ and $V^r(A_j)\sim P_V^r$ for every $(m_j,t_{j-1},k_{j-2})\in(\mathcal{M}_j,\mathcal{T}_{j-1},\mathcal{K}_{j-2})$ and $S_j^r$ is derived by passing $(U^r(m_j,t_{j-1},k_{j-2},L_j),V^r(A_j))$ through the \ac{DMC} $P_{S|U,V}^{\otimes r}$. Therefore \eqref{eq:Gamma_Typicality} holds by weak law of large numbers. We also have 
\begin{align}
 &\expec_{C_r}||P_{U^r,V^r,S_j^r|C_r}-\Gamma_{U^r,V^r,S_j^r|C_r}||_1
 \twocolbreak \nonumber\\ &\leq\expec_{C_r}||P_{M_j,T_{j-1},K_{j-2},S_j^r,L_j,A_j,U^r,V^r,Z_j^r,K_{j-1},T_j,K_j|C_r}- \Gamma_{M_j,T_{j-1},K_{j-2},S_j^r,L_j,A_j,U^r,V^r,Z_j^r,K_{j-1},T_j,K_j|C_r}||_1\underset{r\rightarrow\infty}{\xrightarrow{\hspace{0.2in}}} 0,
 \label{eq:TV_Ineq_app}
\end{align}
where   the \ac{RHS} of \eqref{eq:TV_Ineq_app} vanishes when $r$ grows because of \eqref{eq:General_Distribution_Expansion_2}.

We now define $g_n:\mathcal{U}^r\times\mathcal{V}^r\times\mathcal{S}_j^r\to\mathbb{R}$ as $g_n(u^r,v^r,s_j^r)\triangleq\indic{1}_{\{(u^r,v^r,s_j^r)\notin\mathcal{T}_\epsilon^{(n)}\}}$. We now have
\begin{align}
    &\expec_{C_r}\Prob_P\Big(\big(U^r(m_j,t_{j-1},k_{j-2},L_j),V^r(A_j),S_j^r\big)\notin\mathcal{T}_{\epsilon}^{(n)}|C_r\Big)\nonumber\\
    &=\expec_{C_r}\expec_{P}\Big[g_n(U^r(m_j,t_{j-1},k_{j-2},L_j),V^r(A_j),S_j^r)|C_r\Big]\nonumber\\
    &\leq\expec_{C_r}\expec_{\Gamma}\Big[g_n(U^r(m_j,t_{j-1},k_{j-2},L_j),V^r(A_j),S_j^r)|C_r\Big]\nonumber\\
    &+\expec_{C_r}\Big|\expec_{P}\Big[g_n(U^r(m_j,t_{j-1},k_{j-2},L_j),V^r(A_j),S_j^r)|C_r\Big]-\expec_{\Gamma}\Big[g_n(U^r(m_j,t_{j-1},k_{j-2},L_j),V^r(A_j),S_j^r)|C_r\Big]\Big|\nonumber\\
    &\mathop \leq \limits^{(a)} \expec_{C_r}\expec_{\Gamma}\Big[g_n(U^r(m_j,t_{j-1},k_{j-2},L_j),V^r(A_j),S_j^r)|C_r\Big]+\expec_{C_r}||P_{U^r,V^r,S_j^r|C_r}-\Gamma_{U^r,V^r,S_j^r|C_r}||_1
    \label{eq:Typicality_triangle}
\end{align}where $(a)$ follows from \cite[Property~1]{AllertonCuffSong} for $g_n$ being bounded by $1$. From \eqref{eq:Gamma_Typicality} and \eqref{eq:TV_Ineq_app} the \ac{RHS} of \eqref{eq:Typicality_triangle} vanishes when $r$ grows.

\section{Proof of Theorem~\ref{thm:Simple_Acivable_Region_CSIT_NC}}
\label{sec:Proof_Simple_Inner_Bound_For_CSIT_NC}
\sloppy Fix $P_{U|S}(u|s)$, $x(u,s)$, and $\epsilon_1>\epsilon_2>0$ such that, $P_Z = Q_0$.

\emph{Codebook Generation:}
Let $C_n\triangleq\{U^n(m,\ell)\}_{(m,\ell)\in\mathcal{M}\times\mathcal{L}}$, where $\mathcal{M}\in\intseq{1}{2^{nR}}$ and $\mathcal{L}\in\intseq{1}{2^{nR'}}$, be a random codebook consisting of independent random sequences each generated according to $\prod\nolimits_{i = 1}^n P_U(u_i)$. We denote a realization of $C_n$ by $\CodeBook_n\triangleq\{u^n(m,\ell)\}_{(m,\ell)\in\mathcal{M}\times\mathcal{L}}$. 
We define an ideal \ac{PMF} for codebook $\mathcal{C}_n$, as an approximate distribution to facilitate the analysis
\begin{align}
    &\Gamma_{M,L,U^n,S^n,Z^n}^{(\mathcal{C}_n)}(m,\ell,\tilde{u}^n,s^n,z^n) =2^{-n(R + R')}\indic{1}_{\{ \tilde{u}^n=u^n(m,\ell)\}}
     P_{S|U}^{\otimes n}(s^n|\tilde{u}^n) W_{Z|U,S}^{\otimes n}(z^n|\tilde{u}^n,s^n),\label{eq:Ideal_PMF_Simple_C}
\end{align}where $W_{Z|U,S}$ is the marginal distribution $W_{Z|U,S}=\sum_{x\in\mathcal{X}}\indic{1}_{\{x=x(u,s)\}}W_{Z|X,S}$ and $P_{S|U}$ is defined as follows,
\begin{align}
    P_{S|U}(s|u)&\triangleq\frac{P_{S,U}(s,u)}{P_{U}(u)}=\frac{Q_S(s)P_{U|S}(u|s)}{\sum_{s\in\mathcal{S}}Q_S(s)P_{U|S}(u|s)}.\label{eq:S_Dist_NC_Simple}
\end{align}

\emph{Encoding:}
To send the message $m$ the encoder chooses $\ell$ according to 
\begin{align}
 f(\ell|s^n,m)=\frac{P_{S|U}^{\otimes n}\big(s^n|u^n(m,\ell)\big)}{\sum\limits_{\ell'\in\intseq{1}{2^{nR'}}}P_{S|U}^{\otimes n}\big(s_j^n|u^n(m,\ell')\big)},
\label{eq:Likelihood_Encoder_BM_Simple_C}
\end{align}
where $P_{S|U}$ is defined in~\eqref{eq:S_Dist_NC_Simple}. Based on $(m,\ell)$ the encoder computes $u^n(m,\ell)$ and transmits codeword $x^n$, where $x_i=x(u_i(m,\ell),s_i)$. 

For a fixed codebook $\mathcal{C}_n$, the induced joint distribution over the codebook is as follows
\begin{align}
    \twocolalign &P_{M,S^n,L,U^n,Z^n}^{(\mathcal{C}_n)}(m,s^n,\ell,\tilde{u}^n,z^n)= 2^{-nR} Q_S^{\otimes n}(s^n) f(\ell|s^n,m)\indic{1}_{\{ \tilde{u}^n=u^n(m,\ell)\}} W_{Z|U,S}^{\otimes n}(z^n|\tilde{u}^n,s^n).\label{eq:P_Dist_App}
\end{align}

\emph{{Covert Analysis:}} 
We now show that this coding scheme guarantees that
\begin{align}
    \expec_{C_n}\big[\kld(P_{Z^n|C_n} || Q_Z^{\otimes n} )\big]\underset{n\rightarrow\infty}{\xrightarrow{\hspace{0.2in}}} 0.\label{eq:QZ_CSIT_NC}
\end{align}Then we choose $P_{U|S}$ and $x(u,s)$ such that it satisfies $Q_Z=Q_0$. By Combining Lemma~\ref{lemma:TV_KLD} and the triangle inequality a sufficient condition for \eqref{eq:QZ_CSIT_NC} is to show that the \ac{RHS} of the following inequality vanishes,
\begin{align}
    \expec_{C_n}||P_{Z^n|C_n}-Q_Z^{\otimes n}||_1 &\leq\expec_{C_n}||P_{Z^n|C_n}-\Gamma_{Z^n|C_n}||_1  + \expec_{C_n}||\Gamma_{Z^n|C_n}-Q_Z^{\otimes n}||_1.\label{eq:Triangle_Inequality_2_App}
\end{align} 
By \cite[Corollary~VII.5]{Cuff13} the second term on the \ac{RHS} of \eqref{eq:Triangle_Inequality_2_App} vanishes when $n$ grows if
\begin{align}
    R+R'>\mi(U;Z).\label{eq:res_Simple_CSIT_NC}
\end{align}
To bound the first term on the \ac{RHS} of \eqref{eq:Triangle_Inequality_2_App} we have,
\begin{subequations}\label{eq:P_Gamma_CSIT_Simple_NC}
\begin{align}
  \Gamma_{M}^{(\mathcal{C}_n)} &= 2^{-nR} = P_{M}^{(\mathcal{C}_n)},\label{eq:Uniformity_M_2_App}\\
  \Gamma_{L|M,S^n}^{(\mathcal{C}_n)} &= f(\ell|s^n,m) = P_{L|M,S^n}^{(\mathcal{C}_n)},\label{eq:Likelihood_Encoder_Equality_2_App}\\
  \Gamma_{U^n|M,S^n,L}^{(\mathcal{C}_n)} &= \indic{1}_{\{\tilde{u}^n = u^n(m,\ell)\}} \twocolbreak= P_{U^n|M,S^n,L}^{(\mathcal{C}_n)},\label{eq:eq4_App}\\
    \Gamma_{Z^n|M,S^n,L,U^n}^{(\mathcal{C}_n)} &= W_{Z|U,S}^{\otimes n}  = P_{Z^n|M,S^n,L,U^n}^{(\mathcal{C}_n)},\label{eq:Channel_Equality_2_App}
\end{align}
\end{subequations}
where \eqref{eq:Likelihood_Encoder_Equality_2_App} follows from \eqref{eq:Likelihood_Encoder_BM_Simple_C}. Hence,
\begin{align}
 \expec_{C_n}||P_{Z^n|C_n}-\Gamma_{Z^n|C_n}||_1
  &\leq\expec_{C_n}||P_{M,S^n,L,U^n,Z^n|C_n}- \Gamma_{M,S^n,L,U^n,Z^n|C_n}||_1 \nonumber\\
&\mathop = \limits^{(a)}\expec_{C_n}||P_{S^n,L,U^n,Z^n|M=1,C_n}- \Gamma_{S^n,L,U^n,Z^n|M=1,C_n}||_1 \nonumber\\
&\mathop = \limits^{(b)} \expec_{C_n}||Q_S^{\otimes n} - \Gamma_{S^n|M=1,C_n}||_1,\label{eq:General_Distribution_Expansion_2_App}
\end{align}where $(a)$ follows from \eqref{eq:Likelihood_Encoder_Equality_2_App}-\eqref{eq:Channel_Equality_2_App} and $(b)$ follows from the symmetry of the codebook construction with respect to $M$ and \eqref{eq:Uniformity_M_2_App}. Based on the soft covering lemma \cite[Corollary~VII.5]{Cuff13} the \ac{RHS} of \eqref{eq:General_Distribution_Expansion_2_App} vanishes if
\begin{align}
\label{eq:State_Dependent_Covering_Lemma_2_Simple_CSIT_NC}
    R' &> \mi(U;S).
\end{align}

\emph{{Decoding and Error Probability Analysis:}}
To decode the message $m$, the receiver finds a unique pair $(\hat{m},\hat{\ell})$ such that $\big(u^r(\hat{m},\hat{\ell}),y^r\big)\in\mathcal{T}_{\epsilon}^{(r)}$. 
To analyze the probability of error we define the following error events,
\begin{subequations}\label{eq:error_Events_CSIT_NC_Simple}
\begin{align}
    \mathcal{E} &\triangleq \big\{M\ne\hat{M}\big\},\label{eq:error_Events_CSIT_NC_0_Simple}\\
    \mathcal{E}_1 &\triangleq\big\{\big(U^r(M,L),S^r\big)\notin \mathcal{T}_{\epsilon_1}^{(n)}(U,S)\big\},\label{eq:error_Events_CSIT_NC_1_Simple}\\
    \mathcal{E}_2 &\triangleq \big\{\big(U^r(M,L),Y^r\big)\notin \mathcal{T}_{\epsilon_2}^{(n)}(U,Y)\big\},\label{eq:error_Events_CSIT_NC_2_Simple}\\
    \mathcal{E}_3 &\triangleq \big\{\big(U^r(M,L),Y^r\big)\in \mathcal{T}_{\epsilon_2}^{(n)}(U,Y)\twocolbreakquad\mbox{\,for some $m\ne M$ and $\ell\in [1:2^{rR'}]$}\big\},\label{eq:error_Events_CSIT_NC_3_Simple}
\end{align}
\end{subequations}
where $\epsilon_2>\epsilon_1>0$. 
Now we bound $\Prob(\mathcal{E})$ by using union bound,
\begin{align}
    \Prob(\mathcal{E})\leq \Prob(\mathcal{E}_1)+\Prob(\mathcal{E}_1^c\cap \mathcal{E}_2)+\Prob(\mathcal{E}_2^c\cap \mathcal{E}_3).\label{eq:PError_CSIT_NC_Simple}
\end{align}Similar to Lemma~\ref{lemma:Typicality} one can show that the first term on the \ac{RHS} of \eqref{eq:PError_CSIT_NC_Simple} vanishes when $r$ grows, and by the law of large numbers the second term on the \ac{RHS} of \eqref{eq:PError_CSIT_NC_Simple} vanishes when $r$ grows. Also, according to the law of large numbers and the packing lemma, the last term on the \ac{RHS} of \eqref{eq:PError_CSIT_NC_Simple} vanishes when $r$ grows if \cite{ElGamalKim},
\begin{align}
&R + R' \leq \mi(U;Y).
\label{eq:Decoding_Simple_CSIT_NC}
\end{align}



The region in Theorem~\ref{thm:Simple_Acivable_Region_CSIT_NC} is derived by applying Fourier-Motzkin to \eqref{eq:res_Simple_CSIT_NC}, \eqref{eq:State_Dependent_Covering_Lemma_2_Simple_CSIT_NC}, and  \eqref{eq:Decoding_Simple_CSIT_NC}.

\section{Proof of Theorem~\ref{thm:Upper_Bound_For_CSIT_NC}}
\label{sec:Proof_Upper_Bound_For_CSIT_NC}
Consider any sequence of length-$n$ codes for a state-dependent channel with channel state available non-causally only at the transmitter such that $P_e^{(n)}\leq\epsilon_n$ and $\kld(P_{Z^n}||Q_0^{\otimes n})\leq\delta$ with $\lim_{n\to\infty}\epsilon_n=0$. Note that the converse is consistent with the model and does \emph{not} require $\delta$ to vanish. The following lemma, a version of which with variational distance can be found in \cite[Lemma~VI.3]{Cuff13}, will prove useful. 

\begin{lemma}
\label{lemma:Resolvability_Properties}
If $\kld(P_{Z^n}||Q_0^{\otimes n})\leq\delta$, then $\sum\nolimits_{i=1}^n\mi(Z_i;Z^{i-1})\leq\delta$ and $\sum\nolimits_{i=1}^n\mi(Z_i;Z_{i+1}^n)\leq\delta$. In addition if $T\in\intseq{1}{n}$ is an independent variable uniformly distributed, then $\mi(T;Z_T)\leq\nu$, where $\nu\triangleq\frac{\delta}{n}$.
\end{lemma}
Note that Lemma~\ref{lemma:Resolvability_Properties} is slightly different from~\cite[Lemma VI.3]{Cuff13}, as the upper bounds are tighter and do not include a factor of $n$. This is a consequence of using a constraint based on relative entropy instead of total variation. This is crucial in what follows as we do no necessarily require $\delta\to 0$. The details of the proof are available in Appendix~\ref{sec:Proof_Lemma_Resolvability_Properties}.

\subsection{Epsilon Rate Region}
We first define a region $\mathcal{S}_{\epsilon}$ for $\epsilon>0$ that expands the region defined in~\eqref{eq:finalregion_CSIT_NC} and~\eqref{eq:finalconstraint_CSIT_NC} as follows.
\begin{subequations}\label{eq:Epsilon_Rate_Region_CSIT_NC}
\begin{align}
\calS_\epsilon\eqdef \big\{R\geq 0: \exists P_{U,V,S,X,Y,Z}\in\calD_\epsilon: R\leq \min\{\mi(U;Y)-\mi(U;S),\mi(U,V;Y)-\mi(U;S|V)\} + \epsilon \big\}\label{eq:Sepsilon_CSIT_NC}
\end{align}
where 
\begin{align}
  \calD_\epsilon = \left.\begin{cases}P_{U,V,S,X,Y,Z}:\\
P_{U,V,S,X,Y,Z}=Q_SP_{UV|S}\indic{1}_{\big\{X=X(U,S)\big\}}W_{Y,Z|X,S}\\
\mathbb{D}\left(P_Z\Vert Q_0\right) \leq \epsilon\\
\min\{\mi(U;Y)-\mi(U;S),\mi(U,V;Y)-\mi(U;S|V)\} \geq \mi(V;Z)-\mi(V;S) - 3\epsilon\\
\max\{\card{\calU},\card{\calV}\}\leq \card{\calX}+3
\end{cases}\right\},\label{eq:Depsilon_CSIT_NC}
\end{align}
\end{subequations}where $\epsilon\triangleq\max\{\epsilon_n,\nu\geq\frac{\delta}{n}\}$. 
We next show that if a rate $R$ is achievable then $R\in\mathcal{S}_{\epsilon}$ for any $\epsilon>0$. 

For any $\epsilon_n>0$, we start by upper bounding $nR$ using standard techniques.
\begin{align}
nR &=\ent(M)\nonumber\\
&\mathop \le \limits^{(a)}\mi(M;Y^n)+n\epsilon_n\nonumber\\
&= \sum\limits_{i = 1}^n \mi(M;Y_i|Y^{i-1})  + n\epsilon_n \nonumber\\
&\leq \sum\limits_{i = 1}^n \mi(M,Y^{i-1};Y_i)  + n\epsilon_n \nonumber\\
&= \sum\limits_{i = 1}^n \big[\mi(M,Y^{i-1},S_{i+1}^n;Y_i)-\mi(S_{i+1}^n;Y_i|M,Y^{i-1})\big]  + n\epsilon_n \nonumber\\
&\mathop = \limits^{(b)} \sum\limits_{i = 1}^n \big[\mi(M,Y^{i-1},S_{i+1}^n;Y_i)-\mi(Y^{i-1};S_i|M,S_{i+1}^n)\big]  + n\epsilon_n \nonumber\\
&\mathop = \limits^{(c)} \sum\limits_{i = 1}^n \big[\mi(M,Y^{i-1},S_{i+1}^n;Y_i)-\mi(M,Y^{i-1},S_{i+1}^n;S_i)\big]  + n\epsilon_n \nonumber\\
&\mathop = \limits^{(d)} \sum\limits_{i = 1}^n [\mi(U_i;Y_i) - \mi(U_i;S_i)]  + n\epsilon_n \nonumber\\
&=n\sum\limits_{i = 1}^n \frac{1}{n}\big[\mi(U_i;Y_i|T=i) - \mi(U_i;S_i|T=i)\big]  + n\epsilon_n\nonumber\\
&=n\sum\limits_{i = 1}^n \Prob(T=i)\big[\mi(U_i;Y_i|T=i) - \mi(U_i;S_i|T=i)\big]  + n\epsilon_n\nonumber\\
&=n \big[\mi(U_T;Y_T|T) - \mi(U_T;S_T|T)\big]  + n\epsilon_n\nonumber\\
&\mathop = \limits^{(e)}n \big[\mi(U_T;Y_T|T) - \mi(U_T,T;S_T)\big]  + n\epsilon_n\nonumber\\
&\le \big[\mi(U_T,T;Y_T) - \mi(U_T,T;S_T)\big]  + n\epsilon_n\nonumber\\
&\mathop = \limits^{(f)} n\big[\mi(U;Y)-\mi(U;S)\big] + n\epsilon_n\nonumber\\
&\mathop \leq \limits^{(g)} n\big[\mi(U;Y)-\mi(U;S)\big] + n\epsilon
\label{eq:Upper_Bound_on_MRate_CSIT_NC_2}
\end{align}where 
\begin{itemize}
    \item[$(a)$] follows from Fano's inequality for $n$ large enough and the fact that conditioning does not increase entropy;
    \item[$(b)$] follows from Csisz\'ar-K\"{o}rner sum identity \cite[Lemma~7]{BCC:IT78};
    \item[$(c)$] follows since $S_i$ is independent of $(M,S_{i+1}^n)$;
    \item[$(d)$] follows by defining $U_i\triangleq(M,Y^{i-1},S_{i+1}^n)$;
    \item[$(e)$] follows from the independence of $S_T$ and $T$;
    \item[$(f)$] follows by defining $U=(U_T,T)$, $Y=Y_T$, and $S=S_T$;
    \item[$(g)$] follows by defining $\epsilon\triangleq\max\{\epsilon_n,\nu\geq\frac{\delta}{n}\}$.
\end{itemize}We also have,
\begin{align}
nR &=\ent(M)\nonumber\\
&\mathop \le \limits^{(a)}\mi(M;Y^n)+n\epsilon_n\nonumber\\
&= \sum\limits_{i = 1}^n \mi(M;Y_i|Y^{i-1})  + n\epsilon_n \nonumber\\
&\leq \sum\limits_{i = 1}^n \mi(M,Y^{i-1},Z^{i-1};Y_i)  + n\epsilon_n \nonumber\\
&= \sum\limits_{i = 1}^n \big[\mi(M,Y^{i-1},Z^{i-1},S_{i+1}^n;Y_i)-\mi(S_{i+1}^n;Y_i|M,Y^{i-1},Z^{i-1})\big]  + n\epsilon_n \nonumber\\
&\mathop = \limits^{(b)} \sum\limits_{i = 1}^n \big[\mi(M,Y^{i-1},Z^{i-1},S_{i+1}^n;Y_i)-\mi(Y^{i-1};S_i|M,S_{i+1}^n,Z^{i-1})\big]  + n\epsilon_n \nonumber\\
&\mathop = \limits^{(c)} \sum\limits_{i = 1}^n [\mi(U_i,V_i;Y_i) - \mi(U_i;S_i|V_i)]  + n\epsilon_n \nonumber\\
&=n\sum\limits_{i = 1}^n \frac{1}{n}\big[\mi(U_i,V_i;Y_i|T=i) - \mi(U_i;S_i|V_i,T=i)\big]  + n\epsilon_n\nonumber\\
&=n\sum\limits_{i = 1}^n \Prob(T=i)\big[\mi(U_i,V_i;Y_i|T=i) - \mi(U_i;S_i|V_i,T=i)\big]  + n\epsilon_n\nonumber\\
&=n \big[\mi(U_T,V_T;Y_T|T) - \mi(U_T;S_T|V_T,T)\big]  + n\epsilon_n\nonumber\\
&\le \big[\mi(U_T,V_T,T;Y_T) - \mi(U_T;S_T|V_T,T)\big]  + n\epsilon_n\nonumber\\
&\mathop = \limits^{(d)} n\big[\mi(U,V;Y)-\mi(U;S|V)\big] + n\epsilon_n\nonumber\\
&\mathop \leq \limits^{(e)} n\big[\mi(U,V;Y)-\mi(U;S|V)\big] + n\epsilon,
\label{eq:Upper_Bound_on_MRate_CSIT_NC_2_SEC}
\end{align}where 
\begin{itemize}
    \item[$(a)$] follows from Fano's inequality for $n$ large enough and the fact that conditioning does not increase entropy;
    \item[$(b)$] follows from Csisz\'ar-K\"{o}rner sum identity \cite[Lemma~7]{BCC:IT78};
    \item[$(c)$] follows by defining $U_i\triangleq(M,Y^{i-1},S_{i+1}^n)$ and $V_i\triangleq(M,Z^{i-1},S_{i+1}^n)$;
    \item[$(d)$] follows by defining $U=(U_T,T)$, $V=(V_T,T)$, $Y=Y_T$, and $S=S_T$;
    \item[$(e)$] follows from definition $\epsilon\triangleq\max\{\epsilon_n,\nu\}$.
\end{itemize}

Next, we lower bound $nR$ as follows.
\begin{align}
nR&\geq \ent(M)\nonumber\\
&\ge \mi(M;Z^n)\nonumber\\
&= \sum\limits_{i = 1}^n \mi(M;Z_i|Z^{i-1}) \nonumber\\
&= \sum\limits_{i = 1}^n \big[\mi(M,S_{i+1}^n;Z_i|Z^{i-1}) - \mi(S_{i+1}^n;Z_i|M,Z^{i-1})\big]  \nonumber\\
&\mathop = \limits^{(a)} \sum\limits_{i = 1}^n \big[\mi(M,S_{i+1}^n;Z_i|Z^{i-1}) - \mi(Z^{i-1};S_i|M,S_{i+1}^n)\big]  \nonumber\\
&\mathop \ge \limits^{(b)} \sum\limits_{i = 1}^n \big[\mi(M,S_{i+1}^n,Z^{i-1};Z_i) - \mi(Z^{i-1};S_i|M,S_{i+1}^n)\big]  - \delta \nonumber\\
&\mathop \ge \limits^{(c)} \sum\limits_{i = 1}^n \big[\mi(M,S_{i+1}^n,Z^{i-1};Z_i) - \mi(M,S_{i+1}^n,Z^{i-1};S_i)\big] - \delta \nonumber\\
&\mathop = \limits^{(d)} \sum\limits_{i = 1}^n \big[\mi(V_i;Z_i) - \mi(V_i;S_i)\big] - \delta \nonumber\\
&=n\sum\limits_{i = 1}^n \frac{1}{n}\big[\mi(V_i;Z_i|T=i) - \mi(V_i;S_i|T=i)\big]  - \delta\nonumber\\
&=n\sum\limits_{i = 1}^n \Prob(T=i)\big[\mi(V_i;Z_i|T=i) - \mi(V_i;S_i|T=i)\big]  - \delta\nonumber\\
&= n\big[\mi(V_T;Z_T|T) - \mi(V_T;S_T|T)\big] - \delta \nonumber\\
&\mathop = \limits^{(e)} n\big[\mi(V_T;Z_T|T) - \mi(V_T,T;S_T)\big] - \delta \nonumber\\
&\mathop \ge \limits^{(f)} n\big[\mi(V_T,T;Z_T) - \mi(V_T,T;S_T)\big] - 2\delta \nonumber\\
&\mathop = \limits^{(g)} n\big[\mi(V;Z) - \mi(V;S)\big] - 2\delta 
\end{align}
where 
\begin{itemize}
\item[$(a)$] follows from Csisz\'ar-K\"{o}rner sum identity \cite[Lemma~7]{BCC:IT78};
\item[$(b)$] follows from Lemma~\ref{lemma:Resolvability_Properties};
\item[$(c)$] follows since $S_i$ is independent of $(M,S_{i+1}^n)$;
\item[$(d)$] follows by defining $V_i\triangleq(M,S_{i+1}^n,Z^{i-1})$;
\item[$(e)$] follows from the independence of $S_T$ and $T$;
\item[$(f)$] follows from Lemma~\ref{lemma:Resolvability_Properties};
\item[$(g)$] follows by defining $V=(V_T,T)$, $Z=Z_T$, and $S=S_T$.
\end{itemize}
For any $\nu>0$, choosing $n$ large enough ensures that
\begin{align}
 R &\mathop \geq \mi(V;Z) - \mi(V;S) - 2\nu\nonumber\\
 &\mathop \geq \mi(V;Z) - \mi(V;S) - 2\epsilon,
 \label{eq:Upper_Bound_on_KRate_CSI_NC}
\end{align}where the last inequality follows from definition $\epsilon\triangleq\max\{\epsilon_n,\nu\}$. 
To show that $ \kld(P_Z||Q_0)\leq\epsilon$, note that for $n$ large enough
\begin{align}
\kld(P_Z||Q_0)=\kld(P_{Z_T}||Q_0)=\kld\Bigg(\frac{1}{n}\sum\limits_{i=1}^nP_{Z_i}\Bigg|\Bigg|Q_0\Bigg)\leq\frac{1}{n}\sum\limits_{i=1}^n\kld(P_{Z_i}||Q_0)\leq\frac{1}{n}\kld(P_{Z^n}||Q_0^{\otimes n})\leq\frac{\delta}{n}\leq\nu\leq\epsilon.\label{eq:boundKL}
\end{align}
Combining \eqref{eq:Upper_Bound_on_MRate_CSIT_NC_2}, \eqref{eq:Upper_Bound_on_MRate_CSIT_NC_2_SEC}, \eqref{eq:Upper_Bound_on_KRate_CSI_NC}, and~\eqref{eq:boundKL} shows that $\forall \epsilon_n,\nu>0$, $R\leq \max\{x:x\in\mathcal{S}_{\epsilon}\}$. Therefore,
\begin{align}
  R\leq \max\left\{x:x\in\bigcap_{\epsilon>0}\mathcal{S}_{\epsilon}\right\}.
\end{align}
\subsection{Continuity at zero}
The proof is available in Appendix~\ref{sec:continuity-at-zero}.

\section{Proof of Lemma~\ref{lemma:Resolvability_Properties}}
\label{sec:Proof_Lemma_Resolvability_Properties}
First note that,
\begin{align}
    \sum\limits_{i=1}^n\mi(Z_i;Z^{i-1})&=\sum\limits_{i=1}^n[\ent(Z_i)-\ent(Z_i|Z^{i-1})]\nonumber\\
    &=\sum\limits_{i=1}^n\ent(Z_i)-\ent(Z^n)\nonumber\\
    &=-\sum\limits_{i=1}^n\sum\limits_{z}P_{Z_i}(z)\log P_{Z_i}(z)+\sum\limits_{z^n}P_{Z^n}(z^n)\log P_{Z^n}(z^n)\nonumber\\
    &=-\sum\limits_{i=1}^n\sum\limits_{z}P_{Z_i}(z)\log P_{Z_i}(z)+\sum\limits_{i=1}^n\sum\limits_{z}P_{Z_i}(z)\log Q_0(z)\nonumber\\
    &\qquad-\sum\limits_{i=1}^n\sum\limits_{z}P_{Z_i}(z)\log Q_0(z)+\sum\limits_{z^n}P_{Z^n}(z^n)\log P_{Z^n}(z^n)\nonumber\\
    &=-\sum\limits_{i=1}^n\kld(P_{Z_i}||Q_0)-\sum\limits_{z^n}P_{Z^n}(z^n)\log Q_0^{\otimes n}(z^n)+\sum\limits_{z^n}P_{Z^n}(z^n)\log P_{Z^n}(z^n)\nonumber\\
    &\leq\kld(P_{Z^n}||Q_0^{\otimes n})\nonumber\\
    &\leq\delta.\nonumber
\end{align} Similarly one can proof $\sum\nolimits_{i=1}^n\mi(Z_i;Z_{i+1}^n)\leq\delta$. Next,
\begin{align}
    \mi(Q;Z_Q)&=\ent(Z_Q)-\ent(Z_Q|Q)\nonumber\\
    &=-\sum\limits_{z}\frac{1}{n}\sum\limits_{i=1}^nP_{Z_i}(z)\log \frac{1}{n}\sum\limits_{j=1}^nP_{Z_j}(z)+\frac{1}{n}\sum\limits_{i=1}^n\sum\limits_{z}P_{Z_i}(z)\log P_{Z_i}(z)\nonumber\\
    &=-\sum\limits_{z}\frac{1}{n}\sum\limits_{i=1}^nP_{Z_i}(z)\log \frac{1}{n}\sum\limits_{j=1}^nP_{Z_j}(z)+\sum\limits_{z}\frac{1}{n}\sum\limits_{i=1}^nP_{Z_i}(z)\log Q_0(z)\nonumber\\
    &\qquad -\sum\limits_{z}\frac{1}{n}\sum\limits_{i=1}^nP_{Z_i}(z)\log Q_0(z)+\frac{1}{n}\sum\limits_{i=1}^n\sum\limits_{z}P_{Z_i}(z)\log P_{Z_i}(z)\nonumber\\
    &=-\kld\Bigg(\frac{1}{n}\sum\limits_{i=1}^nP_{Z_i}\Bigg|\Bigg|Q_0\Bigg)+\frac{1}{n}\sum\limits_{i=1}^n\kld(P_{Z_i}||Q_0)\nonumber\\
    &\leq\frac{1}{n}\sum\limits_{i=1}^n\kld(P_{Z_i}||Q_0)\nonumber\\
    &\leq\frac{1}{n}\kld(P_{Z^n}||Q_0^{\otimes n})\nonumber\\
    &\leq\frac{\delta}{n}.
\end{align}

\section{Continuity at zero}
\label{sec:continuity-at-zero}
Our objective is to show that the capacity region is included in the region defined in \eqref{eq:finalSD_CSIT_NC}. 
The challenge, first highlighted in~\cite[Section VI.D]{Cuff13}, is that our converse arguments only establish that the capacity region is included in the region $\bigcap_{\epsilon>0}\calS_\epsilon$ where $\calS_\epsilon$ is defined in \eqref{eq:Epsilon_Rate_Region_CSIT_NC}. 
In the sequel, the continuity of the slackness in the mutual information inequality~\eqref{eq:Depsilon_CSIT_NC} will assume some importance, hence for ease of expression we define and refer to $g(\epsilon)\triangleq3\epsilon$. As $\epsilon$ vanishes, \emph{both} the region $\calS_\epsilon$ and the set of distributions $\calD_\epsilon$ shrink, so that proving the continuity at $\epsilon=0$ is not completely straightforward. We carefully lay out the arguments leading to the result in a series of lemmas.

\begin{lemma}
  \label{lm:d_epsilon}
  For all $\epsilon>0$, the set $\calD_\epsilon$ is closed and bounded, hence compact.
\end{lemma}
\begin{proof}
We need to check that every constraint defining the set $\calD_\epsilon$ defines a close set of distributions, so that $\calD_\epsilon$ is an intersection of closed sets and remains closed. First note that:
\begin{itemize}
\item the function that outputs the marginal $P_Z$ of $P_{U,V,S,X,Y,Z}$ is continuous in $P_Z$;
\item $Q_0$ has support $\calZ$ so that the divergence $\D{P_Z}{Q_0}$ is a continuous function of $P_Z$;
\item mutual information, viewed as a function of the joint distribution of the random variables involved, is continuous;
\item all the constraints in the definition of $\calD_\epsilon$ are inequalities.
\end{itemize}
Consequently, the pre-images of the closed sets defined by the inequalities are pre-images of closed sets through continuous functions, hence closed. $\calD_\epsilon$ is bounded because it is a subset of the probability simplex, hence it is compact.


\end{proof}

\begin{lemma}
  \label{lm:s_epsilon}
  For all $\epsilon>0$ the set $\calS_\epsilon$ is non-empty, closed, and bounded.
\end{lemma}
\begin{proof}
The set of Pareto optimal points in $\calS_\epsilon$ is the image of $\calD_\epsilon$ through a continuous function. Since $\calD_\epsilon$ is compact, the set of Pareto optimal points is compact. In $\bbR$, compact sets are closed, hence the set of Pareto optimal points is closed and $\calS_\epsilon$ itself is closed by definition. 
$\calS_\epsilon$ is also non empty because it contains $0$. $\calS_\epsilon$ is bounded because we can upper bound $R$ by $2\log\card{\calX}+\epsilon$.
\end{proof}

Now define the set 
\begin{align}
\calS'_\epsilon\eqdef \left.\begin{cases}R\geq 0:\\
    \exists P_{U,V,S,X,Y,Z}\in\calD_\epsilon: R\leq \min\{\mi(U;Y)-\mi(U;S),\mi(U,V;Y)-\mi(U;S|V)\}
  \end{cases}\right\}.
\end{align}
Note that $\calS'_\epsilon$ differs from $\calS_\epsilon$ in the absence of $\epsilon$ in the rate constraint.
\begin{lemma}
  \label{lm:sprime_epsilon}
  For all $\epsilon>0$ the set $\calS'_\epsilon$ is closed and bounded.
\end{lemma}
\begin{proof}
  Similar to the proof of Lemma~\ref{lm:s_epsilon}
\end{proof}

\begin{lemma}
  \label{lm:equal_intersection}
  \begin{align}
    \bigcap_{\epsilon>0}\calS'_\epsilon=\bigcap_{\epsilon>0}\calS_\epsilon
  \end{align}
\end{lemma}

\begin{proof}
  First note that $\bigcap_{\epsilon>0}\calS'_\epsilon$ is closed, since it is an intersection of closed sets, and bounded, since the sets $\calS'_\epsilon$ are nested and bounded. Hence, $\bigcap_{\epsilon>0}\calS'_\epsilon$ is compact. Consequently, there exists a maximal element, $r^\dagger$. Consider any $r\in[0,r^\dagger]$. Then $\forall \epsilon >0\exists P_{U,V,S,X,Y,Z}\in\calD_\epsilon \quad r\leq r^\dagger\leq \min\{\mi(U;Y)-\mi(U;S),\mi(U,V;Y)-\mi(U;S|V)\}$ and $r\in\bigcap_{\epsilon>0}\calS'_\epsilon$.

  We now want to show that $\bigcap_{\epsilon>0}\calS'_\epsilon=\bigcap_{\epsilon>0}\calS_\epsilon$. The hard part is showing that $\bigcap_{\epsilon>0}\calS_\epsilon\subset \bigcap_{\epsilon>0}\calS'_\epsilon$ since the other direction follows by the definition of $\calS_\epsilon$ and $\calS'_\epsilon$. We proceed by contradiction. Assume $\exists r^*\in\bigcap_{\epsilon>0}\calS_\epsilon$ such that $r^*\notin \bigcap_{\epsilon>0}\calS'_\epsilon$.  It must be that $r^\dagger<r^*$ for otherwise $r^*\in\bigcap_{\epsilon>0}\calS'_\epsilon$ as noted earlier.

Set $r_0\eqdef\frac{1}{2}(r^\dagger+r^*)$, which is such that $r_0>r^\dagger$ and therefore $r_0\notin\bigcap_{\epsilon>0}\calS'_\epsilon$. Set $\epsilon'>0$ such that $\forall \epsilon\leq\epsilon'$ $g(\epsilon)<\frac{r^*-r^\dagger}{2}$, which exists by the assumptions on $g$. Assume that $\forall \epsilon\in (0;\epsilon']$ $r_0\in\calS'_\epsilon$. Then $r_0\in \bigcap_{\epsilon>0}\calS'_\epsilon$ which contradicts our assumption. Hence, there exists $0<\epsilon_0\leq \epsilon'$ such that $r_0\notin \calS'_{\epsilon_0}$. Hence $\forall P_{U,V,S,X,Y,Z}\in\calD_{\epsilon_0}$ $r_0 > \min\{\mi(U;Y)-\mi(U;S),\mi(U,V;Y)-\mi(U;S|V)\}$. Then $\forall P_{U,V,S,X,Y,Z}\in\calD_{\epsilon_0}$
\begin{align}
  &r_0 > \min\{\mi(U;Y)-\mi(U;S),\mi(U,V;Y)-\mi(U;S|V)\} \\
  &\Rightarrow \frac{r^*+r^\dagger}{2} > \min\{\mi(U;Y)-\mi(U;S),\mi(U,V;Y)-\mi(U;S|V)\}\\
  &\Rightarrow \frac{r^*+r^\dagger}{2} + \frac{r^*-r^\dagger}{2} > \min\{\mi(U;Y)-\mi(U;S),\mi(U,V;Y)-\mi(U;S|V)\} + \frac{r^*-r^\dagger}{2}\\
  &\Rightarrow r^* > \min\{\mi(U;Y)-\mi(U;S),\mi(U,V;Y)-\mi(U;S|V)\} + g(\epsilon_0)
\end{align}
Since $r^*\in\bigcap_{\epsilon>0}\calS_\epsilon$, we have $\forall\epsilon>0$ $\exists P_{U,V,S,X,Y,Z}\in\calD_{\epsilon}$ such that $r^*\leq \min\{\mi(U;Y)-\mi(U;S),\mi(U,V;Y)-\mi(U;S|V)\}+g(\epsilon)$. Hence, there is a contradiction and we must have $r^*\in\bigcap_{\epsilon>0}\calS'_\epsilon$.  
\end{proof}

To conclude, one can prove that $\bigcap_{\epsilon>0}\calS'_\epsilon = \calS_0$, following the exact same arguments as in~\cite[Section IV.C]{Cuff13}.

\section{Proof of Theorem~\ref{thm:Acivable_Rate_For_CSIT_C}}
\label{sec:Proof_Acivable_Rate_Condition_For_CSIT_C}

We adopt a block-Markov encoding scheme in which $B$ independent messages are transmitted over $B$ channel blocks each of length $r$, such that $n=rB$.
The warden's observation is $Z^n=(Z_1^r,\dots,Z_B^r)$, the distribution induced at the output of the warden is $P_{Z^n}$, the target output distribution is $Q_Z^{\otimes n}$, and Equation \eqref{eq:Total_KLD_CoNC}, describing the distance between the two distributions, continues to  hold. 
The random code generation is as follows:

\sloppy Fix  $P_U(u)$, $P_{V|S}(v|s)$, $x(u,s)$, and $\epsilon_1>\epsilon_2>0$ such that, $P_Z = Q_0$.

\emph{Codebook Generation for Keys:}
For each block $j\in\intseq{1}{B}$, let $C_1^{(r)}\triangleq\big\{V^r(\ell_j)\big\}_{\ell_j\in\mathcal{L}}$, where $\mathcal{L}=\intseq{1}{2^{r\tilde{R}}}$, be a random codebook consisting of independent random sequences each generated according to $P_V^{\otimes r}$, where $P_V=\sum_{s\in\mathcal{S}}Q
_S(s)P_{V|S}(v|s)$. We denote a realization of $C_1^{(r)}$ by $\mathcal{C}_1^{(r)}\triangleq\big\{v^r(\ell_j)\big\}_{\ell_j\in\mathcal{L}}$. 
Partition the set of indices $\ell_j\in\intseq{1}{2^{r\tilde{R}}}$ into bins $\mathcal{B}(t)$, $t\in\intseq{1}{2^{rR_T}}$ by using function $\varphi:V^r(\ell_j)\to\intseq{1}{2^{rR_T}}$ through random binning by choosing the value of $\varphi(v^r(\ell_j))$ independently and uniformly at random for every $v^r(\ell_j)\in\mathcal{V}^r$. 
For each block $j\in\intseq{1}{B}$, create a function $\Phi:V^r(\ell_j)\to\intseq{1}{2^{rR_K}}$ through random binning by choosing the value of $\Phi(v^r(\ell_j))$ independently and uniformly at random for every $v^r(\ell_j)\in\mathcal{V}^r$. The key $k_j=\Phi(v^r(\ell_j))$ obtained in block $j\in\intseq{1}{B}$ from the description of the channel state sequence $v^r(\ell_j)$ is used to assist the encoder in block $j+2$.

\emph{Codebook Generation for Messages:}
For each block $j\in\intseq{1}{B}$, let $C_2^{(r)}\triangleq\big\{U^r(m_j,t_{j-1},k_{j-2})\big\}_{(m_j,t_{j-1},k_{j-2})\in\mathcal{M}\times\mathcal{T}\times\mathcal{K}}$, where $\mathcal{M}=\intseq{1}{2^{rR}}$, $\mathcal{T}=\intseq{1}{2^{rR_t}}$, and $\mathcal{K}=\intseq{1}{2^{rR_k}}$, be a random codebook consisting of independent random sequences each generated according to $P_U^{\otimes r}$. We denote a realization of $C_2^{(r)}$ by $\mathcal{C}_2^{(r)}\triangleq\big\{u^r(m_j,t_{j-1},k_{j-2})\big\}_{(m_j,t_{j-1},k_{j-2})\in\mathcal{M}\times\mathcal{T}\times\mathcal{K}}$. Let, $C_r=\big\{C_1^{(r)},C_2^{(r)}\big\}$ and $\mathcal{C}_r=\big\{\mathcal{C}_1^{(r)},\mathcal{C}_2^{(r)}\big\}$. The indices $(m_j,t_{j-1},k_{j-2})$ can be viewed as a two layer binning. 
%
%
We define an ideal \ac{PMF} for codebook $\mathcal{C}_n$, as an approximate distribution to facilitate the analysis
\begin{align}
    &\Gamma_{M_j,T_{j-1},K_{j-2},L_j,U^r,V^r,S_j^r,Z_j^r,K_{j-1},T_j,K_j}^{(\mathcal{C}_n)}(m_j,t_{j-1},k_{j-2},\ell_j,\tilde{u}^r,\tilde{v}^r,s_j^r,z_j^r,k_{j-1},t_j,k_j)  \nonumber\\
    &\qquad= 2^{-r(R + R_t +R_k + \tilde{R})}\indic{1}_{\{ \tilde{u}^r=u^r(m_j,t_{j-1},k_{j-2})\}}\indic{1}_{\{\tilde{v}^r=v^r(\ell_j) \}} P_{S|V}^{\otimes r}(s_j^r| \tilde{v}^r)  W_{Z|U,S}^{\otimes r}(z_j^r|\tilde{u}^r,s_j^r)\nonumber\\
    &\qquad\quad\times 2^{-rR_k}\indic{1}_{\{ t_j=\varphi (\tilde{v}^r)\}}\indic{1}_{\{ k_j=\Phi (\tilde{v}^r)\}},\label{eq:Ideal_PMF_Causal}
\end{align}where $W_{Z|U,S}$ is the marginal distribution of $W_{Z|U,S}=\sum_{x\in\mathcal{X}}\indic{1}_{\{x=x(u,s)\}}W_{Z|X,S}$ and
\begin{align}
    P_{S|V}=\frac{P_{S,V}(s,v)}{P_{V}(v)}=\frac{Q_S(s)P_{V|S}(v|s)}{\sum_{s\in\mathcal{S}}Q_S(s)P_{V|S}(v|s)}.\label{eq:Marginal_LE_C}
\end{align}

\emph{{Encoding:}}
We assume that the transmitter and the receiver have access to shared secret keys $k_{-1}$ and $k_0$ for the first two blocks but after the first two blocks they use the key that they generate from the \ac{CSI}.  

In the first block, to send the message $m_1$ according to $k_{-1}$, the encoder chooses the index $t_0$ uniformly at random (the index $t_0$ does not carry any useful information) and computes $u^r(m_1,t_0,k_{-1})$ and transmits a codeword $x^r$, where $x_i=x(u_i(m_1,t_0,k_{-1}),s_{1,i})$.


 

At the beginning of the second block, to generate a secret key shared between the transmitter and the receiver, the encoder chooses the index $\ell_1$ according to the following distribution with $j=1$,
\begin{align}
    f\big(\ell_j|s_j^r\big)=\frac{P_{S|V}^{\otimes r}\big(s_j^r|v^r(\ell_j)\big)}{\sum\limits_{\ell'\in\intseq{1}{2^{r\tilde{R}}}} P_{S|V}^{\otimes r}\big(s_j^r|v^r(\ell'_j)\big) },\label{eq:Likelihood_Encoder_Causal}
\end{align}where $P_{S|V}$ is defined in~\eqref{eq:Marginal_LE_C}. Then generates the reconciliation index $t_1=\varphi(v^n(\ell_1))$; simultaneously the transmitter generates a key $k_1=\Phi(v^r(\ell_1))$ from the description of its \ac{CSI} of the first block $v^r(\ell_1)$ to be used in the next block. To transmit the message $m_2$ and reconciliation index $t_1$ according to the key $k_0$ the encoder computes $u^r(m_2,t_1,k_0)$ and transmits a codeword $x^r$, where $x_i=x(u_i(m_2,t_1,k_0),s_{2,i})$.

In block $j\in\intseq{3}{B}$, the encoder first selects the index $\ell_{j-1}$ based $s_{j-1}^r$ by using the likelihood encoder described in \eqref{eq:Likelihood_Encoder_Causal} and then generates the reconciliation index $t_{j-1}=\varphi(v^r(\ell_{j-1}))$; simultaneously the encoder generates a key $k_{j-1}=\Phi(v^r(\ell_{j-1}))$ from the description of its \ac{CSI} of the block $j-1$, $v^r(\ell_{j-1})$, to be used in the next block. Then to send message $m_j$ and reconciliation index $t_{j-1}$ according to the generated key $k_{j-2}$ from the previous block and the \ac{CSI} of the current block $s_j^r$, the encoder computes $u^r(m_j,t_{j-1},k_{j-2})$ and transmits a codeword $x^r$, where each coordinate of the transmitted signal is a function of the current state $s_j^r$ as well as the corresponding sample of the transmitter's codeword $u_i$, i.e., $x_i=x(u_i(m_j,t_{j-1},k_{j-2}),s_{j,i})$.

Define
\begin{align}
    \twocolalign &\Upsilon_{M_j,T_{j-1},K_{j-2},U^r,S_j^r,L_j,V^r,Z_j^r,K_{j-1},T_j,K_j}^{(\mathcal{C}_r)}(m_j,t_{j-1},k_{j-2},\tilde{u}^r,s_j^r,\ell_j,\tilde{v}^r,z_j^r,k_{j-1},t_j,k_j) \twocolbreak  \nonumber\\
    &\qquad\triangleq 2^{-r(R + R_t + R_k)}\indic{1}_{\{ \tilde{u}^r=u^r(m_j,t_{j-1},k_{j-2})\}} Q_S^{\otimes r}(s_j^r) f(\ell_j|s_j^r)\indic{1}_{\{ \tilde{v}^r=v^r(\ell_j)\}}\nonumber\\
    &\qquad\quad\times W_{Z|U,S}^{\otimes r}(z_j^r|\tilde{u}^r,s_j^r) 2^{-rR_k}\indic{1}_{\{ t_j=\varphi (\tilde{v}^r)\}}\indic{1}_{\{ k_j=\Phi (\tilde{v}^r)\}}.\label{eq:P_Dist_C_CSIT}
\end{align}For a fixed codebook $\mathcal{C}_r$, the induced joint distribution over the codebook (i.e. $P^{(\mathcal{C}_r)}$) satisfies
\begin{align}
\kld\Big(P_{M_j,T_{j-1},K_{j-2},U^r,S_j^r,L_j,V^r,Z_j^r,K_{j-1},T_j,K_j}^{(\mathcal{C}_r)}||\Upsilon_{M_j,T_{j-1},K_{j-2},U^r,S_j^r,L_j,V^r,Z_j^r,K_{j-1},T_j,K_j}^{(\mathcal{C}_r)}\Big)\leq\epsilon.\label{eq:Uniformity_Key_C_CSIT}
\end{align} 
This intermediate distribution $\Upsilon^{(\mathcal{C}_r)}$ approximates the true distribution $P^{(\mathcal{C}_r)}$ and will be used in the sequel for bounding purposes. Expression~\eqref{eq:Uniformity_Key_C_CSIT} holds because the main difference in $\Upsilon^{(\mathcal{C}_r)}$ is assuming the keys $K_{j-2}$, $K_{j-1}$ and the reconciliation index $T_{j-1}$ are uniformly distributed, which is made (arbitrarily) nearly uniform in $P^{(\mathcal{C}_r)}$ with appropriate control of rate.

\begin{figure*}
\centering
\includegraphics[width=6.0in]{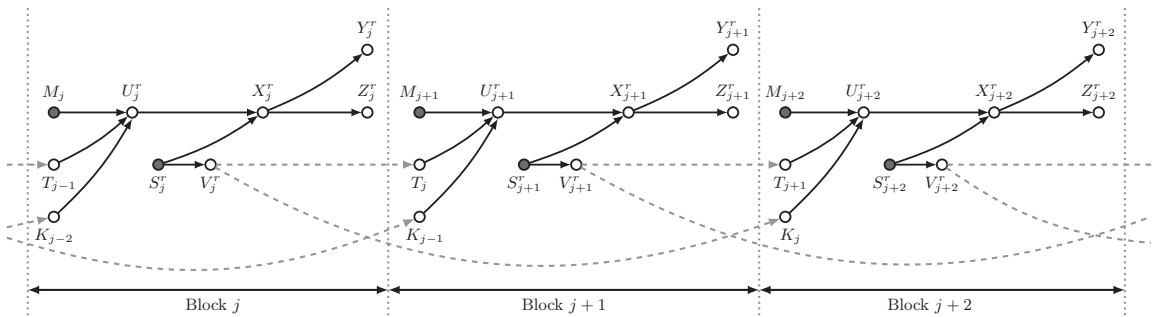}
\caption{Functional dependence graph for the block-Markov encoding scheme}
\label{fig:Chaining_C_CSIT}

\end{figure*}
\emph{Covert Analysis:}
We now show that this coding scheme guarantees that $\expec_{C_n}[\kld(P_{Z^n|C_n} || Q_Z^{\otimes n} )]\underset{n\rightarrow\infty}{\xrightarrow{\hspace{0.2in}}} 0$, where $C_n$ is the set of all the codebooks for all blocks, and then choose $P_U$, $P_V$, and $x(u,s)$ such that it satisfies $Q_Z=Q_0$. 
Similar to \eqref{eq:Total_KLD_Bound_NC_CSIT} using the functional dependence graph depicted in Fig.~\ref{fig:Chaining_C_CSIT} it follows that,
\begin{align}
    \kld(P_{Z^n}||Q_Z^{\otimes n})\leq 2\sum\limits_{j = 1}^B {\kld(P_{Z_j^r,K_{j-1},T_j,K_j}||Q_Z^{\otimes r}Q_{K_{j-1}}Q_{T_j}Q _{K_j})}.\label{eq:Total_KLD_Bound_C_CSIT}
\end{align}
To bound the \ac{RHS} of \eqref{eq:Total_KLD_Bound_C_CSIT} by using Lemma~\ref{lemma:TV_KLD} and the triangle inequality we have,
\begin{align}
    &\expec_{C_r}||P_{Z_j^r,K_{j-1},T_j,K_j|C_r}-Q_Z^{\otimes r}Q_{K_{j-1}}Q_{T_j}Q _{K_j}||_1\nonumber\\ &\leq\expec_{C_r}||P_{Z_j^r,K_{j-1},T_j,K_j|C_r}-\Gamma_{Z_j^r,K_{j-1},T_j,K_j|C_r}||_1  \twocolbreak+ \expec_{C_r}||\Gamma_{Z_j^r,K_{j-1},T_j,K_j|C_r}-Q_Z^{\otimes r}Q_{K_{j-1}}Q_{T_j}Q _{K_j}||_1\nonumber\\
    &\leq\expec_{C_r}||P_{Z_j^r,K_{j-1},T_j,K_j|C_r}-\Upsilon_{Z_j^r,K_{j-1},T_j,K_j|C_r}||_1 + \expec_{C_r}||\Upsilon_{Z_j^r,K_{j-1},T_j,K_j|C_r}-\Gamma_{Z_j^r,K_{j-1},T_j,K_j|C_r}||_1  \nonumber\\
    &\quad+ \expec_{C_r}||\Gamma_{Z_j^r,K_{j-1},T_j,K_j|C_r}-Q_Z^{\otimes r}Q_{K_{j-1}}Q_{T_j}Q _{K_j}||_1.\label{eq:Triangle_Inequality_Causaul_CSIT}
\end{align}
From \eqref{eq:Uniformity_Key_C_CSIT} and the monotonicity of KL-divergence the first term on the \ac{RHS} of \eqref{eq:Triangle_Inequality_Causaul_CSIT} goes to zero when $r$ grows. To bound the second term on the \ac{RHS} of \eqref{eq:Triangle_Inequality_Causaul_CSIT} for a fixed codebook $\mathcal{C}_r$ we have,
\begin{subequations}\label{eq:Gamma_components_Causal_CSIT}
\begin{align}
  \twocolalign\Gamma_{M_j,T_{j-1},K_{j-2}}^{(\mathcal{C}_r)} \onecolalign= 2^{-r(R+R_t+R_k)} = \Upsilon_{M_j,T_{j-1},K_{j-2}}^{(\mathcal{C}_r)},\label{eq:Uniformity_M_2_C}\\
  \twocolalign\Gamma_{U^r|M_j,T_{j-1},K_{j-2},S_j^r}^{(\mathcal{C}_r)} \onecolalign= \indic{1}_{\{ \tilde{u}^r=u^r(m_j,t_{j-1},k_{j-2})\}} \twocolbreak= \Upsilon_{U^r|M_j,T_{j-1},K_{j-2},S_j^r}^{(\mathcal{C}_r)},\label{eq:eq4_C}\\
  \twocolalign\Gamma_{L_j|M_j,T_{j-1},K_{j-2},S_j^r,U^r}^{(\mathcal{C}_r)} \onecolalign= f(\ell_j|s_j^r) = \Upsilon_{L_j|M_j,T_{j-1},K_{j-2},S_j^r,U^r}^{(\mathcal{C}_r)},\label{eq:Likelihood_Encoder_Equality_2_C}\\
  \twocolalign\Gamma_{V^r|M_j,T_{j-1},K_{j-2},S_j^r,L_j,U^r}^{(\mathcal{C}_r)} \onecolalign= \indic{1}_{\{ \tilde{v}^r=v^r(\ell_j)\}} \twocolbreak= \Upsilon_{V^r|M_j,T_{j-1},K_{j-2},S_j^r,L_j,U^r}^{(\mathcal{C}_r)},\label{eq:eq5_C}\\
    \twocolalign\Gamma_{Z_j^r|M_j,T_{j-1},K_{j-2},S_j^r,L_j,U^r,V^r}^{(\mathcal{C}_r)} \onecolalign= W_{Z|U,S}^{\otimes r}  = \Upsilon_{Z_j^r|M_j,T_{j-1},K_{j-2},S_j^r,L_j,U^r,V^r}^{(\mathcal{C}_r)},\label{eq:Channel_Equality_2_C}\\
    \twocolalign\Gamma_{K_{j-1}|M_j,T_{j-1},K_{j-2},S_j^r,L_j,U^r,V^r,Z_j^r}^{(\mathcal{C}_r)} \onecolalign= 2^{-rR_k}  = \Upsilon_{K_{j-1}|M_j,T_{j-1},K_{j-2},S_j^r,L_j,U^r,V^r,Z_j^r}^{(\mathcal{C}_r)},\label{eq:Second_Key_C}\\
    \twocolalign\Gamma_{T_j|M_j,T_{j-1},K_{j-2},S_j^r,L_j,U^r,V^r,Z_j^r}^{(\mathcal{C}_r)} \onecolalign=\indic{1}_{\{ t_j=\sigma (v^r(\ell_j)) \}}=\Upsilon_{T_j|M_j,T_{j-1},K_{j-2},S_j^r,L_j,U^r,V^r,Z_j^r}^{(\mathcal{C}_r)},\label{eq:K_j_Distribution_C}\\
    \twocolalign\Gamma_{K_j|M_j,T_{j-1},K_{j-2},S_j^r,L_j,U^r,V^r,Z_j^r,T_j}^{(\mathcal{C}_r)} \onecolalign=\indic{1}_{\{ k_j=\Phi (v^r(\ell_j))\}}=\Upsilon_{K_j|M_j,T_{j-1},K_{j-2},S_j^r,L_j,U^r,V^r,Z_j^r,T_j}^{(\mathcal{C}_r)},\label{eq:T_j_Distribution_C}
\end{align}
\end{subequations}
where \eqref{eq:Likelihood_Encoder_Equality_2_C} follows from \eqref{eq:Likelihood_Encoder_Causal}. Hence,
\begin{align}
 &\expec_{C_r}||\Upsilon_{Z_j^r,K_{j-1},T_j,K_j|C_r}-\Gamma_{Z_j^r,K_{j-1},T_j,K_j|C_r}||_1
 \twocolbreak \nonumber\\ &\leq\expec_{C_r}||\Upsilon_{M_j,T_{j-1},K_{j-2},S_j^r,L_j,U^r,V^r,Z_j^r,K_{j-1},T_j,K_j|C_r}- \Gamma_{M_j,T_{j-1},K_{j-2},S_j^r,L_j,U^r,V^r,Z_j^r,K_{j-1},T_j,K_j|C_r}||_1 \nonumber\\
&\mathop = \limits^{(a)}\expec_{C_r}||\Upsilon_{M_j,T_{j-1},K_{j-2},S_j^r|C_r}- \Gamma_{M_j,T_{j-1},K_{j-2},S_j^r|C_r}||_1 \nonumber\\
&\mathop = \limits^{(b)} \expec_{C_r}||Q_S^{\otimes r} - \Gamma_{S_j^r|M_j=1,T_{j-1}=1,K_{j-2}=1,C_r}||_1,\label{eq:General_Distribution_Expansion_2_C}
\end{align}where $(a)$ follows from \eqref{eq:eq4_C}-\eqref{eq:T_j_Distribution_C} and $(b)$ follows from the symmetry of the codebook construction with respect to $M_j$, $T_{j-1}$, and $K_{j-2}$ and \eqref{eq:Uniformity_M_2_C}. Based on the soft covering lemma \cite[Corollary~VII.5]{Cuff13} the \ac{RHS} of \eqref{eq:General_Distribution_Expansion_2_C} vanishes when $r$ grows if
\begin{align}\label{eq:State_Dependent_Covering_Lemma_2_C}
    \tilde{R} > \mi(S;V).
\end{align}
We now proceed to bound the third term on the \ac{RHS} of \eqref{eq:Triangle_Inequality_Causaul_CSIT}. First, consider the marginal,
\begin{align}
    &\Gamma_{Z_j^r,K_{j-1},T_j,K_j|C_r} (z_j^r,k_{j-1},t_j,k_j)\nonumber\\
    &= \sum\limits_{m_j} \sum\limits_{t_{j-1}}\sum\limits_{k_{j-2}}\sum\limits_{\ell_j} \sum\limits_{s_j^r} 
    \frac{1}{2^{r(R + R_t + 2R_k+\tilde{R})}} P_{S|V}^{\otimes r}\big(s_j^r|V^r(\ell_j)\big)\nonumber\\
    &\quad\times W_{Z|U,S}^{\otimes r}\big(z_j^r|U^r(m_j,t_{j-1},k_{j-2}),s_j^r\big)
     \indic{1}_{\{ t_j=\varphi(V^r(\ell_j))\}} \indic{1}_{\{ k_j=\Phi (V^r(\ell_j))\}}\nonumber\\
    &= \sum\limits_{m_j} \sum\limits_{t_{j-1}}\sum\limits_{k_{j-2}}\sum\limits_{\ell_j} 
    \frac{1}{2^{r(R + R_t + 2R_k+\tilde{R})}}  W_{Z|U,V}^{\otimes r}\big(z_j^r|U^r(m_j,t_{j-1},k_{j-2}),V^r(\ell_j)\big)
     \indic{1}_{\{ t_j=\varphi(V^r(\ell_j))\}} \indic{1}_{\{ k_j=\Phi (V^r(\ell_j))\}},
    \label{eq:Gamma_ZK_CSIT_C}  
\end{align}where $W_{Z|U,V}(z|u,v)=\sum_{s\in\mathcal{S}} P_{S|V}(s|v)W_{Z|U,S}(z|u,s)$.
\sloppy To bound the third term on the \ac{RHS} of \eqref{eq:Triangle_Inequality_Causaul_CSIT} by using Pinsker's inequality, it is sufficient to bound $\expec_{C_r}[\kld({\Gamma_{Z_j^r,K_{j-1},T_j,K_j|C_r}||Q_Z^{\otimes r}Q_{K_{j-1}}Q_{T_j}Q_{K_j}})]$ as follows,
\begin{align}
&\expec_{C_r}\big[\kld({\Gamma_{Z_j^r,K_{j-1},T_j,K_j|C_r}||Q_Z^{\otimes r} Q_{K_{j-1}}Q_{T_j}Q_{K_j}})\big] \nonumber\\ 
&= \expec_{C_r}\Bigg[\sum\limits_{z_j^r,k_{j-1},t_j,k_j} \Gamma_{Z_j^r,K_{j-1},T_j,K_j|C_r}(z_j^r,k_{j-1},t_j,k_j)\log \bigg(\frac{\Gamma_{Z_j^r,K_{j-1},T_j,K_j|C_r}(z_j^r,k_{j-1},t_j,k_j)}{Q_Z^{\otimes r}(z_j^r) Q_{K_{j-1}}(k_{j-1}) Q_{T_j}(t_j) Q_{K_j}(k_j)}\bigg) \Bigg] \nonumber\\ 
&= \expec_{C_r}\Bigg[\sum\limits_{(z_j^r,k_{j-1},t_j,k_j)}\sum\limits_{m_j} \sum\limits_{t_{j-1}}\sum\limits_{k_{j-2}}  \sum\limits_{\ell_j}\frac{W_{Z|U,V}^{\otimes r}\big(z_j^r|U^r(m_j,t_{j-1},k_{j-2}),V^r(\ell_j)\big)\indic{1}_{\{ t_j=\varphi(V^r(\ell_j))\}} \indic{1}_{\{ k_j=\Phi (V^r(\ell_j))\}}}{2^{r(R + R_t + 2R_k+\tilde{R})}} 
 \nonumber\\
&\times\log\Bigg(\frac{\sum\limits_{\tilde{m}_j}\sum\limits_{\tilde{t}_{j-1}}\sum\limits_{\tilde{k}_{j-2}} \sum\limits_{\tilde{\ell}_{j}} W_{Z|U,V}^{\otimes r}\big(z_j^r|U^r(\tilde{m}_j,\tilde{t}_{j-1},\tilde{k}_{j-2}),V^r(\tilde{\ell}_j)\big)\indic{1}_{\{ t_j=\varphi(V^r(\tilde{\ell}_{j}))\}} \indic{1}_{\{ k_j=\Phi (V^r(\tilde{\ell}_{j}))\}}}{2^{r(R + R_t + R_k +\tilde{R} - R_T - R_K)} Q_Z^{\otimes r}(z_j^r)}\Bigg)     \Bigg]\nonumber\\ 
&\mathop\le\limits^{(a)}\sum\limits_{(z_j^r,k_{j-1},t_j,k_j)} \sum\limits_{m_j} \sum\limits_{t_{j-1}}\sum\limits_{k_{j-2}}\sum\limits_{\ell_j} \frac{1}{2^{r(R + R_t + 2R_k+\tilde{R})}}\sum\limits_{\big(u^r(m_j,t_{j-1},k_{j-2}),v^r(\ell_j)\big)} \hspace{-2mm}\Gamma^{(\mathcal{C}_r)}_{U^r,V^r,Z_j^r}\big(u^r(m_j,t_{j-1},k_{j-2}),v^r(\ell_j),z_j^r\big)\nonumber\\
&\times\expec_{\varphi(v^r(\ell_j))}\big[\indic{1}_{\{ t_j=\varphi(v^r(\ell_j))\}}\big]\times\expec_{\Phi(v^r(\ell_j))}\big[\indic{1}_{\{ k_j=\Phi (v^r(\ell_j))\}}\big] \nonumber\\
&\times\log \expec_{\mathop {\backslash (m_j,t_{j-1},k_{j-2},\ell_j),}\limits_{\backslash (\varphi (v^r(\ell_j)),\Phi (v^r(\ell_j)))} } \Bigg[ \frac{\sum\limits_{\tilde{m}_j}\sum\limits_{\tilde{t}_{j-1}}\sum\limits_{\tilde{k}_{j-2}}  \sum\limits_{\tilde{\ell}_{j}} W_{Z|U,V}^{\otimes r}\big(z_j^r|U^r(\tilde{m}_j,\tilde{t}_{j-1},\tilde{k}_{j-2}),V^r(\tilde{\ell}_j)\big)\indic{1}_{\{ t_j=\varphi(V^r(\tilde{\ell}_j))\}} \indic{1}_{\{ k_j=\Phi (V^r(\tilde{\ell}_j))\}}}{2^{r(R + R_t + R_k +\tilde{R}  - R_T - R_K)} Q_Z^{\otimes r}(z_j^r)}  \Bigg]      \nonumber\\
&\mathop\le\limits^{(b)} \sum\limits_{(z_j^r,k_{j-1},t_j,k_j)} \sum\limits_{m_j} \sum\limits_{t_{j-1}}\sum\limits_{k_{j-2}}\sum\limits_{\ell_j} \frac{1}{2^{r(R + R_t + 2R_k+\tilde{R}+R_T+R_K)}}\sum\limits_{\big(u^r(m_j,t_{j-1},k_{j-2}),v^r(\ell_j)\big)} \hspace{-2mm}\Gamma^{(\mathcal{C}_r)}_{U^r,V^r,Z_j^r}\big(u^r(m_j,t_{j-1},k_{j-2}),v^r(\ell_j),z_j^r\big)\nonumber\\
&\times \log\frac{1}{2^{r(R + R_t + R_k +\tilde{R}  - R_T - R_K)} Q_Z^{\otimes r}(z_j^r)}\times\Bigg(W_{Z|U,V}^{\otimes r}\big(z_j^r|u^r(m_j,t_{j-1},k_{j-2}),v^r(\ell_j)\big) \nonumber\\
&+  \expec_{\backslash (m_j,t_{j-1},k_{j-2})}\Bigg[\sum\limits_{(\tilde{m}_j,\tilde{t}_{j-1},\tilde{k}_{j-2}) \ne (m_j,t_{j-1},k_{j-2})} W_{Z|U,V}^{\otimes r}\big(z_j^r|U^r(\tilde{m}_j,\tilde{t}_{j-1},\tilde{k}_{j-2}),v^r(\ell_j)\big)\Bigg]\nonumber\\
& +  \expec_{\mathop {\backslash \ell_j,}\limits_{\backslash (\varphi (v^r(\ell_j)),\Phi (v^r(\ell_j)))} }\Bigg[\sum\limits_{\tilde{\ell}_j \ne \ell_j} W_{Z|U,V}^{\otimes r}\big(z_j^r|u^r(m_j,t_{j-1},k_{j-2}),V^r(\tilde{\ell}_j)\big)\indic{1}_{\{ t_j=\varphi (V^r(\tilde{\ell}_j))\}}\indic{1}_{\{ k_j=\Phi (V^r(\tilde{\ell}_j))\}}\Bigg]\nonumber\\
&+  \expec_{\mathop {\backslash (m_j,t_{j-1},k_{j-2},\ell_j),}\limits_{\backslash (\varphi (v^r(\ell_j)),\Phi (v^r(\ell_j)))} }\Bigg[\sum\limits_{\tilde{\ell}_{j} \ne \ell_j}\sum\limits_{(\tilde{m}_j,\tilde{t}_{j-1},\tilde{k}_{j-2}) \ne (m_j,t_{j-1},k_{j-2})}\hspace{-1.2cm} W_{Z|U,V}^{\otimes r}\big(z_j^r|U^r(\tilde{m}_j,\tilde{t}_{j-1},\tilde{k}_{j-2}),V^r(\tilde{\ell}_j)\big)\indic{1}_{\{ t_j=\varphi (V^r(\tilde{\ell}_j))\}}\indic{1}_{\{ k_j=\Phi (V^r(\tilde{\ell}_j))\}}\Bigg]\Bigg)    \nonumber\\
&\mathop\le\limits^{(c)} \sum\limits_{(z_j^r,k_{j-1},t_j,k_j)} \sum\limits_{m_j} \sum\limits_{t_{j-1}}\sum\limits_{k_{j-2}}\sum\limits_{\ell_j} \frac{1}{2^{r(R + R_t + 2R_k+\tilde{R}+R_T+R_K)}}\sum\limits_{\big(u^r(m_j,t_{j-1},k_{j-2}),v^r(\ell_j)\big)} \hspace{-2mm}\Gamma^{(\mathcal{C}_r)}_{U^r,V^r,Z_j^r}\big(u^r(m_j,t_{j-1},k_{j-2}),v^r(\ell_j),z_j^r\big)\nonumber\\
&\times\log \Bigg(\frac{W_{Z|U,V}^{\otimes r}\big(z_j^r|u^r(m_j,t_{j-1},k_{j-2}),v^r(\ell_j)\big)}{2^{r(R + R_t + R_k +\tilde{R}  - R_T - R_K)} Q_Z^{\otimes r}(z_j^r)} +\sum\limits_{(\tilde{m}_j,\tilde{t}_{j-1},\tilde{k}_{j-2}) \ne (m_j,t_{j-1},k_{j-2})}\frac{ W_{Z|V}^{\otimes r}\big(z_j^r|v^r(\ell_j)\big)}{{{2^{r(R + R_t + R_k +\tilde{R}  - R_T - R_K)}} Q_Z^{\otimes r}(z_j^r)}}\nonumber\\
& +  \sum\limits_{\tilde{\ell}_{j} \ne \ell_j}\frac{ W_{Z|U}^{\otimes r}\big(z_j^r|u^r(m_j,t_{j-1},k_{j-2})\big)}{{{2^{r(R + R_t + R_k +\tilde{R})}} Q_Z^{\otimes r}(z_j^r)}} +1\Bigg)\nonumber\\
&\le \sum\limits_{(z_j^r,k_{j-1},t_j,k_j)} \sum\limits_{m_j} \sum\limits_{t_{j-1}}\sum\limits_{k_{j-2}}\sum\limits_{\ell_j} \frac{1}{2^{r(R + R_t + 2R_k+\tilde{R}+R_T+R_K)}}\sum\limits_{\big(u^r(m_j,t_{j-1},k_{j-2}),v^r(\ell_j)\big)} \hspace{-2mm}\Gamma^{(\mathcal{C}_r)}_{U^r,V^r,Z_j^r}\big(u^r(m_j,t_{j-1},k_{j-2}),v^r(\ell_j),z_j^r\big)\nonumber\\
&\times\log \Bigg(\frac{  W_{Z|U,V}^{\otimes r}\big(z_j^r|u^r(m_j,t_{j-1},k_{j-2}),v^r(\ell_j)\big)}{2^{r(R + R_t + R_k +\tilde{R}  - R_T - R_K)} Q_Z^{\otimes r}(z_j^r)} +\frac{ W_{Z|V}^{\otimes r}\big(z_j^r|v^r(\ell_j)\big)}{{{2^{r(\tilde{R}- R_T - R_K)}} Q_Z^{\otimes r}(z_j^r)}} +  \frac{  W_{Z|U}^{\otimes r}\big(z_j^r|u^r(m_j,t_{j-1},k_{j-2})\big)}{{{2^{r(R + R_t + R_k)}} Q_Z^{\otimes r}(z_j^r)}}  +1\Bigg) \nonumber\\
&\triangleq \Psi_1 + \Psi_2,\label{eq:KLD_ZK_C_CSIT}
\end{align}
where $(a)$ follows from Jensen's inequality, $(b)$ and $(c)$ hold because $\indic{1}_{\{ k_j=\Phi (\tilde s_j^r)\}}\leq 1$. We defined $\Psi_1$ and $\Psi_2$ as
\begin{align}
&{\Psi _1} =  \sum\limits_{(k_{j-1},t_j,k_j)} \sum\limits_{m_j} \sum\limits_{t_{j-1}}\sum\limits_{k_{j-2}}\sum\limits_{\ell_j} \frac{1}{2^{r(R + R_t + 2R_k+\tilde{R}+R_T+R_K)}}\sum\limits_{\big(u^r(m_j,t_{j-1},k_{j-2}),v^r(\ell_j),z_j^r\big)\in\mathcal{T}_{\epsilon}^{(r)}} \hspace{-1.2cm}\Gamma^{(\mathcal{C}_r)}_{U^r,V^r,Z_j^r}\big(u^r(m_j,t_{j-1},k_{j-2}),v^r(\ell_j),z_j^r\big)\nonumber\\
&\times\log \Bigg(\frac{  W_{Z|U,V}^{\otimes r}\big(z_j^r|u^r(m_j,t_{j-1},k_{j-2}),v^r(\ell_j)\big)}{2^{r(R + R_t + R_k +\tilde{R}  - R_T - R_K)} Q_Z^{\otimes r}(z_j^r)} +\frac{ W_{Z|V}^{\otimes r}\big(z_j^r|v^r(\ell_j)\big)}{{{2^{r(\tilde{R}- R_T - R_K)}} Q_Z^{\otimes r}(z_j^r)}} +  \frac{  W_{Z|U}^{\otimes r}\big(z_j^r|u^r(m_j,t_{j-1},k_{j-2})\big)}{{{2^{r(R + R_t + R_k)}} Q_Z^{\otimes r}(z_j^r)}}  +1\Bigg) \nonumber\\
& \le \log \Bigg(\frac{2^{r(R_T+R_K)}\times 2^{ - r(1 - \epsilon )\ent(Z|U,V)}}{2^{r(R + R_t + R_k+\tilde{R})}\times 2^{ - r(1 + \epsilon )\ent(Z)}} + \frac{2^{r(R_T+R_K)}\times 2^{ - r(1 - \epsilon )\ent(Z|V)}}{2^{r\tilde{R}}\times 2^{ - r(1 + \epsilon )\ent(Z)}} +\frac{2^{ - r(1 - \epsilon )\ent(Z|U)}}{2^{r(R + R_t + R_k)}\times 2^{ - r(1 + \epsilon )\ent(Z)}}+ 1\Bigg)\label{eq:Psi_1_C_CSIT}\\
&{\Psi _2} =  \sum\limits_{(k_{j-1},t_j,k_j)} \sum\limits_{m_j} \sum\limits_{t_{j-1}}\sum\limits_{k_{j-2}}\sum\limits_{\ell_j} \frac{1}{2^{r(R + R_t + 2R_k+\tilde{R}+R_T+R_K)}}\sum\limits_{\big(u^r(m_j,t_{j-1},k_{j-2}),v^r(\ell_j),z_j^r\big)\notin\mathcal{T}_{\epsilon}^{(r)}} \hspace{-1.2cm}\Gamma^{(\mathcal{C}_r)}_{U^r,V^r,Z_j^r}\big(u^r(m_j,t_{j-1},k_{j-2}),v^r(\ell_j),z_j^r\big)\nonumber\\
&\times\log \Bigg(\frac{  W_{Z|U,V}^{\otimes r}\big(z_j^r|u^r(m_j,t_{j-1},k_{j-2}),v^r(\ell_j)\big)}{2^{r(R + R_t + R_k +\tilde{R}  - R_T - R_K)} Q_Z^{\otimes r}(z_j^r)} +\frac{ W_{Z|V}^{\otimes r}\big(z_j^r|v^r(\ell_j)\big)}{{{2^{r(\tilde{R}- R_T - R_K)}} Q_Z^{\otimes r}(z_j^r)}} +  \frac{  W_{Z|U}^{\otimes r}\big(z_j^r|u^r(m_j,t_{j-1},k_{j-2})\big)}{{{2^{r(R + R_t + R_k)}} Q_Z^{\otimes r}(z_j^r)}}  +1\Bigg) \nonumber\\
&\le 2|V||U||Z|e^{ - r\epsilon^2 {\mu _{V,U,Z}}}r\log\Big(\frac{3}{\mu_Z} + 1\Big). \label{eq:Psi_2_C_CSIT}
\end{align}
In \eqref{eq:Psi_2_C_CSIT} $\mu_{V,U,Z}=\min\limits_{(v,u,z)\in(\mathcal{V},\mathcal{U},\mathcal{Z})}P_{V,U,Z}(v,u,z)$ and $\mu_Z=\min\limits_{z\in\mathcal{Z}}P_Z(z)$. When $r\to\infty$ then $\Psi_2\to 0$ and $\Psi_1$ goes to zero when $r$ grows if
\begin{subequations}\label{eq:Covert_Analysis_Causal_CSIT}
\begin{align}
    R + R_t + R_k + \tilde{R} - R_T - R_K &>\mi(U,V;Z),\label{eq:Covert_Analysis_Causal_CSIT_1}\\
    \tilde{R} - R_T - R_K &> \mi(V;Z),\label{eq:Covert_Analysis_Causal_CSIT_2}\\
    R + R_t + R_k  &>\mi(U;Z).\label{eq:Covert_Analysis_Causal_CSIT_4}
\end{align}
\end{subequations}

\emph{{Decoding and Error Probability Analysis:}}
At the end of block $j\in\intseq{1}{B}$, using its knowledge of the key $k_{j-2}$ generated from the block $j-2$, the receiver finds a unique pair $(\hat{m}_j,\hat{t}_{j-1})$ such that $\big(u^r(\hat{m}_j,\hat{t}_{j-1},k_{j-2}),y_j^r\big)\in\mathcal{T}_{\epsilon}^{(r)}$. According to the law of large numbers and the packing lemma probability of error vanishes when $r$ grows if \cite{ElGamalKim},
\begin{align}
&R + R_t < \mi(U;Y).
\label{eq:Decoding_BM_C_CSIT}
\end{align} 
We now analyze the probability of error at the encoder and the decoder for key generation. Let $(L_{j-1},T_{j-1})$ denote the chosen indices at the encoder and $\hat{L}_{j-1}$ and $\hat{T}_{j-1}$ be the estimate of the index $L_{j-1}$ and $T_{j-1}$ at the decoder. At the end of block $j$, by decoding $U_j^r$, the receiver has access to $T_{j-1}$ and to find $L_{j-1}$ we define the error event,
\begin{align}
    \mathcal{E} = \Big\{ \big(V_{j-1}^r(\hat{L}_{j-1}),S_{j-1}^r,U_{j-1}^r,Y_{j-1}^r\big) \notin \mathcal{T}_\epsilon ^{(n)}\Big\},
\end{align}and consider the events,
\begin{subequations}\label{eq:Error_Wyner_Ziv}
\begin{align}
    \mathcal{E}_1 &= \Big\{ \big(V_{j-1}^r(\ell_{j-1}),S_{j-1}^r\big) \notin \mathcal{T}_{\epsilon'}^{(n)} \,\,\mbox{for all}\,\, \ell_{j-1}\in\intseq{1}{2^{r\tilde{R}}}\Big\},\\
    \mathcal{E}_2 &= \Big\{ \big(V_{j-1}^r(L_{j-1}),S_{j-1}^r,U_{j-1}^r,Y_{j-1}^r\big) \notin \mathcal{T}_\epsilon ^{(n)}\Big\},\\
    \mathcal{E}_3 &= \Big\{ \big(V_{j-1}^r(\tilde{\ell}_{j-1}),U_{j-1}^r,Y_{j-1}^r\big) \in \mathcal{T}_{\epsilon}^{(n)} \,\,\mbox{for some}\,\, \ell_{j-1}\in\mathcal{B}(T_{j-1}),\tilde{\ell}_{j-1}\ne \ell_{j-1}\Big\},
\end{align}
\end{subequations}where $\epsilon>\epsilon'>0$. By the union bound we have,
\begin{align}
    P(\mathcal{E})\leq P(\mathcal{E}_1) + P(\mathcal{E}_1^c\cap\mathcal{E}_2) + P(\mathcal{E}_3)\label{eq:Error_Event_Causal_CSIT}.
\end{align}
Similar to the proof of Lemma~\ref{lemma:Typicality} one can show that the first term on the \ac{RHS} of \eqref{eq:Error_Event_Causal_CSIT} vanishes when $r$ grows if we have \eqref{eq:State_Dependent_Covering_Lemma_2_C}. Following the steps in \cite[Sec.~11.3.1]{ElGamalKim}, the last two terms on the \ac{RHS} of \eqref{eq:Error_Event_Causal_CSIT} go to zero when $r$ grows if,
\begin{subequations}\label{eq:Error_Analysis_Encoder_CSIT}
\begin{align}
    \tilde{R}&>\mi(S;V),\label{eq:Error_Analysis_Encoder_CSIT_1}\\
    \tilde{R}-R_t&<\mi(V;U,Y).\label{eq:Error_Analysis_Encoder_CSIT_2}
\end{align}
\end{subequations}

Applying Fourier-Motzkin to \eqref{eq:State_Dependent_Covering_Lemma_2_C}, \eqref{eq:Covert_Analysis_Causal_CSIT}, 
\eqref{eq:Decoding_BM_C_CSIT}, and  \eqref{eq:Error_Analysis_Encoder_CSIT} and remarking that the scheme requires $R_t + R_k\leq R_T + R_K$ results in the achievable region in Theorem~\ref{thm:Acivable_Rate_For_CSIT_C}.

\section{Proof of Theorem~\ref{thm:Acivable_Rate_For_CSIT_C_Checking}}
\label{sec:Proof_Simple_Inner_Bound_For_CSIT_C}
\sloppy Fix $P_U(u)$, $x(u,s)$, and $\epsilon_1>\epsilon_2>0$ such that, $P_Z = Q_0$.

\emph{Codebook Generation:}
Let $C_n\triangleq\{U^n(m)\}_{m\in\mathcal{M}}$, where $\mathcal{M}\in\intseq{1}{2^{nR}}$, be a random codebook consisting of independent random sequences each generated according to $\prod\nolimits_{i = 1}^n P_U(u_i)$. We denote a realization of $C_n$ by $\CodeBook_n\triangleq\{u^n(m)\}_{m\in\mathcal{M}}$. 

\emph{{Encoding:}}
To send the message $m$ the encoder computes $u^n(m)$ and transmits codeword $x^n$, where $x_i=x(u_i(m),s_i)$. 
For a fixed codebook $\mathcal{C}_n$, the induced joint distribution over the codebook is as follows
\begin{align}
    \twocolalign &P_{M,S^n,U^n,Z^n}^{(\mathcal{C}_n)}(m,s^n,\tilde{u}^n,z^n)= 2^{-nR} Q_S^{\otimes n}(s^n) \indic{1}_{\{ \tilde{u}^n=u^n(m)\}} W_{Z|U,S}^{\otimes n}(z^n|\tilde{u}^n,s^n).\label{eq:P_Dist_Simple_C}
\end{align}

\emph{{Covert Analysis:}} 
We now show that this coding scheme guarantees that
\begin{align}
    \expec_{C_n}\big[\kld(P_{Z^n|C_n} || Q_Z^{\otimes n} )\big]\underset{n\rightarrow\infty}{\xrightarrow{\hspace{0.2in}}} 0,\label{eq:QZ_CSIT_C}
\end{align}
where $P_{Z^n|C_n}$ is the marginal distribution of the distribution defined in \eqref{eq:P_Dist_Simple_C} and is as follows
\begin{align}
    \twocolalign P_{Z^n|C_n}(z^n)&=\sum\limits_{m}\sum\limits_{s^n} 2^{-nR} Q_S^{\otimes n}(s^n) W_{Z|U,S}^{\otimes n}\big(z^n|u^n(m),s^n\big)\\
    &=\sum\limits_{m}2^{-nR} W_{Z|U}^{\otimes n}\big(z^n|u^n(m)\big),\label{eq:P_Marginal_Dist_Simple_C}
\end{align}where $W_{Z|U}=\sum_{s\in\mathcal{S}} Q_S(s)W_{Z|U,S}(z|u,s)$. Then we choose $P_U$ and $x(u,s)$ such that it satisfies $Q_Z=Q_0$. 
By \cite[Theorem~1]{Hou13} one can show that \eqref{eq:QZ_CSIT_C} holds if,
\begin{align}
    R>\mi(U;Z).\label{eq:res_Simple_C}
\end{align}

\emph{{Decoding and Error Probability Analysis:}}
Upon receiving $y^n$ the receiver finds a unique message $\hat{m}$ such that $\big(u^n(\hat{m}),y^n\big)\in\mathcal{T}_{\epsilon}^{(n)}$. According to the law of large numbers and the packing lemma probability of error vanishes when $n$ grows if \cite{ElGamalKim},
\begin{align}
&R  < \mi(U;Y).
\label{eq:Decoding_Simple_C}
\end{align}


\section{Proof of Theorem~\ref{thm:Upper_Bound_For_CSIT_C}}
\label{sec:Proof_Upper_Bound_For_CSIT_C}
Consider any sequence of length-$n$ codes for a state-dependent channel with channel state available causally only at the transmitter such that $P_e^{(n)}\leq\epsilon_n$ and $\kld(P_{Z^n}||Q_0^{\otimes n})\leq\delta$ with $\lim_{n\to\infty}\epsilon_n=0$. Note that the converse is consistent with the model and does \emph{not} require $\delta$ to vanish.
\subsection{Epsilon Rate Region}
We first define a region $\mathcal{S}_{\epsilon}$ for $\epsilon>0$ that expands the region defined in~\eqref{eq:finalregion_CSIT_C} and~\eqref{eq:finalconstraint_CSIT_C} as follows.
\begin{subequations}\label{eq:Epsilon_Rate_Region_CSIT_C}
\begin{align}
\calS_\epsilon\eqdef \big\{R\geq 0: \exists P_{U,V,S,X,Y,Z}\in\calD_\epsilon: R\leq \mi(U;Y) + \epsilon \big\}\label{eq:Sepsilon_CSIT_C}
\end{align}
where 
\begin{align}
  \calD_\epsilon = \left.\begin{cases}P_{U,V,S,X,Y,Z}:\\
P_{U,V,S,X,Y,Z}=Q_SP_VP_{U|V}\indic{1}_{\big\{X=X(U,S)\big\}}W_{Y,Z|X,S}\\
\mathbb{D}\left(P_Z\Vert Q_0\right) \leq \epsilon\\
\mi(U;Y) \geq \mi(V;Z) - 3\epsilon\\
\max\{\card{\calU},\card{\calV}\}\leq \card{\calX}
\end{cases}\right\},\label{eq:Depsilon_CSIT_C}
\end{align}
\end{subequations}where $\epsilon\triangleq\max\{\epsilon_n,\nu\geq\frac{\delta}{n}\}$. 
We next show that if a rate $R$ is achievable then $R\in\mathcal{S}_{\epsilon}$ for any $\epsilon>0$. 

For any $\epsilon_n>0$, we start by upper bounding $nR$ using standard techniques.
\begin{align}
nR &=\ent(M)\nonumber\\
&\mathop \le \limits^{(a)}\mi(M;Y^n)+n\epsilon_n\nonumber\\
&= \sum\limits_{i = 1}^n \mi(M;Y_i|Y^{i-1})  + n\epsilon_n \nonumber\\
&\leq \sum\limits_{i = 1}^n \mi(M,Y^{i-1};Y_i)  + n\epsilon_n \nonumber\\
&\mathop \leq \limits^{(b)} \sum\limits_{i = 1}^n \mi(M,S^{i-1};Y_i)  + n\epsilon_n \nonumber\\
&\mathop = \limits^{(c)} \sum\limits_{i = 1}^n \mi(U_i;Y_i)  + n\epsilon_n \nonumber\\
& =  n\sum\limits_{i = 1}^n \frac{1}{n}\mi(U_i;Y_i)  + n\epsilon_n \nonumber\\
& =  n\sum\limits_{i = 1}^n \Prob(T=i)\mi(U_i;Y_i|T=i)  + n\epsilon_n \nonumber\\
& =  n\mi(U_T;Y_T|T)  + n\epsilon_n \nonumber\\
& \leq  n\mi(U_T,T;Y_T)  + n\epsilon_n \nonumber\\
&\mathop = \limits^{(d)}  n\mi(U;Y)  + n\epsilon_n\nonumber\\
&\mathop \leq \limits^{(e)}  n\mi(U;Y)  + n\epsilon,
\label{eq:Upper_Bound_on_MRate_CSIT_C_2_First_2}
\end{align}
where 
\begin{itemize}
    \item[$(a)$] follows from Fano's inequality;
    \item[$(b)$] follows since $(M,Y^{i-1})-(M,S^{i-1})-Y_i$, note that we also have  $V_i-(M,S^{i-1})-Y_i$, where $V_i\triangleq(M,Z^{i-1})$;
    \item[$(c)$]  follows by defining $U_i\triangleq(M,S^{i-1})$;
    \item[$(d)$] follows by defining $U=(U_T,T)$ and $Y=Y_T$;
    \item[$(e)$] follows by defining $\epsilon\triangleq\max\{\epsilon_n,\nu\geq\frac{\delta}{n}\}$.
\end{itemize}

Next, we lower bound $nR$ as follows,
\begin{align}
nR&\geq \ent(M)\nonumber\\
&\ge \mi(M;Z^n)\nonumber\\
&= \sum\limits_{i = 1}^n \mi(M;Z_i|Z^{i-1}) \nonumber\\
&\mathop \ge \limits^{(a)} \sum\limits_{i = 1}^n \mi(M,Z^{i-1};Z_i) -\delta  \nonumber\\
&\mathop = \limits^{(b)} \sum\limits_{i = 1}^n \mi(V_i;Z_i) - \delta \nonumber\\
&= n\sum\limits_{i = 1}^n \Prob(T=i)\mi(V_i;Z_i|T=i)  - \delta \nonumber\\
&= n\mi(V_T;Z_T|T)  - \delta \nonumber\\
&\mathop \ge \limits^{(c)} n\mi(V_T,T;Z_T) - 2\delta \nonumber\\
&\mathop = \limits^{(d)} n\mi(V;Z) - 2\delta
\label{eq:Upper_Bound_on_KRate_FCSI_C_Second_2}
\end{align}
where 
\begin{itemize}
    \item[$(a)$] follows from Lemma~\ref{lemma:Resolvability_Properties};
    \item[$(b)$] follows from the definition of $V_i\triangleq(M,Z^{i-1})$, which is defined in the process of deriving \eqref{eq:Upper_Bound_on_MRate_CSIT_C_2_First_2};
    \item[$(c)$] follows from Lemma~\ref{lemma:Resolvability_Properties};
    \item[$(d)$] follows by defining $V=(V_T,T)$ and $Z=Z_T$.
\end{itemize}
For any $\nu>0$, choosing $n$ large enough ensures that
\begin{align}
 R &\mathop \geq \mi(V;Z) - 2\nu\nonumber\\
 &\mathop \geq \mi(V;Z) - 2\epsilon,
 \label{eq:Upper_Bound_on_KRate_CSI_C}
\end{align}
where the last equality follows from definition $\epsilon\triangleq\max\{\epsilon_n,\nu\}$. To show that $ \kld(P_Z||Q_0)\leq\epsilon$, note that for $n$ large enough
\begin{align}
\kld(P_Z||Q_0)=\kld(P_{Z_T}||Q_0)=\kld\Bigg(\frac{1}{n}\sum\limits_{i=1}^nP_{Z_i}\Bigg|\Bigg|Q_0\Bigg)\leq\frac{1}{n}\sum\limits_{i=1}^n\kld(P_{Z_i}||Q_0)\leq\frac{1}{n}\kld(P_{Z^n}||Q_0^{\otimes n})\leq\frac{\delta}{n}\leq\nu\leq\epsilon.\label{eq:boundKL_CSIT_C}
\end{align}
Combining \eqref{eq:Upper_Bound_on_MRate_CSIT_C_2_First_2} and \eqref{eq:Upper_Bound_on_KRate_CSI_C}, and~\eqref{eq:boundKL_CSIT_C} shows that $\forall \epsilon_n,\nu>0$, $R\leq \max\{x:x\in\mathcal{S}_{\epsilon}\}$. Therefore,
\begin{align}
  R\leq \max\left\{x:x\in\bigcap_{\epsilon>0}\mathcal{S}_{\epsilon}\right\}.
\end{align}
\subsubsection{Continuity at Zero}
One can prove the continuity at zero of $\calS_\epsilon$ by substituting $\min\{\mi(U;Y)-\mi(U;S),\mi(U,V;Y)-\mi(U;S|V)\}$ with $\mi(U;Y)$ and $\mi(V;Z)-\mi(V;S)$ with $\mi(V;Z)$ in Appendix~\ref{sec:continuity-at-zero} and following the exact same arguments.

\section{Proof of Theorem~\ref{thm:Strictly_Causal_State_CSIT}}
\label{sec:Proof_Acivable_Rate_Condition_For_CSIT_SC}
We adopt a block-Markov encoding scheme in which $B$ independent messages are transmitted over $B$ channel blocks each of length $r$, such that $n=rB$.
The warden's observation is $Z^n=(Z_1^r,\dots,Z_B^r)$, the distribution induced at the output of the warden is $P_{Z^n}$, the target output distribution is $Q_Z^{\otimes n}$, and Equation \eqref{eq:Total_KLD_CoNC}, describing the distance between the two distributions, continues to  hold. 
The random code generation is as follows:

%

\sloppy Fix $P_X(x)$, $P_{V|S}(v|s)$, and $\epsilon_1>\epsilon_2>0$ such that, $P_Z = Q_0$.

\emph{Codebook Generation for Keys:}
For each block $j\in\intseq{1}{B}$, let $C_1^{(r)}\triangleq\big\{V^r(\ell_j)\big\}_{\ell_j\in\mathcal{L}}$, where $\mathcal{L}=\intseq{1}{2^{r\tilde{R}}}$, be a random codebook consisting of independent random sequences each generated according to $P_V^{\otimes r}$, where $P_V=\sum_{s\in\mathcal{S}}Q
_S(s)P_{V|S}(v|s)$. We denote a realization of $C_1^{(r)}$ by $\mathcal{C}_1^{(r)}\triangleq\big\{v^r(\ell_j)\big\}_{\ell_j\in\mathcal{L}}$. 
Partition the set of indices $\ell_j\in\intseq{1}{2^{r\tilde{R}}}$ into bins $\mathcal{B}(t)$, $t\in\intseq{1}{2^{rR_T}}$ by using function $\varphi:V^r(\ell_j)\to\intseq{1}{2^{rR_T}}$ through random binning by choosing the value of $\varphi(v^r(\ell_j))$ independently and uniformly at random for every $v^r(\ell_j)\in\mathcal{V}^r$. 
For each block $j\in\intseq{1}{B}$, create a function $\Phi:V^r(\ell_j)\to\intseq{1}{2^{rR_K}}$ through random binning by choosing the value of $\Phi(v^r(\ell_j))$ independently and uniformly at random for every $v^r(\ell_j)\in\mathcal{V}^r$. The key $k_j=\Phi(v^r(\ell_j))$ obtained in block $j\in\intseq{1}{B}$ from the description of the channel state sequence $v^r(\ell_j)$ is used to assist the encoder in block $j+2$.

\emph{Codebook Generation for Messages:}
For each block $j\in\intseq{1}{B}$, let $C_2^{(r)}\triangleq\big\{X^r(m_j,t_{j-1},k_{j-2})\big\}_{(m_j,t_{j-1},k_{j-2})\in\mathcal{M}\times\mathcal{T}\times\mathcal{K}}$, where $\mathcal{M}=\intseq{1}{2^{rR}}$, $\mathcal{T}=\intseq{1}{2^{rR_t}}$, and $\mathcal{K}=\intseq{1}{2^{rR_k}}$, be a random codebook consisting of independent random sequences each generated according to $P_X^{\otimes r}$. We denote a realization of $C_2^{(r)}$ by $\mathcal{C}_2^{(r)}\triangleq\big\{x^r(m_j,t_{j-1},k_{j-2})\big\}_{(m_j,t_{j-1},k_{j-2})\in\mathcal{M}\times\mathcal{T}\times\mathcal{K}}$. Let, $C_r=\big\{C_1^{(r)},C_2^{(r)}\big\}$ and $\mathcal{C}_r=\big\{\mathcal{C}_1^{(r)},\mathcal{C}_2^{(r)}\big\}$. The indices $(m_j,t_{j-1},k_{j-2})$ can be viewed as a two layer binning.  We define an ideal \ac{PMF} for codebook $\mathcal{C}_r$, as an approximate distribution to facilitate the analysis
\begin{align}
    &\Gamma_{M_j,T_{j-1},K_{j-2},L_j,X^r,V^r,S_j^r,Z_j^r,K_{j-1},T_j,K_j}^{(\mathcal{C}_r)}(m_j,t_{j-1},k_{j-2},\ell_j,\tilde{x}^r,\tilde{v}^r,s_j^r,z_j^r,k_{j-1},t_j,k_j)  \nonumber\\
    &\qquad= 2^{-r(R + R_t +R_k + \tilde{R})}\indic{1}_{\{ \tilde{x}^r=x^r(m_j,t_{j-1},k_{j-2})\}}\indic{1}_{\{\tilde{v}^r=v^r(\ell_j) \}} P_{S|V}^{\otimes r}(s_j^r| \tilde{v}^r)  W_{Z|X,S}^{\otimes r}(z_j^r|\tilde{x}^r,s_j^r)\nonumber\\
    &\qquad\quad\times 2^{-rR_k}\indic{1}_{\{ t_j=\varphi (\tilde{v}^r)\}}\indic{1}_{\{ k_j=\Phi (\tilde{v}^r)\}},\label{eq:Ideal_PMF_Strictly_Causal}
\end{align}where $W_{Z|X,S}$ is the marginal distribution of $W_{Y,Z|X,S}$ defined in Theorem~\ref{thm:Strictly_Causal_State_CSIT} and
\begin{align}
    P_{S|V}=\frac{P_{S,V}(s,v)}{P_{V}(v)}=\frac{Q_S(s)P_{V|S}(v|s)}{\sum_{s\in\mathcal{S}}Q_S(s)P_{V|S}(v|s)}.\label{eq:Marginal_LE_SC}
\end{align}

\emph{{Encoding:}}
We assume that the transmitter and the receiver have access to shared secret keys $k_{-1}$ and $k_0$ for the first two blocks but after the first two blocks they use the key that they generate from the \ac{CSI}.  

In the first block, to send the message $m_1$ according to $k_{-1}$, the encoder chooses the index $t_0$ uniformly at random (the index $t_0$ does not carry any useful information) and computes $x^r(m_1,t_0,k_{-1})$ and transmits it over the channel.

At the beginning of the second block, to generate a secret key shared between the transmitter and the receiver, the encoder chooses the index $\ell_1$ according to the following distribution with $j=1$,
\begin{align}
    f(\ell_j|s_j^r)=\frac{P_{S|V}^{\otimes r}(s_j^r|v^r(\ell_j))}{\sum\limits_{\ell'\in\intseq{1}{2^{r\tilde{R}}}} P_{S|V}^{\otimes r}(s_j^r|v^r(\ell'_j)) },\label{eq:Likelihood_Encoder_Str_Causal}
\end{align}where $P_{S|V}$ is defined in~\eqref{eq:Marginal_LE_SC}. Then generates the reconciliation index $t_1=\varphi(v^n(\ell_1))$; simultaneously the transmitter generates a key $k_1=\Phi(v^r(\ell_1))$ from the description of its \ac{CSI} of the first block $v^r(\ell_1)$ to be used in the next block. To transmit the message $m_2$ and reconciliation index $t_1$ according to the key $k_0$ the encoder computes $x^r(m_2,t_1,k_0)$ and transmits it over the channel.


In block $j\in\intseq{3}{B}$, the encoder first selects the index $\ell_{j-1}$ based $s_{j-1}^r$ by using the likelihood encoder described in \eqref{eq:Likelihood_Encoder_Str_Causal} and then generates the reconciliation index $t_{j-1}=\varphi(v^r(\ell_{j-1}))$; simultaneously the transmitter generates a key $k_{j-1}=\Phi(v^r(\ell_{j-1}))$ from the description of its \ac{CSI} of the block $j-1$, $v^r(\ell_{j-1})$, to be used in the next block. Then to send the message $m_j$ and the reconciliation index $t_{j-1}$ according to the generated key $k_{j-2}$ from the previous block, the encoder computes $x^r(m_j,t_{j-1},k_{j-2})$ and transmits it over the channel.

Define
\begin{align}
    \twocolalign &\Upsilon_{M_j,T_{j-1},K_{j-2},X^r,S_j^r,L_j,V^r,Z_j^r,K_{j-1},T_j,K_j}^{(\mathcal{C}_r)}(m_j,t_{j-1},k_{j-2},\tilde{x}^r,s_j^r,\ell_j,v^r,z_j^r,k_{j-1},t_j,k_j) \nonumber\\
    &\triangleq 2^{-r(R + R_t + R_k)}\indic{1}_{\{ \tilde{x}^r=x^r(m_j,t_{j-1},k_{j-2})\}} Q_S^{\otimes r}(s_j^r) f(\ell_j|s_j^r) \indic{1}_{\{ \tilde{v}^r=v^r(\ell_j)\}}\nonumber\\
    &\quad\times W_{Z|X,S}^{\otimes r}(z_j^r|\tilde{x}^r,s_j^r) 2^{-rR_k}\indic{1}_{\{ t_j=\sigma (v^r(\ell_j)) \}}\indic{1}_{\{ k_j=\Phi (v^r(\ell_j))\}}.\label{eq:P_Dist_SC_CSIT}
\end{align}For a fixed codebook $\mathcal{C}_r$, the induced joint distribution over the codebook (i.e. $P^{(\mathcal{C}_r)}$) satisfies
\begin{align}
\kld\Big(P_{M_j,T_{j-1},K_{j-2},X^r,S_j^r,L_j,V^r,Z_j^r,K_{j-1},T_j,K_j}^{(\mathcal{C}_r)}||\Upsilon_{M_j,T_{j-1},K_{j-2},X^r,S_j^r,L_j,V^r,Z_j^r,K_{j-1},T_j,K_j}^{(\mathcal{C}_n)}\Big)\leq\epsilon.\label{eq:Uniformity_Key_SC_CSIT}
\end{align}
This intermediate distribution $\Upsilon^{(\mathcal{C}_r)}$ approximates the true distribution $P^{(\mathcal{C}_r)}$ and will be used in the sequel for bounding purposes. Expression~\eqref{eq:Uniformity_Key_SC_CSIT} holds because the main difference in $\Upsilon^{(\mathcal{C}_r)}$ is assuming the keys $K_{j-2}$, $K_{j-1}$ and the reconciliation index $T_{j-1}$ are uniformly distributed, which is made (arbitrarily) nearly uniform in $P^{(\mathcal{C}_r)}$ with appropriate control of rate.

\begin{figure*}
\centering
\includegraphics[width=6.0in]{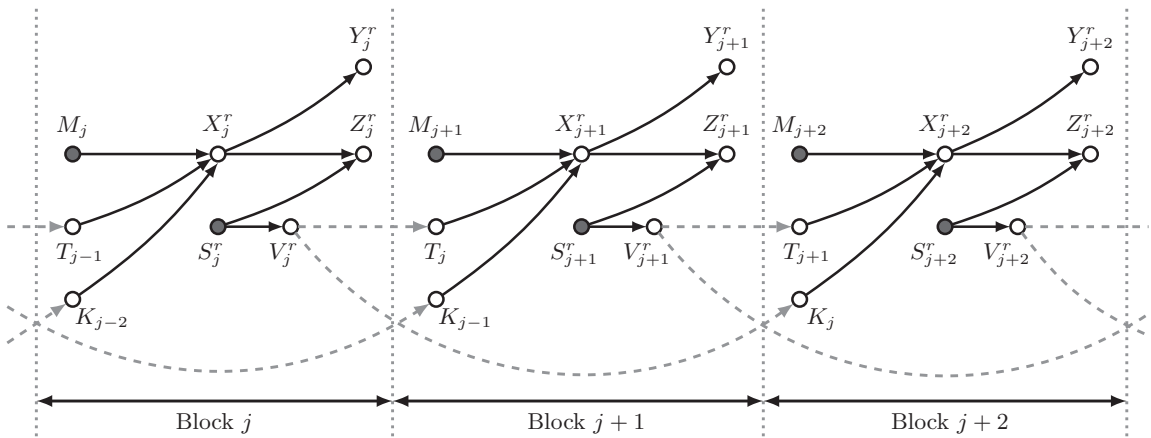}
\caption{Functional dependence graph for the block-Markov encoding scheme}
\label{fig:Chaining_SC_CSIT}

\end{figure*}
\emph{Covert Analysis:} We now show that this coding scheme  guarantees that $\expec_{C_n}[\kld(P_{Z^n|C_n} || Q_Z^{\otimes n} )]\underset{n\rightarrow\infty}{\xrightarrow{\hspace{0.2in}}} 0$, where $C_n$ is the set of all the codebooks for all blocks, and then choose $P_X$ and $P_V$ such that it satisfies $Q_Z=Q_0$. 
Similar to \eqref{eq:Total_KLD_Bound_NC_CSIT} by using functional dependence graph depicted in Fig.~\ref{fig:Chaining_SC_CSIT} one can show that
\begin{align}
    \kld(P_{Z^n}||Q_Z^{\otimes n})\leq 2\sum\limits_{j = 1}^B {\kld(P_{Z_j^r,K_{j-1},T_j,K_j}||Q_Z^{\otimes r}Q_{K_{j-1}}Q_{T_j}Q _{K_j})}.\label{eq:Total_KLD_Bound_SC_CSIT}
\end{align}
To bound the \ac{RHS} of \eqref{eq:Total_KLD_Bound_SC_CSIT} by using Lemma~\ref{lemma:TV_KLD} and the triangle inequality we have,
\begin{align}
        &\expec_{\mathcal{C}_r}||P_{Z_j^r,K_{j-1},T_j,K_j|C_r}-Q_Z^{\otimes r}Q_{K_{j-1}}Q_{T_j}Q _{K_j}||_1 \nonumber\\ &\leq\expec_{C_r}||P_{Z_j^r,K_{j-1},T_j,K_j|C_r}-\Gamma_{Z_j^r,K_{j-1},T_j,K_j|C_r}||_1  \twocolbreak+ \expec_{C_r}||\Gamma_{Z_j^r,K_{j-1},T_j,K_j|C_r}-Q_Z^{\otimes r}Q_{K_{j-1}}Q_{T_j}Q _{K_j}||_1\nonumber\\
    &\leq\expec_{C_r}||P_{Z_j^r,K_{j-1},T_j,K_j|C_r}-\Upsilon_{Z_j^r,K_{j-1},T_j,K_j|C_r}||_1 + \expec_{C_r}||\Upsilon_{Z_j^r,K_{j-1},T_j,K_j|C_r}-\Gamma_{Z_j^r,K_{j-1},T_j,K_j|C_r}||_1  \nonumber\\
    &\quad+ \expec_{C_r}||\Gamma_{Z_j^r,K_{j-1},T_j,K_j|C_r}-Q_Z^{\otimes r}Q_{K_{j-1}}Q_{T_j}Q _{K_j}||_1.\label{eq:Triangle_Inequality_Strictly_Causaul_CSIT}
\end{align}
From \eqref{eq:Uniformity_Key_SC_CSIT} and the monotonicity of KL-divergence the first term on the \ac{RHS} of \eqref{eq:Triangle_Inequality_Strictly_Causaul_CSIT} goes to zero when $n$ grows. To bound the second term on the \ac{RHS} of \eqref{eq:Triangle_Inequality_Strictly_Causaul_CSIT} for a fixed codebook $\mathcal{C}_r$ we have,
\begin{subequations}\label{eq:SC_CSIT_equalities}
\begin{align}
  \twocolalign\Gamma_{M_j,T_{j-1},K_{j-2}}^{(\mathcal{C}_r)} \onecolalign= 2^{-r(R+R_t+R_k)} = \Upsilon_{M_j,T_{j-1},K_{j-2}}^{(\mathcal{C}_r)},\label{eq:Uniformity_M_2_SC}\\
  \twocolalign\Gamma_{X^r|M_j,T_{j-1},K_{j-2},S_j^r}^{(\mathcal{C}_r)} \onecolalign= \indic{1}_{\{ \tilde{x}^r=x^r(m_j,t_{j-1},k_{j-2})\}} \twocolbreak= \Upsilon_{X^r|M_j,T_{j-1},K_{j-2},S_j^r}^{(\mathcal{C}_r)},\label{eq:eq4_SC}\\
  \twocolalign\Gamma_{L_j|M_j,T_{j-1},K_{j-2},S_j^r,X^r}^{(\mathcal{C}_r)} \onecolalign= f(\ell_j|s_j^r) = \Upsilon_{L_j|M_j,T_{j-1},K_{j-2},S_j^r,X^r}^{(\mathcal{C}_r)},\label{eq:Likelihood_Encoder_Equality_2_SC}\\
  \twocolalign\Gamma_{V^r|M_j,T_{j-1},K_{j-2},S_j^r,L_j,X^r}^{(\mathcal{C}_r)} \onecolalign= \indic{1}_{\{ \tilde{v}^r=v^r(\ell_j)\}} \twocolbreak= \Upsilon_{V^r|M_j,T_{j-1},K_{j-2},S_j^r,L_j,X^r}^{(\mathcal{C}_r)},\label{eq:eq5_SC}\\
    \twocolalign\Gamma_{Z_j^r|M_j,T_{j-1},K_{j-2},S_j^r,L_j,X^r,V^r}^{(\mathcal{C}_r)} \onecolalign= W_{Z|X,S}^{\otimes r}  = \Upsilon_{Z_j^r|M_j,T_{j-1},K_{j-2},S_j^r,L_j,X^r,V^r}^{(\mathcal{C}_r)},\label{eq:Channel_Equality_2_SC}\\
    \twocolalign\Gamma_{K_{j-1}|M_j,T_{j-1},K_{j-2},S_j^r,L_j,X^r,V^r,Z_j^r}^{(\mathcal{C}_r)} \onecolalign= 2^{-rR_k}  = \Upsilon_{K_{j-1}|M_j,T_{j-1},K_{j-2},S_j^r,L_j,X^r,V^r,Z_j^r}^{(\mathcal{C}_r)},\label{eq:Second_Key_SC}\\
    \twocolalign\Gamma_{T_j|M_j,T_{j-1},K_{j-2},S_j^r,L_j,X^r,V^r,Z_j^r}^{(\mathcal{C}_r)} \onecolalign=\indic{1}_{\{ t_j=\sigma (v^r(\ell_j)) \}}=\Upsilon_{T_j|M_j,T_{j-1},K_{j-2},S_j^r,L_j,X^r,V^r,Z_j^r}^{(\mathcal{C}_r)},\label{eq:K_j_Distribution_SC}\\
    \twocolalign\Gamma_{K_j|M_j,T_{j-1},K_{j-2},S_j^r,L_j,X^r,V^r,Z_j^r,T_j}^{(\mathcal{C}_r)} \onecolalign=\indic{1}_{\{ k_j=\Phi (v^r(\ell_j))\}}=\Upsilon_{K_j|M_j,T_{j-1},K_{j-2},S_j^r,L_j,X^r,V^r,Z_j^r,T_j}^{(\mathcal{C}_r)},\label{eq:T_j_Distribution_SC}
\end{align}
\end{subequations}
where \eqref{eq:Likelihood_Encoder_Equality_2_SC} follows from \eqref{eq:Likelihood_Encoder_Str_Causal}. Hence,
\begin{align}
 &\expec_{C_r}||\Upsilon_{Z_j^r,K_{j-1},T_j,K_j|C_r}-\Gamma_{Z_j^r,K_{j-1},T_j,K_j}||_1
 \twocolbreak \nonumber\\ &\leq\expec_{C_r}||\Upsilon_{M_j,T_{j-1},K_{j-2},S_j^r,L_j,X^r,V^r,Z_j^r,K_{j-1},T_j,K_j|C_r}- \Gamma_{M_j,T_{j-1},K_{j-2},S_j^r,L_j,X^r,V^r,Z_j^r,K_{j-1},T_j,K_j|C_r}||_1 \nonumber\\
&\mathop = \limits^{(a)}\expec_{C_r}||\Upsilon_{M_j,T_{j-1},K_{j-2},S_j^r|C_r}- \Gamma_{M_j,T_{j-1},K_{j-2},S_j^r|C_r}||_1 \nonumber\\
&\mathop = \limits^{(b)} \expec_{C_r}||Q_S^{\otimes r} - \Gamma_{S_j^r|M_j=1,T_{j-1}=1,K_{j-2}=1,C_r}||_1,\label{eq:General_Distribution_Expansion_2_SC}
\end{align}where $(a)$ follows from \eqref{eq:eq4_SC}-\eqref{eq:T_j_Distribution_SC} and $(b)$ follows from the symmetry of the codebook construction with respect to $M_j$, $T_{j-1}$, and $K_{j-2}$ and \eqref{eq:Uniformity_M_2_SC}. Based on the soft covering lemma \cite[Corollary~VII.5]{Cuff13} the \ac{RHS} of \eqref{eq:General_Distribution_Expansion_2_SC} vanishes if
\begin{align}\label{eq:State_Dependent_Covering_Lemma_2_SC}
    \tilde{R} > \mi(S;V).
\end{align}
We now proceed to bound the third term on the \ac{RHS} of \eqref{eq:Triangle_Inequality_Strictly_Causaul_CSIT}. First, consider the marginal
\begin{align}
    &\Gamma_{Z_j^r,K_{j-1},T_j,K_j|C_r} (z_j^r,k_{j-1},t_j,k_j)= \sum\limits_{m_j} \sum\limits_{t_{j-1}}\sum\limits_{k_{j-2}}\sum\limits_{\ell_j} \sum\limits_{s_j^r} 
    \frac{1}{2^{r(R + R_t + R_k+\tilde{R})}} P_{S|V}^{\otimes r}\big(s_j^r|V^r(\ell_j)\big)\nonumber\\
    &\quad
    \times W_{Z|X,S}^{\otimes r}\big(z_j^r|X^r(m_j,t_{j-1},k_{j-2}),s_j^r\big) 2^{-rR_k} \indic{1}_{\{ t_j=\sigma(V^r(\ell_j))\}}\indic{1}_{\{ k_j=\Phi (V^r(\ell_j))\}}\nonumber\\
    &= \sum\limits_{m_j} \sum\limits_{t_{j-1}}\sum\limits_{k_{j-2}}\sum\limits_{\ell_j} 
    \frac{1}{2^{r(R + R_t + R_k+\tilde{R})}}  W_{Z|X,V}^{\otimes r}\big(z_j^r|X^r(m_j,t_{j-1},k_{j-2}),V^r(\ell_j)\big) 2^{-rR_k} \indic{1}_{\{ t_j=\sigma(V^r(\ell_j))\}}\indic{1}_{\{ k_j=\Phi (V^r(\ell_j))\}},
\end{align}where $W_{Z|X,V}(z|x,v)=\sum_{s\in\mathcal{S}} P_{S|V}(s|v)W_{Z|X,V}(z|x,v)$. To bound the third term on the \ac{RHS} of \eqref{eq:Triangle_Inequality_Strictly_Causaul_CSIT} by using Pinsker's inequality, it is sufficient to bound $\expec_{C_r}[\kld({\Gamma_{Z_j^r,K_{j-1},T_j,K_j|C_r}||Q_Z^{\otimes r}Q_{K_{j-1}}Q_{T_j}Q_{K_j}})]$ as follows,
\begin{align}
&\expec_{C_r}[\kld({\Gamma_{Z_j^r,K_{j-1},T_j,K_j|C_r}||Q_Z^{\otimes r} Q_{K_{j-1}}Q_{T_j}Q_{K_j}})] \nonumber\\ 
&= \expec_{C_r}\Big[\sum\limits_{z_j^r,k_{j-1},t_j,k_j} \Gamma_{Z_j^r,K_{j-1},T_j,K_j|C_r}(z_j^r,k_{j-1},t_j,k_j)\log \bigg(\frac{\Gamma_{Z_j^r,K_{j-1},T_j,K_j|C_r}(z_j^r,k_{j-1},t_j,k_j)}{Q_Z^{\otimes r}(z_j^r) Q_{K_{j-1}}(k_{j-1}) Q_{T_j}(t_j) Q_{K_j}(k_j)}\bigg) \Big] \nonumber\\ 
&= \expec_{C_r}\Bigg[\sum\limits_{(z_j^r,k_{j-1},t_j,k_j)}\sum\limits_{m_j} \sum\limits_{t_{j-1}}\sum\limits_{k_{j-2}}  \sum\limits_{\ell_j}\frac{1}{2^{r(R + R_t + 2R_k+\tilde{R})}}W_{Z|X,V}^{\otimes r}\big(z_j^r|X^r(m_j,t_{j-1},k_{j-2}),V^r(\ell_j)\big)\nonumber\\
&\times\indic{1}_{\{ t_j=\varphi(V^r(\ell_j))\}} \indic{1}_{\{ k_j=\Phi (V^r(\ell_j))\}} \nonumber\\
&\times\log\Bigg(\frac{\sum\limits_{\tilde{m}_j}\sum\limits_{\tilde{t}_{j-1}}\sum\limits_{\tilde{k}_{j-2}} \sum\limits_{\tilde{\ell}_{j}}  W_{Z|X,V}^{\otimes r}\big(z_j^r|X^r(\tilde{m}_j,\tilde{t}_{j-1},\tilde{k}_{j-2}),V^r(\tilde{\ell}_{j})\big)\indic{1}_{\{ t_j=\varphi(V^r(\tilde{\ell}_{j}))\}} \indic{1}_{\{ k_j=\Phi (V^r(\tilde{\ell}_{j}))\}}}{2^{r(R + R_t + R_k +\tilde{R} - R_T - R_K)} Q_Z^{\otimes r}(z_j^r)}\Bigg)     \Bigg]\nonumber\\ 
&\mathop\le\limits^{(a)}\sum\limits_{(z_j^r,k_{j-1},t_j,k_j)} \sum\limits_{m_j} \sum\limits_{t_{j-1}}\sum\limits_{k_{j-2}}\sum\limits_{\ell_j} \frac{1}{2^{r(R + R_t + 2R_k+\tilde{R})}}\sum\limits_{\big(x^r(m_j,t_{j-1},k_{j-2}),s_j^r,v^r(\ell_j)\big)} \hspace{-2mm}\Gamma^{(\mathcal{C}_r)}_{X^r,S_j^r,V^r,Z_j^r}\big(x^r(m_j,t_{j-1},k_{j-2}),s_j^r,v^r(\ell_j),z_j^r\big)\nonumber\\
&\times\expec_{\varphi(v^r(\ell_j))}\big[\indic{1}_{\{ t_j=\varphi(v^r(\ell_j))\}}\big]\times\expec_{\Phi(v^r(\ell_j))}\big[\indic{1}_{\{ k_j=\Phi (v^r(\ell_j))\}}\big] \nonumber\\
&\times\log \expec_{\mathop {\backslash (m_j,t_{j-1},k_{j-2},\ell_j),}\limits_{\backslash (\varphi (v^r(\ell_j)),\Phi (v^r(\ell_j)))} } \Bigg[ \frac{\sum\limits_{\tilde{m}_j}\sum\limits_{\tilde{t}_{j-1}}\sum\limits_{\tilde{k}_{j-2}}  \sum\limits_{\tilde{\ell}_{j}} W_{Z|X,V}^{\otimes r}\big(z_j^r|X^r(\tilde{m}_j,\tilde{t}_{j-1},\tilde{k}_{j-2}),V^r(\tilde{\ell}_j)\big)\indic{1}_{\{ t_j=\varphi(V^r(\tilde{\ell}_j))\}} \indic{1}_{\{ k_j=\Phi (V^r(\tilde{\ell}_j))\}}}{2^{r(R + R_t + R_k +\tilde{R}  - R_T - R_K)} Q_Z^{\otimes r}(z_j^r)}  \Bigg]      \nonumber\\
&\mathop\le\limits^{(b)} \sum\limits_{(z_j^r,k_{j-1},t_j,k_j)} \sum\limits_{m_j} \sum\limits_{t_{j-1}}\sum\limits_{k_{j-2}}\sum\limits_{\ell_j} \frac{1}{2^{r(R + R_t + 2R_k+\tilde{R})}}\sum\limits_{\big(x^r(m_j,t_{j-1},k_{j-2}),s_j^r,v^r(\ell_j)\big)} \hspace{-2mm}\Gamma^{(\mathcal{C}_r)}_{X^r,S_j^r,V^r,Z_j^r}\big(x^r(m_j,t_{j-1},k_{j-2}),s_j^r,v^r(\ell_j),z_j^r\big)\nonumber\\
&\times \frac{1}{2^{r(R_T+R_K)}}\log \frac{1}{2^{r(R + R_t + R_k +\tilde{R}  - R_T - R_K)} Q_Z^{\otimes r}(z_j^r)}\Bigg( W_{Z|X,V}^{\otimes r}\big(z_j^r|x^r(m_j,t_{j-1},k_{j-2}),v^r(\ell_j)\big) \nonumber\\
&+  \expec_{\backslash (m_j,t_{j-1},k_{j-2})}\Bigg[\sum\limits_{(\tilde{m}_j,\tilde{t}_{j-1},\tilde{k}_{j-2}) \ne (m_j,t_{j-1},k_{j-2})}   W_{Z|X,V}^{\otimes r}\big(z_j^r|X^r(\tilde{m}_j,\tilde{t}_{j-1},\tilde{k}_{j-2}),v^r(\ell_j)\big)\Bigg]\nonumber\\
& +  \expec_{\mathop {\backslash \ell_j,}\limits_{\backslash (\varphi (v^r(\ell_j)),\Phi (v^r(\ell_j)))} }\Bigg[\sum\limits_{\tilde{\ell}_{j} \ne \ell_j} W_{Z|X,V}^{\otimes r}\big(z_j^r|x^r(m_j,t_{j-1},k_{j-2}),V^r(\tilde{\ell}_j)\big)\indic{1}_{\{ t_j=\varphi (V^r(\tilde{\ell}_j))\}}\indic{1}_{\{ k_j=\Phi (V^r(\tilde{\ell}_j))\}}\Bigg]\nonumber\\
&+  \expec_{\mathop {\backslash (m_j,t_{j-1},k_{j-2},\ell_j),}\limits_{\backslash (\varphi (v^r(\ell_j)),\Phi (v^r(\ell_j)))} }\Bigg[
\sum\limits_{\tilde{\ell}_{j} \ne \ell_j}\sum\limits_{(\tilde{m}_j,\tilde{t}_{j-1},\tilde{k}_{j-2}) \ne (m_j,t_{j-1},k_{j-2})} W_{Z|X,V}^{\otimes r}\big(z_j^r|X^r(\tilde{m}_j,\tilde{t}_{j-1},\tilde{k}_{j-2}),V^r(\tilde{\ell}_j)\big)\nonumber\\
&\times\indic{1}_{\{ t_j=\varphi (V^r(\tilde{\ell}_j))\}}\indic{1}_{\{ k_j=\Phi (V^r(\tilde{\ell}_j))\}}\Bigg]\Bigg)    \nonumber\\
&\mathop\le\limits^{(c)} \sum\limits_{(z_j^r,k_{j-1},t_j,k_j)} \sum\limits_{m_j} \sum\limits_{t_{j-1}}\sum\limits_{k_{j-2}}\sum\limits_{\ell_j} \frac{1}{2^{r(R + R_t + 2R_k+\tilde{R})}}\sum\limits_{\big(x^r(m_j,t_{j-1},k_{j-2}),s_j^r,v^r(\ell_j)\big)} \hspace{-2mm}\Gamma^{(\mathcal{C}_r)}_{X^r,S_j^r,V^r,Z_j^r}\big(x^r(m_j,t_{j-1},k_{j-2}),s_j^r,v^r(\ell_j),z_j^r\big)\nonumber\\
&\times \frac{1}{2^{r(R_T+R_K)}}\log \Bigg(\frac{ W_{Z|X,V}^{\otimes r}\big(z_j^r|x^r(m_j,t_{j-1},k_{j-2}),v^r(\ell_j)\big)}{2^{r(R + R_t + R_k +\tilde{R}  - R_T - R_K)} Q_Z^{\otimes r}(z_j^r)} \nonumber\\
&+\sum\limits_{(\tilde{m}_j,\tilde{t}_{j-1},\tilde{k}_{j-2}) \ne (m_j,t_{j-1},k_{j-2})}\frac{ W_{Z|V}^{\otimes r}\big(z_j^r|v^r(\ell_j)\big)}{{{2^{r(R + R_t + R_k +\tilde{R}  - R_T - R_K)}} Q_Z^{\otimes r}(z_j^r)}}+  \sum\limits_{\tilde{\ell}_{j} \ne \ell_j}\frac{ W_{Z|X}^{\otimes r}\big(z_j^r|x^r(m_j,t_{j-1},k_{j-2})\big)}{{{2^{r(R + R_t + R_k +\tilde{R})}}\times Q_Z^{\otimes r}(z_j^r)}}+1\Bigg)\nonumber\\
&\le \sum\limits_{(z_j^r,k_{j-1},t_j,k_j)} \sum\limits_{m_j} \sum\limits_{t_{j-1}}\sum\limits_{k_{j-2}}\sum\limits_{\ell_j} \frac{1}{2^{r(R + R_t + 2R_k+\tilde{R}+R_T+R_K)}}\sum\limits_{\big(x^r(m_j,t_{j-1},k_{j-2}),s_j^r,v^r(\ell_j)\big)} \hspace{-12mm}\Gamma^{(\mathcal{C}_r)}_{X^r,S_j^r,V^r,Z_j^r}\big(x^r(m_j,t_{j-1},k_{j-2}),s_j^r,v^r(\ell_j),z_j^r\big)\nonumber\\
&\times\log \Bigg(\frac{ W_{Z|X,V}^{\otimes r}\big(z_j^r|x^r(m_j,t_{j-1},k_{j-2}),v^r(\ell_j)\big)}{2^{r(R + R_t + R_k +\tilde{R}  - R_T - R_K)} Q_Z^{\otimes r}(z_j^r)} +\frac{ W_{Z|V}^{\otimes r}\big(z_j^r|v^r(\ell_j)\big)}{{{2^{r(\tilde{R}- R_T - R_K)}} Q_Z^{\otimes r}(z_j^r)}}+  \frac{  W_{Z|X}^{\otimes r}\big(z_j^r|x^r(m_j,t_{j-1},k_{j-2})\big)}{{{2^{r(R + R_t + R_k)}} Q_Z^{\otimes r}(z_j^r)}} +1\Bigg) \nonumber\\
&\triangleq \Psi_1 + \Psi_2,\label{eq:KLD_ZK_SC_CSIT}
\end{align}
where $(a)$ follows from Jensen's inequality, $(b)$ and $(c)$ hold because $\indic{1}_{\{ k_j=\Phi (\tilde s_j^r)\}}\leq 1$. We defined $\Psi_1$ and $\Psi_2$ as
\begin{align}
&{\Psi _1} = \sum\limits_{(k_{j-1},t_j,k_j)} \sum\limits_{m_j} \sum\limits_{t_{j-1}}\sum\limits_{k_{j-2}} \sum\limits_{\ell_j} \frac{1}{{{2^{r(R + R_t + 2R_k + R_T +R_K)}}}}\nonumber\\
&\qquad\times\sum\limits_{\big(v^r,x^r(m_j,t_{j-1},k_{j-2}),z_j^r\big) \in \mathcal{T}_\epsilon ^{(r)}} \Gamma^{(\mathcal{C}_r)}_{V^r,X^r,Z^r}(v^r,x^r(m_j,t_{j-1},k_{j-2}),z_j^r)\nonumber\\
&\times\log \Bigg(\frac{ W_{Z|X,V}^{\otimes r}\big(z_j^r|x^r(m_j,t_{j-1},k_{j-2}),v^r(\ell_j)\big)}{2^{r(R + R_t + R_k +\tilde{R}  - R_T - R_K)} Q_Z^{\otimes r}(z_j^r)} +\frac{ W_{Z|V}^{\otimes r}\big(z_j^r|v^r(\ell_j)\big)}{{{2^{r(\tilde{R}- R_T - R_K)}} Q_Z^{\otimes r}(z_j^r)}}+  \frac{  W_{Z|X}^{\otimes r}\big(z_j^r|x^r(m_j,t_{j-1},k_{j-2})\big)}{{{2^{r(R + R_t + R_k)}} Q_Z^{\otimes r}(z_j^r)}} +1\Bigg) \nonumber\\
& \le \log \Bigg(\frac{2^{r(R_T+R_K)}\times 2^{ - r(1 - \epsilon )\ent(Z|X,V)}}{2^{r(R + R_t + R_k+\tilde{R})}\times 2^{ - r(1 + \epsilon )\ent(Z)}} + \frac{2^{r(R_T+R_K)}\times 2^{ - r(1 - \epsilon )\ent(Z|V)}}{2^{r\tilde{R}}\times 2^{ - r(1 + \epsilon )\ent(Z)}} +\frac{ 2^{ - r(1 - \epsilon )\ent(Z|X)}}{2^{r(R+R_t+R_k)}\times 2^{ - r(1 + \epsilon )\ent(Z)}}+ 1\Bigg)\label{eq:Psi_1_SC_CSIT}\\
&{\Psi _2} =\sum\limits_{(k_{j-1},t_j,k_j)} \sum\limits_{m_j} \sum\limits_{t_{j-1}}\sum\limits_{k_{j-2}} \sum\limits_{\ell_j} \frac{1}{{{2^{r(R + R_t + 2R_k + R_T +R_K)}}}}\nonumber\\
&\qquad\times\sum\limits_{\big(v^r,x^r(m_j,t_{j-1},k_{j-2}),z_j^r\big) \notin \mathcal{T}_\epsilon ^{(r)}} \Gamma^{(\mathcal{C}_r)}_{V^r,X^r,Z^r}(v^r,x^r(m_j,t_{j-1},k_{j-2}),z_j^r)\nonumber\\
&\times\log \Bigg(\frac{ W_{Z|X,V}^{\otimes r}\big(z_j^r|x^r(m_j,t_{j-1},k_{j-2}),v^r(\ell_j)\big)}{2^{r(R + R_t + R_k +\tilde{R}  - R_T - R_K)} Q_Z^{\otimes r}(z_j^r)} +\frac{ W_{Z|V}^{\otimes r}\big(z_j^r|v^r(\ell_j)\big)}{{{2^{r(\tilde{R}- R_T - R_K)}} Q_Z^{\otimes r}(z_j^r)}}+  \frac{  W_{Z|X}^{\otimes r}\big(z_j^r|x^r(m_j,t_{j-1},k_{j-2})\big)}{{{2^{r(R + R_t + R_k)}} Q_Z^{\otimes r}(z_j^r)}} +1\Bigg) \nonumber\\
&\le 2|V||X||Z|e^{ - r\epsilon^2 {\mu _{V,X,Z}}}r\log\Big(\frac{3}{\mu_Z} + 1\Big). \label{eq:Psi_2_SC_CSIT}
\end{align}
In \eqref{eq:Psi_2_SC_CSIT} $\mu_{X,V,Z}=\min\limits_{(x,v,z)\in(\mathcal{X},\mathcal{V},\mathcal{Z})}P_{X,V,Z}(x,v,z)$ and $\mu_Z=\min\limits_{z\in\mathcal{Z}}P_Z(z)$. When $r\to\infty$ then $\Psi_2\to 0$ and $\Psi_1$ goes to zero when $r$ grows if
\begin{subequations}\label{eq:Covert_Analysis_Str_Causal_CSIT}
\begin{align}
    R + R_t + R_k + \tilde{R} - R_T - R_K &>\mi(X,V;Z),\label{eq:Covert_Analysis_Str_Causal_CSIT_1}\\
    \tilde{R} - R_T - R_K &> \mi(V;Z),\label{eq:Covert_Analysis_Str_Causal_CSIT_2}\\
    R + R_t + R_k  &>\mi(X;Z).\label{eq:Covert_Analysis_Str_Causal_CSIT_4}
\end{align}
\end{subequations}

\emph{{Decoding and Error Probability Analysis:}}
At the end of block $j\in\intseq{1}{B}$, using its knowledge of the key $k_{j-2}$ generated from the block $j-2$, the receiver finds a unique pair $(\hat{m}_j,\hat{t}_{j-1})$ such that $\big(x^r(\hat{m}_j,\hat{t}_{j-1},k_{j-2}),y_j^r\big)\in\mathcal{T}_{\epsilon}^{(r)}$. According to the law of large numbers and the packing lemma probability of error vanishes when $r$ grows if \cite{ElGamalKim},
\begin{align}
&R + R_t < \mi(X;Y).
\label{eq:Decoding_BM_SC_CSIT}
\end{align}
We now analyze the probability of error at the encoder and the decoder for key generation. Let $(L_{j-1},T_{j-1})$ denote the chosen indices at the encoder and $\hat{L}_{j-1}$ and $\hat{T}_{j-1}$ be the estimate of the index $L_{j-1}$ and $T_{j-1}$ at the decoder. At the end of block $j$, by decoding $U_j^r$, the receiver has access to $T_{j-1}$ and to find $L_{j-1}$ we define the error event,
\begin{align}
    \mathcal{E} = \Big\{ \big(V_{j-1}^r(\hat{L}_{j-1}),S_{j-1}^r,X_{j-1}^r,Y_{j-1}^r\big) \notin \mathcal{T}_\epsilon ^{(n)}\Big\},
\end{align}and consider the events
\begin{subequations}\label{eq:Encoder_Errors_SC_CSIT}
\begin{align}
    \mathcal{E}_1 &= \Big\{ \big(V_{j-1}^r(\ell_{j-1}),S_{j-1}^r\big) \notin \mathcal{T}_{\epsilon'}^{(n)} \,\,\mbox{for all}\,\, \ell_{j-1}\in\intseq{1}{2^{r\tilde{R}}}\Big\},\\
    \mathcal{E}_2 &= \Big\{ \big(V_{j-1}^r(L_{j-1}),S_{j-1}^r,X_{j-1}^r,Y_{j-1}^r\big) \notin \mathcal{T}_\epsilon ^{(n)}\Big\},\\
    \mathcal{E}_3 &= \Big\{ \big(V_{j-1}^r(\tilde{\ell}_{j-1}),X_{j-1}^r,Y_{j-1}^r\big) \in \mathcal{T}_{\epsilon}^{(n)} \,\,\mbox{for some}\,\, \ell_{j-1}\in\mathcal{B}(T_{j-1}),\tilde{\ell}_{j-1}\ne \ell_{j-1}\Big\}.
\end{align}
\end{subequations}By the union bound we have
\begin{align}
    P(\mathcal{E})\leq P(\mathcal{E}_1) + P(\mathcal{E}_1^c\cap\mathcal{E}_2) + P(\mathcal{E}_3)\label{eq:Error_Event_SCausal_CSIT}.
\end{align}
According to \cite[Lemma~2]{Ziv_Strong_Sec_BC_Cooperation} the first term on the \ac{RHS} of \eqref{eq:Error_Event_SCausal_CSIT} vanishes when $r$ grows if we have \eqref{eq:State_Dependent_Covering_Lemma_2_SC}. Following the steps in \cite[Sec.~11.3.1]{ElGamalKim}, the last two terms on the right hand side of \eqref{eq:Error_Event_SCausal_CSIT} go to zero when $r$ grows if we have
\begin{subequations}\label{eq:Error_Analysis_Encoder_SC_CSIT}
\begin{align}
    \tilde{R}&>\mi(V;S),\label{eq:Error_Analysis_Encoder_SC_CSIT_1}\\
    \tilde{R}-R_t&<\mi(V;X,Y).\label{eq:Error_Analysis_Encoder_SC_CSIT_2}
\end{align}
\end{subequations}


Applying Fourier-Motzkin to 
\eqref{eq:State_Dependent_Covering_Lemma_2_SC}, \eqref{eq:Covert_Analysis_Str_Causal_CSIT},  \eqref{eq:Decoding_BM_SC_CSIT}, and \eqref{eq:Error_Analysis_Encoder_SC_CSIT} and remarking that the scheme requires $R_t + R_k\leq R_T + R_K$ results in the region in Theorem~\ref{thm:Strictly_Causal_State_CSIT}.

\section{Proof of Theorem~\ref{thm:Strictly_Causal_State_CSIT_App}}
\label{sec:Proof_Simple_Inner_Bound_For_CSIT_SC}
\sloppy Fix $P_X(x)$ and $\epsilon_1>\epsilon_2>0$ such that, $P_Z = Q_0$.

\emph{Codebook Generation:}
Let $C_n\triangleq\{X^n(m)\}_{m\in\mathcal{M}}$, where $\mathcal{M}\in\intseq{1}{2^{nR}}$, be a random codebook consisting of independent random sequences each generated according to $\prod\nolimits_{i = 1}^n P_X(x_i)$. We denote a realization of $C_n$ by $\CodeBook_n\triangleq\{x^n(m)\}_{m\in\mathcal{M}}$. 

\emph{{Encoding:}}
To send the message $m$ the encoder computes $x^n(m)$ and transmits it over the channel. 

For a fixed codebook $\mathcal{C}_n$, the induced joint distribution over the codebook is as follows
\begin{align}
    \twocolalign &P_{M,S^n,X^n,Z^n}^{(\mathcal{C}_n)}(m,s^n,\tilde{x}^n,z^n)= 2^{-nR} Q_S^{\otimes n}(s^n) \indic{1}_{\{ \tilde{x}^n=x^n(m)\}} W_{Z|X,S}^{\otimes n}(z^n|\tilde{x}^n,s^n).\label{eq:P_Dist_Simple_SC}
\end{align}

\emph{{Covert Analysis:}} 
We now show that this coding scheme guarantees that
\begin{align}
    \expec_{C_n}\big[\kld(P_{Z^n|C_n} || Q_Z^{\otimes n} )\big]\underset{n\rightarrow\infty}{\xrightarrow{\hspace{0.2in}}} 0,\label{eq:QZ_CSIT_SC}
\end{align}
where $P_{Z^n|C_n}$ is the marginal distribution of the distribution induced by the code design defined in \eqref{eq:P_Dist_Simple_SC} and is as follows
\begin{align}
    \twocolalign P_{Z^n|C_n}(z^n)&=\sum\limits_{m}\sum\limits_{s^n} 2^{-nR} Q_S^{\otimes n}(s^n) W_{Z|X,S}^{\otimes n}\big(z^n|x^n(m),s^n\big)\\
    &=\sum\limits_{m}2^{-nR} W_{Z|X}^{\otimes n}\big(z^n|x^n(m)\big),\label{eq:P_Marginal_Dist_Simple_SC}
\end{align}where $W_{Z|X}=\sum_{s\in\mathcal{S}} Q_S(s)W_{Z|X,S}(z|x,s)$. Then we choose $P_X$ such that it satisfies $Q_Z=Q_0$.

By \cite[Theorem~1]{Hou13} one can show that \eqref{eq:QZ_CSIT_SC} holds if
\begin{align}
    R>\mi(X;Z).\label{eq:res_Simple_SC}
\end{align}

\emph{{Decoding and Error Probability Analysis:}}
Upon receiving $y^n$ the receiver finds a unique message $\hat{m}$ such that $\big(x^n(\hat{m}),y^n\big)\in\mathcal{T}_{\epsilon}^{(n)}$.  According to the law of large numbers and the packing lemma probability of error vanishes when $n$ grows if \cite{ElGamalKim},
\begin{align}
&R  < \mi(X;Y).
\label{eq:Decoding_Simple_SC}
\end{align}


\section{Proof of Theorem~\ref{thm:Upper_Bound_For_CSIT_SC}}
\label{sec:Proof_Upper_Bound_For_CSIT_SC}
Consider any sequence of length-$n$ codes for a state-dependent channel with channel state available strictly causally only at the transmitter such that $P_e^{(n)}\leq\epsilon_n$ and $\kld(P_{Z^n}||Q_0^{\otimes n})\leq\delta$ with $\lim_{n\to\infty}\epsilon_n=0$. Note that the converse is consistent with the model and does \emph{not} require $\delta$ to vanish.
\subsection{Epsilon Rate Region}
We first define a region $\mathcal{S}_{\epsilon}$ for $\epsilon>0$ that expands the region defined in~\eqref{eq:finalregion_CSIT_SC} and~\eqref{eq:finalconstraint_CSIT_SC} as follows.
\begin{subequations}\label{eq:Epsilon_Rate_Region_CSIT_SC}
\begin{align}
\calS_\epsilon\eqdef \big\{R\geq 0: \exists P_{V,S,X,Y,Z}\in\calD_\epsilon: R\leq \mi(X;Y) + \epsilon \big\}\label{eq:Sepsilon_CSIT_SC}
\end{align}
where 
\begin{align}
  \calD_\epsilon = \left.\begin{cases}P_{V,S,X,Y,Z}:\\
P_{V,S,X,Y,Z}=Q_SP_VP_{X|V}W_{Y,Z|X,S}\\
\mathbb{D}\left(P_Z\Vert Q_0\right) \leq \epsilon\\
\mi(X;Y) \geq \mi(V;Z) - 3\epsilon\\
\card{\calV}\leq \card{\calX}
\end{cases}\right\},\label{eq:Depsilon_CSIT_SC}
\end{align}
\end{subequations}where $\epsilon\triangleq\max\{\epsilon_n,\nu\geq\frac{\delta}{n}\}$. 
We next show that if a rate $R$ is achievable then $R\in\mathcal{S}_{\epsilon}$ for any $\epsilon>0$. 

For any $\epsilon_n>0$, we start by upper bounding $nR$ using standard techniques. 
\begin{align}
nR &= \ent(M)\nonumber\\
&\mathop \le \limits^{(a)} \mi(M;Y^n) + n\epsilon_n \nonumber\\
&= \sum\limits_{i = 1}^n \mi(M;Y_i|Y^{i - 1})  + n\epsilon_n \nonumber\\
&\leq \sum\limits_{i = 1}^n \mi(M,Y^{i - 1};Y_i)  + n\epsilon_n \nonumber\\
&\mathop \le \limits^{(b)} \sum\limits_{i = 1}^n \mi(M,S^{i - 1};Y_i)  + n\epsilon_n \nonumber\\
&\leq \sum\limits_{i = 1}^n \mi(M,S^{i - 1},Z^{i - 1};Y_i)  + n\epsilon_n \nonumber\\
&\mathop \le \limits^{(c)} \sum\limits_{i = 1}^n \mi(X_{i};Y_i)  + n\epsilon_n \label{eq:Upper_Bound_Strictly_Causal_CSIT_1}\\
& =  n\sum\limits_{i = 1}^n \frac{1}{n}\mi(X_i;Y_i)  + n\epsilon_n \nonumber\\
& =  n\sum\limits_{i = 1}^n \Prob(T=i)\mi(X_i;Y_i|T=i)  + n\epsilon_n \nonumber\\
& =  n\mi(X_T;Y_T|T)  + n\epsilon_n \nonumber\\
& \leq  n\mi(X_T,T;Y_T)  + n\epsilon_n \nonumber\\
&\mathop = \limits^{(d)}  n\mi(X;Y)  + n\epsilon_n\nonumber\\
&\mathop \leq \limits^{(e)}  n\mi(X;Y)  + n\epsilon,
\label{eq:Upper_Bound_Strictly_Causal_CSIT_2}
\end{align}where
\begin{itemize}
    \item[$(a)$] follows from Fano's inequality;
    \item[$(b)$] follows from the Markov chain $(M,Y^{i-1})-(M,S^{i-1})-Y_i$;
    \item[$(c)$] follows since $S_i$ is independent of $\big(M,S^{i-1},Z^{i-1},X_i(M,S^{i-1})\big)$ and therefore \begin{equation}
    \mi(M,S^{i-1},Z^{i-1};Y_i|X_i)\leq\mi(M,S^{i-1},Z^{i-1};Y_i|X_i,S_i)=0,
    \end{equation}
    that is $(M,S^{i-1},Z^{i-1})-X_i-Y_i$ forms a Markov chain, which implies $V_i-X_i-Y_i$, where $V_i\triangleq(M,Z^{i-1})$, forms a Markov chain;
    \item[$(d)$] follows by defining $X=(X_T,T)$ and $Y=Y_T$.
    \item[$(e)$] follows by defining $\epsilon\triangleq\max\{\epsilon_n,\nu\geq\frac{\delta}{n}\}$.
\end{itemize}

Next, we lower bound $nR$ as follows,
\begin{align}
nR &\geq \ent(M)\nonumber\\
&\ge \mi(M;Z^n)\nonumber\\
&= \sum\limits_{i = 1}^n\mi(M;Z_i|Z^{i-1})\nonumber\\
&\mathop \ge \limits^{(a)} \sum\limits_{i = 1}^n\mi(M,Z^{i-1};Z_i)-\delta\nonumber\\
&\mathop = \limits^{(b)} \sum\limits_{i = 1}^n \big[\mi(V_i;Z_i)\big] - \delta \nonumber\\
&= n\mi(V_T;Z_T|T) - \delta \nonumber\\
&\mathop \ge \limits^{(c)} n\mi(V_T,T;Z_T) - 2\delta \nonumber\\
&\mathop = \limits^{(d)} n\mi(V;Z) - 2\delta
\label{eq:Upper_Bound_on_KRate_FCSI_SC}
\end{align}where 
\begin{itemize}
    \item[$(a)$] follows from Lemma~\ref{lemma:Resolvability_Properties};
    \item[$(b)$] follows by defining $V_i\triangleq(M,Z^{i-1})$;
    \item[$(c)$] follows from Lemma~\ref{lemma:Resolvability_Properties};
    \item[$(d)$] follows by defining $V=(V_T,T)$ and $Z=Z_T$.
\end{itemize}
For any $\nu>0$, choosing $n$ large enough ensures that
\begin{align}
 R &\mathop \geq \mi(V;Z) - 2\nu,\nonumber\\
 &\mathop \geq \mi(V;Z) - 2\epsilon,
 \label{eq:Upper_Bound_on_KRate_CSI_SC}
\end{align}where the last inequality follows from definition $\epsilon\triangleq\max\{\epsilon_n,\nu\}$. 
To show that $ \kld(P_Z||Q_0)\leq\epsilon$, note that for $n$ large enough
\begin{align}
\kld(P_Z||Q_0)=\kld(P_{Z_T}||Q_0)=\kld\Bigg(\frac{1}{n}\sum\limits_{i=1}^nP_{Z_i}\Bigg|\Bigg|Q_0\Bigg)\leq\frac{1}{n}\sum\limits_{i=1}^n\kld(P_{Z_i}||Q_0)\leq\frac{1}{n}\kld(P_{Z^n}||Q_0^{\otimes n})\leq\frac{\delta}{n}\leq\nu\leq\epsilon.\label{eq:boundKL_CSIT_SC}
\end{align}
Combining \eqref{eq:Upper_Bound_Strictly_Causal_CSIT_2} and \eqref{eq:Upper_Bound_on_KRate_CSI_SC}, and~\eqref{eq:boundKL_CSIT_SC} shows that $\forall \epsilon_n,\nu>0$, $R\leq \max\{x:x\in\mathcal{S}_{\epsilon}\}$. Therefore,
\begin{align}
  R\leq \max\left\{x:x\in\bigcap_{\epsilon>0}\mathcal{S}_{\epsilon}\right\}.
\end{align}
\subsubsection{Continuity at Zero}
One can prove the continuity at zero of $\calS_\epsilon$ by substituting $\min\{\mi(U;Y)-\mi(U;S),\mi(U,V;Y)-\mi(U;S|V)\}$ with $\mi(X;Y)$ and $\mi(V;Z)-\mi(V;S)$ with $\mi(V;Z)$ in Appendix~\ref{sec:continuity-at-zero} and following the exact same arguments.

\end{appendices}

\bibliographystyle{IEEEtran}
\bibliography{IEEEabrv,mybibfile}

\end{document}

%% file: CommandsAndMacros.tex
\newcommand{\calD}{\mathcal{D}}

\newcommand{\bbR}{\mathbb{R}}

\newcommand{\calS}{\mathcal{S}}

\newcommand{\calU}{\mathcal{U}}

\newcommand{\calV}{\mathcal{V}}

\newcommand{\calX}{\mathcal{X}}

\newcommand{\calZ}{\mathcal{Z}}


%% file: Acronyms.tex
\acrodef{ACDIS}[ACDIS]{Adaptive Communication Decision and Information Systems}
\acrodef{AEP}{Asymptotic Equipartition Property}
\acrodef{AoA}{Angle of Arrival}
\acrodef{AWGN}{Additive White Gaussian Noise}
\acrodef{AVC}[AVC]{Arbitrarily Varying Channel}
\acrodef{BER}{Bit-Error-Rate}
\acrodef{BEC}{Binary Erasure Channel}
\acrodef{BPSK}{Binary Phase-Shift Keying}
\acrodef{BSC}{Binary Symmetric Channel}
\acrodef{BICM}[BICM]{Bit-Interleaved Coded-Modulation}
\acrodef{CDF}[CDF]{Cumulative Distribution Function}
\acrodef{CGF}[CGF]{Cumulant Generating Function}
\acrodef{CLT}[CLT]{Central Limit Theorem}
\acrodef{CSI}[CSI]{Channel State Information}
\acrodef{DMC}[DMC]{Discrete Memoryless Channel}
\acrodef{DMS}[DMS]{Discrete Memoryless Source}
\acrodef{ERM}[ERM]{Empirical Risk Minimization}
\acrodef{FER}[FER]{Frame Error Rate}
\acrodef{ICA}[ICA]{Independent Component Analysis}
\acrodef{iid}[i.i.d.]{independent and identically distributed}
\acrodef{IoT}[IoT]{Internet of Things}
\acrodef{KKT}[KKT]{Karush-Kuhn Tucker}
\acrodef{LASSO}[LASSO]{Least Absolute Shrinkage and Selection Operator}
\acrodef{LPD}[LPD]{Low Probability of Detection}
\acrodef{LDPC}[LDPC]{Low-Density Parity-Check}
\acrodef{LLMS}[LLMS]{Linear Least Mean Square}
\acrodef{LMS}[LMS]{Least Mean Square}
\acrodef{MAC}[MAC]{multiple-access channel}
\acrodef{MGF}[MGF]{Moment Generating Function}
\acrodef{MLC}[MLC]{Multi-Level Coding}
\acrodef{MLE}[MLE]{Maximum Likelihood Estimate}
\acrodef{MIMO}[MIMO]{Multiple-Input Multiple-Output}
\acrodef{MISO}{Multiple-Input Single-Output}
\acrodef{MSD}[MSD]{Multi-Stage Decoding}
\acrodef{MMSE}[MMSE]{Minimum Mean-Square Error}
\acrodef{PAC}[PAC]{Probably Approximately Correct}
\acrodef{PCA}[PCA]{Principal Component Analysis}
\acrodef{PDF}[PDF]{Probability Density Function}
\acrodef{PMF}[PMF]{Probability Mass Function}
\acrodef{PPM}[PPM]{Pulse Position Modulation}
\acrodef{PSD}{Power Spectral Density}
\acrodef{PSK}{Phase Shift Keying}
\acrodef{QKD}{Quantum Key Distribution}
\acrodef{ROC}{Receiver Operating Characteristic}
\acrodef{CVQKD}{Continuous-Variable \ac{QKD}}
\acrodef{QPSK}{Quadrature Phase-Shift Keying}
\acrodef{RV}{random variable}
\acrodef{SIMO}{Single-Input Multiple-Output}
\acrodef{SNR}{Signal-to-Noise Ratio}
\acrodef{SVM}[SVM]{Support Vector Machine}
\acrodef{TPCP}{Trace-Preserving Completely-Positive}
\acrodef{wrt}[w.r.t.]{with respect to}
\acrodef{WSS}{Wide Sense Stationary}
\acrodef{RHS}{Right Hand Side}

%% file: Revision1-arXiv-CovertCSI.bbl
\begin{thebibliography}{10}
\providecommand{\url}[1]{#1}
\csname url@samestyle\endcsname
\providecommand{\newblock}{\relax}
\providecommand{\bibinfo}[2]{#2}
\providecommand{\BIBentrySTDinterwordspacing}{\spaceskip=0pt\relax}
\providecommand{\BIBentryALTinterwordstretchfactor}{4}
\providecommand{\BIBentryALTinterwordspacing}{\spaceskip=\fontdimen2\font plus
\BIBentryALTinterwordstretchfactor\fontdimen3\font minus
  \fontdimen4\font\relax}
\providecommand{\BIBforeignlanguage}[2]{{%
\expandafter\ifx\csname l@#1\endcsname\relax
\typeout{** WARNING: IEEEtran.bst: No hyphenation pattern has been}%
\typeout{** loaded for the language `#1'. Using the pattern for}%
\typeout{** the default language instead.}%
\else
\language=\csname l@#1\endcsname
\fi
#2}}
\providecommand{\BIBdecl}{\relax}
\BIBdecl

\bibitem{LPD_on_AWGN}
A.~B. Bash, D.~Goeckel, and D.~Towsley, ``Limits of reliable communication with
  low probability of detection on {AWGN} channels,'' \emph{{IEEE} J. Sel. Areas
  Commun.}, vol.~31, no.~9, pp. 1921--1930, Sep. 2013.

\bibitem{Reliable_Deniable_Comm}
P.~H. Che, M.~Bakshi, and S.~Jaggi, ``Reliable deniable communication: Hiding
  messages in noise,'' in \emph{Proc. {IEEE} Int. Symp. on Info. Theory
  (ISIT)}, Istanbul, Turkey, Jul. 2013, pp. 2945--2949.

\bibitem{LPD_over_DMC}
L.~Wang, G.~W. Wornell, and L.~Zheng, ``Fundamental limits of communication
  with low probability of detection,'' \emph{{IEEE} Trans. Inf. Theory},
  vol.~62, no.~6, pp. 3493--3503, Jun. 2016.

\bibitem{LPD_by_Resolvability}
M.~R. Bloch, ``Covert communication over noisy channels: A resolvability
  perspective,'' \emph{{IEEE} Trans. Inf. Theory}, vol.~62, no.~5, pp.
  2334--2354, May 2016.

\bibitem{OrderAsymtMehrdad}
M.~Tahmasbi and M.~R. Bloch, ``First- and second-order asymptotics in covert
  communication,'' \emph{{IEEE} Trans. Inf. Theory}, vol.~65, no.~4, pp.
  2190--2212, Apr. 2019.

\bibitem{Covert_With_State}
S.-H. Lee, L.~Wang, A.~Khisti, and G.~W. Wornell, ``Covert communication with
  channel-state information at the transmitter,'' \emph{{IEEE} Trans. Inf.
  Forensics Security}, vol.~13, no.~9, pp. 2310--2319, Sep. 2018.

\bibitem{Covert_Saleh}
S.~Salehkalaibar, M.~H. Yassaee, and V.~Y.~F. Tan, ``Covert communication over
  a compound channel,'' \emph{available at
  \url{https://arxiv.org/abs/1906.06675}}, Jun. 2019.

\bibitem{Stealthy_SKG}
P.-H. Lin, C.~R. Janda, and E.~A. Jorswieck, ``Stealthy secret key
  generation,'' in \emph{Proc. {IEEE} Global Conf. on Signal and Info.
  Processing (GlobalSIP)}, Montreal, QC, Canada, Mar. 2017, pp. 492--496.

\bibitem{Stealthy_Keyless_SKG}
P.-H. Lin, C.~R. Janda, E.~A. Jorswieck, and R.~F. Schaefer, ``Stealthy keyless
  secret key generation from degraded sources,'' in \emph{51st Asilomar
  Conference on Signals, Systems, and Computers}, Pacific Grove, CA, USA, Apr.
  2018, pp. 14--18.

\bibitem{Covert_SKG}
M.~Tahmasbi and M.~R. Bloch, ``Covert secret key generation,'' in \emph{Proc.
  {IEEE} Conf. on Commun. and Network Security (CNS)}, Las Vegas, NV, USA, Dec.
  2017, pp. 540--544.

\bibitem{Covert_SKG_Quantom}
------, ``Framework for covert and secret key expansion over classical-quantum
  channels,'' \emph{Phys. Rev. A}, vol.~99, p. 052329, May 2019.

\bibitem{WTC_With_State}
Y.~Chen and A.~J. Han~Vinck, ``Wiretap channel with side information,''
  \emph{{IEEE} Trans. Inf. Theory}, vol.~54, no.~1, pp. 395--402, Jan. 2008.

\bibitem{SKA_with_State}
A.~Khisti, S.~N. Diggavi, and G.~W. Wornell, ``Secret key agreement using
  asymmetry in channel state knowledge,'' in \emph{Proc. {IEEE} Int. Symp. on
  Info. Theory (ISIT)}, Seoul, South Korea, Jul. 2009, pp. 2286--2290.

\bibitem{Chia_WTC_with_State}
Y.-K. Chia and A.~El~Gamal, ``Wiretap channel with causal state information,''
  \emph{{IEEE} Trans. Inf. Theory}, vol.~58, no.~5, pp. 2838--2849, May 2012.

\bibitem{On_Sec_Capa_WTC_CSI}
H.~Fujita, ``On the secrecy capacity of wiretap channels with side information
  at the transmitter,'' \emph{{IEEE} Trans. Inf. Forensics Security}, vol.~11,
  no.~11, pp. 2441--2452, Nov. 2016.

\bibitem{WTC_with_SC_State}
A.~Sonee and G.~A. Hodtani, ``Wiretap channel with strictly causal side
  information at encoder,'' in \emph{Proc. Iran Workshop on commun. and Info.
  Theory (IWCIT)}, Tehran, Iran, May 2014, pp. 1--6.

\bibitem{Han_WTC_with_CSI}
T.~S. Han and M.~Sasaki, ``Wiretap channels with causal state information:
  Strong secrecy,'' \emph{{IEEE} Trans. Inf. Theory}, vol.~65, no.~10, pp.
  6750--6765, Oct. 2019.

\bibitem{Ziv_WTC_with_NonCausaul_CSI}
Z.~Goldfeld, P.~Cuff, and H.~H. Permuter, ``Wiretap channels with causal state
  information: Strong secrecy,'' \emph{{IEEE} Trans. Inf. Theory}, vol.~66,
  no.~3, pp. 1497--1519, Mar. 2020.

\bibitem{ITWPaper}
H.~ZivariFard, M.~R. Bloch, and A.~Nosratinia, ``Keyless covert communication
  in the presence of non-causal channel state information,'' in \emph{Proc.
  {IEEE} Info. Theory Workshop (ITW)}, Visby, Sweden, Aug. 2019, pp. 1--5.

\bibitem{ISIT2020}
------, ``Keyless covert communication in the presence of channel state
  information,'' in \emph{{IEEE} Int. Symp. on Info. Theory (ISIT)}, Los
  Angeles, CA, USA, Jun. 2020, pp. 834--839.

\bibitem{KeylessCovertArxiv}
------, ``Keyless covert communication via channel state information,''
  \emph{available at \url{https://arxiv.org/abs/2003.03308}}, Mar. 2020.

\bibitem{Hypotheses_Book}
E.~L. Lehmann and J.~P. Romano, \emph{Testing Statistical Hypotheses}.\hskip
  1em plus 0.5em minus 0.4em\relax New York, NY, USA: Springer-Verlag, 2005.

\bibitem{FDivergence}
I.~Sason and S.~Verd\'u, ``$f$-divergence inequalities,'' \emph{{IEEE} Trans.
  Inf. Theory}, vol.~62, no.~11, pp. 5973--6006, Nov. 2016.

\bibitem{GelfandPinsker}
S.~I. Gel'fand and M.~S. Pinsker, ``Coding for channel with random
  parameters,'' \emph{Problem Control Inf. Theory}, vol.~9, no.~1, pp. 19--31,
  Jan. 1980.

\bibitem{WynerZiv76}
A.~D. Wyner and J.~Ziv, ``The rate-distortion function for source coding with
  side information at the decoder,'' \emph{{IEEE} Trans. Inf. Theory}, vol.~22,
  no.~1, pp. 1--10, Jan. 1976.

\bibitem{Cuff13}
P.~Cuff, ``Distributed channel synthesis,'' \emph{{IEEE} Trans. Inf. Theory},
  vol.~59, no.~11, pp. 7071--7096, Nov. 2013.

\bibitem{Yassaee13}
M.~H. Yassaee, M.~R. Aref, and A.~A. Gohari, ``A technique for deriving
  one-shot achievability results in network information theory,'' in
  \emph{Proc. {IEEE} Int. Symp. on Info. Theory (ISIT)}, Istanbul, Turkey, Jul.
  2013, pp. 1287--1291.

\bibitem{Watanabe15}
S.~Watanabe, S.~Kuzuoka, and V.~Y.~F. Tan, ``Nonasymptotic and second-order
  achievability bounds for coding with side-information,'' \emph{{IEEE} Trans.
  Inf. Theory}, vol.~61, no.~4, pp. 1574--1605, Apr. 2015.

\bibitem{Lapidoth2010}
A.~Lapidoth and Y.~Steinberg, ``The multiple access channel with causal and
  strictly causal side information at the encoders,'' in \emph{Int. Zurich
  Semin. Commun.}, Zurich, Switzerland, Mar. 2010, pp. 13--16.

\bibitem{Eggleston}
H.~G. Eggleston, \emph{Convexity}, 6th~ed.\hskip 1em plus 0.5em minus
  0.4em\relax Cambridge, U.K: Cambridge University Press, 1958.

\bibitem{ElGamalKim}
A.~El~Gamal and Y.-H. Kim, \emph{Network Information Theory}, 1st~ed.\hskip 1em
  plus 0.5em minus 0.4em\relax Cambridge, U.K: Cambridge University Press,
  2012.

\bibitem{Shannon58state}
C.~E. Shannon, ``Channels with side information at the transmitter,'' \emph{IBM
  J. Res. Develop.}, vol.~2, no.~4, pp. 289--293, Oct. 1958.

\bibitem{YassaeeMAWC}
M.~H. Yassaee and M.~R. Aref, ``Multiple access wiretap channels with strong
  secrecy,'' in \emph{Proc. {IEEE} Info. Theory Workshop (ITW)}, Dublin,
  Ireland, Sep. 2010, pp. 1--5.

\bibitem{AllertonCuffSong}
E.~C. Song, P.~Cuff, and H.~V. Poor, ``A rate-distortion based secrecy system
  with side information at the decoders,'' in \emph{Proc. 52th Annual Allerton
  Conference on Communication, Control, and Computing}, Monticello, IL, Sep.
  2014, pp. 755--762.

\bibitem{BCC:IT78}
I.~Csisz\'ar and J.~K\"{o}rner, ``Broadcast channels with confidential
  messages,'' \emph{{IEEE} Trans. Inf. Theory}, vol.~24, no.~3, pp. 339--348,
  May 1978.

\bibitem{Hou13}
J.~Hou and G.~Kramer, ``Informational divergence approximations to product
  distributions,'' in \emph{Proc. 13th Can. Workshop Inf. Theory (CWIT)},
  Toronto, ON, Canada, Jun. 2013, pp. 76--81.

\bibitem{Ziv_Strong_Sec_BC_Cooperation}
Z.~Goldfeld, G.~Kramer, H.~H. Permuter, and P.~Cuff, ``Strong secrecy for
  cooperative broadcast channels,'' \emph{{IEEE} Trans. Inf. Theory}, vol.~63,
  no.~19, pp. 469--495, Jan. 2017.

\end{thebibliography}
